\documentclass[11pt,a4paper]{article}
\pdfoutput=1
\usepackage[left=1in,right=1in,top=1in,bottom=1in]{geometry}
\usepackage{graphicx}
\usepackage{float}
\usepackage{subcaption}
\usepackage[tbtags]{amsmath}
\usepackage{amssymb}
\usepackage{amsthm}
\usepackage{lmodern}
\usepackage{hyperref}
\usepackage[capitalise,nameinlink]{cleveref}
\hypersetup{colorlinks={true},linkcolor={blue},citecolor=red}
\usepackage[numbers]{natbib}
\usepackage{mdframed}
\mdfsetup{backgroundcolor=white,nobreak=true}
\newtheorem{theorem}{Theorem}

\newtheorem{lemma}{Lemma}
\newtheorem*{lemmanonum}{Lemma}
\newtheorem*{theoremnonum}{Theorem}
\newtheorem*{conjecturenonum}{Conjecture}
\theoremstyle{definition}
\newtheorem{definition}{Definition}
\theoremstyle{remark}
\newtheorem{remark}{Remark}
\def \R{\mathbb R}
\newcommand{\sset}[1]{\left\{ #1\right\}}
\newcommand{\fwh}[1]{\; \left| \; #1 \right.}
\newcommand{\card}[1]{\left| #1 \right|}
\newcommand{\probabilityext}[2]{\ensuremath{\mathrm{Pr}_{#1}\left[#2\right]}}
\newcommand{\rev}{\ensuremath{\text{\rm\sc Rev}}}
\newcommand{\srev}{\ensuremath{\text{\rm\sc SRev}}}
\newcommand{\union}{\cup}
\newcommand{\bigunion}{\bigcup}     
\newcommand{\map}{\longrightarrow}
\newcommand{\ifif}{\Longleftrightarrow} 
\newcommand{\inters}{\cap}     
\newcommand{\then}{\Longrightarrow} 
\DeclareMathOperator*{\expect}{\mathbb E}
\newcommand{\vecc}[1]{\ensuremath{\mathbf{#1}}}
\DeclareMathOperator{\chull}{\mathcal H}
\DeclareMathOperator{\permuts}{\mathcal P}
\DeclareMathOperator{\dclosure}{\mathcal D}
\newcommand{\biglor}{\bigvee}
\newcommand{\bigland}{\bigwedge}
\newcommand{\ssets}[1]{\{ #1\}}
\newcommand{\cards}[1]{|#1|}  

\newcommand{\uniformunit}{\ensuremath{\mathcal{U}}}
\newcommand{\discube}[1]{\ensuremath{\Delta(#1)}}
\newcommand{\slice}[3]{\left.#1\right|_{{#2}:#3}}  
\newcommand{\slicesmall}[3]{{#1}|_{{#2}:#3}}  

\title{Duality and Optimality of Auctions for Uniform Distributions\footnote{The research leading to these results has received funding from the European Research Council under the European Union's Seventh Framework Programme (FP7/2007-2013) / ERC grant agreement 321171. 
\newline \indent
A preliminary version of this paper appeared in~\cite{gk2014-ec}.}
}
\author{Yiannis Giannakopoulos\thanks{Department of Computer Science, University of Oxford. Email: \href{mailto:ygiannak@cs.ox.ac.uk}{\nolinkurl{ygiannak@cs.ox.ac.uk} }} \and Elias Koutsoupias\thanks{Department of Computer Science, University of Oxford. Email: \href{mailto:elias@cs.ox.ac.uk}{\nolinkurl{elias@cs.ox.ac.uk} }}}
\date{January 15, 2018}
\begin{document}
\maketitle

\begin{abstract}
We develop a general duality-theory framework for revenue maximization in additive Bayesian auctions.
The framework extends linear programming duality and complementarity to constraints with partial derivatives. The dual system reveals the geometric nature of the problem and highlights its connection with the theory of bipartite graph matchings.
We demonstrate the power of the framework by applying it to a multiple-good monopoly setting where the buyer has uniformly distributed valuations for the items, the canonical long-standing open problem in the area.
We propose a deterministic selling mechanism called \emph{Straight-Jacket Auction (SJA)}, which we prove to be \emph{exactly} optimal for up to 6 items, and conjecture its optimality for any number of goods.
The duality framework is used not only for proving optimality, but perhaps more importantly for deriving the optimal mechanism itself; as a result, SJA is defined by natural geometric constraints.
\end{abstract}

\section{Introduction}

The problem of maximizing revenue in multidimensional Bayesian auctions is one of the most prominent within the area of Mechanism Design. An auctioneer wants to sell a number of items to some potential buyers (bidders). Each bidder has a value for every item; this is the maximum price that she is willing to pay to get the item and it is a private information. The value of a set of items is simply the sum of the values of the items in the set (\emph{additive valuations}). The buyers submit their bids and the auctioneer must decide, perhaps with randomization, what items to allocate to each player and how much to charge each one of them for this transaction. The seller has some prior (incomplete) knowledge about how much each player values the items, captured by a (joint) probability distribution over the space of all possible valuations. However, 
assuming standard selfish game-theoretic behavior, the players would lie about their true values and submit false bids if this is to increase their personal gain. The goal is to design auction protocols that maximize the total expected revenue of the seller, by also ensuring the truthful participation of the bidders.

For the single-dimensional case where only one item is to be auctioned among the players, the seminal work of \citet{Myerson:1981aa} has completely settled the problem. His solution is simple and elegant: the optimal auction is deterministic and easy to describe by a ``virtual valuations'' transformation and reduction to a social welfare maximization problem which can be solved using the well-studied VCG auction~\citep{Nisan:2007zr,Hartline:2007aa}.

Unfortunately, for the many-items setting 
these elegant properties and results do not hold in general. It is very likely that there is no simple closed-form description of optimal revenue auctions, especially in a unified way similar to Myerson's solution. However, for the most commonly studied probability distributions, e.g., the uniform and normal, we would like to have such clear, closed-form descriptions of the optimal auctions, or at least algorithms---preferably simple and intuitive---that compute optimal auctions (their allocation and payment functions). \emph{But we are far from such a goal. There exists no interesting continuous probability distribution for which we know the optimal auction for more than three items.} The difficulty of the problem is illustrated by the lack of general results for the canonical case of uniform i.i.d.~valuations in the unit interval $[0,1]$ even for a single bidder. In this work, we resolve this case for up to 6 items. We give an exact, analytic and intuitive way of computing the optimal prices; the solution is in closed-form, but involves roots of polynomials of degree equal to the number of items.
We do that as a special application of a much more general construction:
a duality-theory framework for proving exact and approximate optimality of many-bidder multi-item auctions for arbitrary continuous distributions. We expect this framework to be essential for helping generalizing Myerson's solution to many-items settings, the holy grail of auction theory.

It is known that even in the simple case of one bidder, randomized auctions can perform strictly better than deterministic ones~\citep{Hart:2012zr,Hart:2012ys,Manelli:2006vn,Pycia:2006pd,Daskalakis:2013vn}. \citet{Manelli:2006vn} provide some sufficient conditions for deterministic auctions to be optimal, 
but these are quite involved, in the form of functional inequalities that incorporate abstract partitions of the valuation space, and admittedly difficult to interpret. They were able to instantiate them though for the case of two and three uniform i.i.d.\ distributions and completely determine an optimal deterministic auction.
For more items, it is not known whether the optimal auction is a deterministic one. Our results here show that the optimal auction for up to 6 items is indeed deterministic. \emph{We conjecture that this is true for any number of items}; we also conjecture that for more than one bidder the optimal auction is not deterministic. \citet{Hart:2012uq} have provided a very simple sufficient condition in the case of two i.i.d.~items for the deterministic full-bundle auction to be optimal and deploy it to show that this is the case for the equal-revenue distribution. Finally, \citet{Daskalakis:2013vn} were also able to deal with the special case of two items and independent (not necessarily identical) exponential distributions and give an exact solution, which in this case is randomized. 
Essentially this is all that was known prior to our work regarding exact descriptions of optimal auctions with continuous probability distributions,

Given the difficulty of designing optimal auctions, \citet{Hart:2012uq} study the performance of the two most straightforward deterministic mechanisms for the single-buyer setting: the one that sells all items in a full bundle and the one that sells each item independently. They provide elegant approximation ratio guarantees (logarithmic with respect to the number of items) that hold universally for all product (independent) distributions, without even assuming standard regularity conditions (as, e.g., in~\citep{Chawla:2007aa,Manelli:2006vn,Myerson:1981aa}). \citet{Li:2013ty} further improved their results.
The difficulty of providing \emph{exact} optimal solutions for multi-item settings is further supported by a recent computational hardness result by \citet{Daskalakis:2012fk}, where it is shown that even for a single buyer and independent (but not identical) valuations with finite support of size $2$, it is {\#}P-hard to compute the allocation function of an optimal auction. However this does not exclude the possibility of efficiently computing approximate solutions. In fact, \citet{Cai:2013kx} and \citet{Daskalakis:2012aa} have presented PTAS (polynomial-time approximation schemes) for i.i.d.~settings.

\citet*{Daskalakis:2013vn} have also published a duality approach to the problem, inspired by optimal transport theory. With its use, they gave optimal mechanisms for two-item settings for exponential distributions. Their approach assumes independent item distributions that either have unbounded interval supports and decrease more steeply than $1/x^2$ or bounded ones but they vanish to zero at the right bound of the interval. Thus their method cannot be directly applied to uniform valuations. Our aim is to provide a duality theory framework for multi-item optimal auctions, which is as general and clean as possible for many bidders and arbitrary joint distributions (not necessarily independent ones). For that reason, we deploy a ``proof-from-scratch'' approach directly inspired by linear programming duality which is easily comprehensible and applicable, and with which the reader will immediately feel familiar. At their core, the two duality frameworks are based on similar ideas; although we expect our framework to have wider applicability, we also believe that there will be special cases in which the framework of \citep{Daskalakis:2013vn} will be more suitable to apply.

Finally, we mention some very recent developments after the initial conference version~\citep{gk2014-ec} of our paper: in a ground-breaking work \citet{Babaioff:2014ys} showed that a constant approximation of the optimal revenue can be achieved for the case of one bidder and independent items by very simple deterministic mechanisms, using a core-tail decomposition technique inspired by~\cite{Li:2013ty}, and \citet{Yao:2014vn} later generalized this idea to many-player settings. 
\subsection{Model and Notation}\label{sec:model}

We denote the real unit interval by $I=[0,1]$, the nonnegative reals by $\R_+=[0,\infty)$. We consider auctions of $n$ bidders who are interested in buying any subset of $m$ items. For any positive integer $m$ we use the notation $[m]=\sset{1,2,\dots,m}$. The value of bidder $i$ for item $j$ is in interval $D_{i,j}=[L_{i,j}, H_{i,j}]\subseteq\R_+$; we denote by $D_i=\prod_{j=1}^mD_{i,j}$ the hyperrectangle of all possible values of bidder $i$, and by $D=\prod_{i=1}^nD_i$ the space of all valuation inputs to the mechanism. The seller knows some probability distribution over $D$ with an almost everywhere\footnote{With respect to the standard Lebesgue measure $\mu$ in $\R^{n\times m}$.} (a.e.) differentiable density function $f$. Intervals $D_{i,j}$ need not be bounded; that is, we allow $H_{i,j}\in\R\union\sset{\infty}$.

Let $\vecc 0_m=(0,0,\dots,0)$ and $\vecc 1_m=(1,1,\dots,1)$ denote the
$m$-dimensional zero and unit vectors, respectively. We will drop subscript $m$ whenever this causes no confusion. For two $m$-dimensional vectors $\vecc x=(x_1,x_2,\dots,x_m)$ and $\vecc y=(y_1,y_2,\dots,y_m)$ we write $\vecc x\leq \vecc y$ as a shortcut for $x_j\leq y_j$ for all $j\in[m]$. 
For any matrix $\vecc x\in \R^{n\times m}$, $\vecc x_i$ will denote its $i$-th ($m$-dimensional) row vector.
For a function $f:\R^{n\times m}\to\R$ and $i\in[n]$ we denote $\nabla_i f(\vecc x)\equiv(\frac{\partial f(\vecc x)}{\partial x_{i,1}}, \frac{\partial f(\vecc x)}{\partial x_{i,2}}, \ldots, \frac{\partial f(\vecc x)}{\partial x_{i,m}})$; notice how only the derivatives with respect to the variables in row $\vecc x_i$ appear.
Finally, we use the standard game theoretic notation $\vecc x_{-j}=(x_1,x_2\dots,x_{j-1},x_{j+1},\dots,x_m)$ to denote the resulting vector if we remove $\vecc x$'s $j$-th coordinate. Then, $\vecc x=(\vecc x_{-j},x_j)$. Similarly, $\vecc x_{-(i,j)}$ will denote all values of the $n\times m$ matrix $\vecc x$ when we remove the $(i,j)$-th entry.
For a large part of the paper we will restrict our attention to a single bidder. In this case, we drop the subscript $i$ completely; for example, we write $L_j$ instead of $L_{1,j}$.

\subsubsection{Mechanisms and Truthfulness}

In this paper we study auctions for selling $m$ items to $n$ bidders when bidder $i\in[n]$ has a nonnegative valuation $x_{i,j}\in D_{i,j}$ for item $j\in[m]$. This is private information of the bidder, and intuitively represents the amount of money she is willing to pay to get this item. The seller has only some incomplete prior knowledge of the valuations $\vecc x$ in the form of a joint probability distribution $F$ over $D$ from which $\vecc x$ is drawn.

A direct revelation mechanism (auction) $\mathcal M=(\vecc a, \vecc p)$ on this setting is a protocol which, after receiving a bid vector $\vecc x_i'$ from each bidder $i$ as input (the bidder may lie about her true valuations $\vecc x_i$ and misreport $\vecc x_i'\neq\vecc x_i$), offers item $j$ to bidder $i$ with probability $a_{i,j}(\vecc x')\in[0,1]$, and bidder $i$ pays $p_i(\vecc x')\in\R$. We assume that each item can only be sold to at most one bidder, or equivalently $\sum_i a_{i,j}(\vecc x') \leq 1$. The total revenue extracted from the auction is $\sum_i p_i(\vecc x')$. If we want to restrict our attention only to deterministic auctions, we take $a_{i,j}(\vecc x')\in\ssets{0,1}$. Notice also that we do not demand nonnegative payments $\vecc p\geq \vecc 0$, i.e., we don't assume what is known as the No Positive Transfers (NPT) condition, since that is not explicitly needed for our results. However, as argued, e.g., in~\citep[Sect.~2.1]{Hart:2012uq}, assuming such a condition would be without loss of generality for the revenue maximization problem. 

More formally, a mechanism consists of an \emph{allocation} function $\vecc a\, : \, D\map I^{n\times m}$, which satisfies $\sum_i a_{i,j}(\vecc x) \leq 1$ for all $\vecc x\in D$ and all items $j\in [m]$, paired with \emph{payment} functions $p_i\, : \, D\map\R$. We consider each bidder having \emph{additive} valuations for the items, her ``happiness'' when she has (true) valuations $\vecc x_i$ and players report $\vecc x'=(\vecc x_{-i}',\vecc x_i')$  to the mechanism being captured by her \emph{utility} function 
\begin{equation}\label{eq:utilitydef}
u_i(\vecc x'|\vecc x_i)\equiv\vecc a_i(\vecc x')\cdot\vecc x_i-p_i(\vecc x')=\sum_{j=1}^ma_{i,j}(\vecc x')x_{i,j}-p_i(\vecc x'), 
\end{equation}
the expected sum of the valuations she receives from the items she manages to purchase minus the payment she has to submit to the seller for this purchase. The player is completely rational and selfish, wanting to maximize her utility, and that's why she will not hesitate to misreport $\vecc x_i'$ instead of her private values $\vecc x_i$ if this is to give her a higher utility in~\eqref{eq:utilitydef}. On the other hand, the seller's happiness is captured by the total \emph{revenue} of the mechanism 
\begin{equation}\label{eq:revenue}
\sum_{i=1}^n p_i(\vecc x')=\sum_{i=1}^n\left( \vecc a_i(\vecc x')\cdot \vecc x_i-u_i(\vecc x'|\vecc x_i)\right),
\end{equation}
which is a simple rearrangement of~\eqref{eq:utilitydef}.

It is standard in Mechanism Design to ask for auctions to respect the following two properties, for any player $i\in[n]$:
\begin{itemize}
\item \emph{Individual Rationality (IR)}: $u_i(\vecc x|\vecc x_i)\geq 0$ for all $\vecc x\in D$
\item \emph{Incentive Compatibility (IC)}: $u_i(\vecc x|\vecc x_i)\geq u_i((\vecc x_{-i},\vecc x_i')|\vecc x_i)$ for all $\vecc x \in D$ and $\vecc x_i'\in D_i$
\end{itemize}
The IR constraint corresponds to the notion of voluntary participation, that is, a bidder cannot harm herself by truthfully taking part in the auction, while IC captures the fundamental property that truthtelling is a dominant strategy\footnote{In this work, we consider Dominant Strategy Incentive Compatibility (DSIC), the strongest notion of incentive compatibility in which the bidders know all values.} for the bidder in the underlying game, i.e.~she will never receive a better utility by lying. Auctions that satisfy IC are also called \emph{truthful}. From now on we will focus on truthful IR mechanisms, and so we will relax notation $u_i(\vecc x|\vecc x_i)$ to just $u_i(\vecc x)$, considering bidder's utility as a function $u_i:D\map\R_+$. 
The following is an elegant, extremely useful analytic characterization of truthful mechanisms due to~\citet{Rochet:1985aa}. For proofs of this we recommend~\citep{Hart:2012uq,Manelli:2007kx}.
\begin{theorem}\label{thm:truthfulconvex}
An auction $\mathcal M=(\vecc a,\vecc p)$ is truthful (IC) if and only if the utility functions $u_i$ that induces have the following properties with respect to the $i$-th row coordinates, for all bidders $i$:
\begin{enumerate}
\item $u_i(\vecc x_{-i},\cdot)$ is a \emph{convex} function
\item $u_i(\vecc x_{-i},\cdot)$ is almost everywhere (a.e.) differentiable with \label{cond:convexity}
\begin{equation}\label{eq:allocationgradient}
\frac{\partial u_i(\vecc x)}{\partial x_{i,j}}=a_{i,j}(\vecc x)\quad\text{for all items $j\in[m]$ and a.e.~$\vecc x\in D$}.
\end{equation} 
The allocation function $\vecc a_i$ is a \emph{subgradient} of $u_i$.
\end{enumerate}
\end{theorem}
\Cref{thm:truthfulconvex} essentially establishes a kind of correspondence between truthful mechanisms and utility functions. Not only does every auction induce well-defined utility functions for the bidders, but also conversely, given nonnegative convex functions that satisfy the properties of the theorem, we can fully recover a corresponding  mechanism from expressions~\eqref{eq:allocationgradient} and~\eqref{eq:revenue}.

\subsubsection{Optimal Auctions}\label{sec:optintro}
In this paper we study the problem of maximizing the seller's expected revenue based on his prior knowledge of the joint distribution $F$, under the IR and IC constraints, thus (by \cref{thm:truthfulconvex} and~\eqref{eq:revenue})
\begin{mdframed}
\begin{equation}\label{eq:totalrevenue}
\sup_{u_1,\dots,u_n}\sum_{i=1}^n \int_{D}\left(\nabla u_i(\vecc x)\cdot \vecc x_i - u_i(\vecc x)\right)\, d F(\vecc x)
\end{equation} 
over the space of \emph{nonnegative convex functions} $u_i$ on $D$ having the properties
\begin{align}
  \sum_{i=1}^n \frac{\partial u_i(\vecc x)}{\partial x_{i,j}} &\leq 1 \tag{\text{$z_j(\vecc x)$}} \\
  \frac{\partial u_i(\vecc x)}{\partial x_{i,j}} &\geq 0 \tag{\text{$s_{i,j}(\vecc x)$}}
\end{align}
for a.e.~$\vecc x\in D$, all $i\in [n]$ and $j \in [m]$.
\end{mdframed}

\subsubsection{Deterministic Auctions}
Given the characterization of \cref{thm:truthfulconvex}, in case one wants to focus on deterministic auctions then it is enough to consider only utility functions that are the maximum of affine hyperplanes with slopes either $0$ or $1$ with respect to any direction (see, e.g., \cite{Rockafellar:1997aa}). So, for example, any single-bidder ($n=1$) deterministic and symmetric\footnote{This means that the auction does not discriminate between items, i.e., any permutation of the valuations profile $\vecc x$ results to the same permutation of the output allocation vector $\vecc a(\vecc x)$.} auction corresponds to a utility function of the form
$$
u(\vecc x)=\max_{J\subseteq [m]}\left(\sum_{j\in J}x_i-p_{\card{J}}\right),
$$
where $p_{r}$ is the price offered to the buyer for any bundle of $r$ items, $r\in [m]$.

\section{Outline of Our Work}
\label{sec:outline-our-work}

We give here an outline of our work which bypasses many technical issues but brings out a few central ideas. The reader may also find it helpful to revisit this outline during the more technical exposition later on.

\subsection{Duality for a Single Bidder}
\label{sec:dualityoutline}

We first develop a general duality framework that applies to almost all interesting continuous probability distributions (\cref{sec:duality}). We view the problem of maximizing revenue as an optimization problem in which the unknowns are the utility functions $u_i(\vecc x)$ of the bidders (\hyperref[eq:totalrevenue]{Program~\eqref{eq:totalrevenue}}). There are two main restrictions imposed to these functions by truthfulness (see \cref{thm:truthfulconvex}): the convexity restriction (the utility function $u_i(\vecc x)$  must be convex with respect to the private values $\vecc x_i$ of bidder $i$) and the gradient restriction (the derivatives of this function 
must be nonnegative and they have to be at most $1$ for every item).

We simplify things by dropping the convexity constraint and keep only the gradient constraints. Surprisingly, the convexity constraint can be recovered for free from the optimal solution of the remaining constraints for a large class of distributions which includes the uniform distribution. 
We view the resulting formulation as an infinite linear program with variable the utility function of the bidder. Its essential constraints (labeled by $(z_j)$ in~\eqref{eq:totalrevenue}) are that the derivatives for each item must be at most $1$ and its objective is to maximize the expected value of $\sum_i\nabla_i u_i(\vecc x)\cdot x - u_i(\vecc x)$. We carefully rewrite the integral in~\eqref{eq:totalrevenue} to bring it into a form which does not include any derivatives. Remarkably, Myerson's solution for the special case of one item is based on a different rewriting of the system in which the primal variables are the derivatives of the utility, instead of the utility itself. In fact, since the allocation constraints involve exactly the derivatives, this is the most natural choice of primal variables. Unfortunately \emph{such an approach does not seem to work for the case of many items}, since the partial derivatives are not independent functions and, if we treat them as such, we run the risk of violating the gradient constraints.

Having rewritten the original system in terms of the utility functions $u_i(\vecc x)$, we define a proper dual system (\hyperref[eq:dual]{Program~\eqref{eq:dual}}) with variables functions $z_j(\vecc x)$, one for every item, and functions $s_{i,j}(\vecc  x)$, one for every pair of bidder and item. The dual constraints require that functions $z_j-s_{i,j}$ take small values at the lower boundary of the domain and high values at the upper boundary of the domain. Furthermore, the objective is to minimize the sum of $z_j's$ integrals (\cref{fig:dualprogram}). This would have been a trivial problem---for example, each $z_j$ could crawl at the minimum possible value until it reached the upper boundary and then shoot up to the required high value---had it not existed another constraint which requires that the sum of the derivatives of these functions is bounded above (and therefore the functions have to start rising sooner, to be able to reach the high value at the upper boundary).

Although the derivation of the dual program is natural and straightforward, there is no guarantee that the dual optimal value matches the primal optimal value, since these are infinite, indeed uncountable, linear programs. We directly show that the two systems satisfy the weak duality property (\cref{lemma:weakdualitymany}). This gives a general framework to prove optimality of a mechanism, by finding a dual solution and showing that their values match. Unfortunately, in most cases this is extremely hard, since the optimal value may be very complicated (for example, it turns out that the optimal value for the uniform distribution of $m$ items consists of algebraic numbers of degree $m$). Instead we prove a complementarity theorem, which allows one to prove optimality by giving primal and dual solutions that satisfy the complementary slackness conditions. In fact, we prove a generalization of complementarity (\cref{lemma:complementaritymany}), which allows us later to seek finite combinatorial solutions instead of continuous ones (\cref{fig:2discretecoloring}).

A similar duality, limited to a single bidder and to a restricted set of probability distributions, was used by \citet{Daskalakis:2013vn}. Their duality framework does not apply to the uniform distribution, the canonical example of continuous probability. Our approach manages to handle a much wider class of probability distributions, which includes the uniform distribution, by taking care of the boundary issues.

\subsection{Duality for the Uniform Distribution and a Single Bidder}
\label{sec:dual-unif-distr}

Then, we zoom in to the canonical problem for revenue maximization: \emph{we consider uniform distributions over $[0,1]^m$ of $m$ items and only one bidder}. This may seem like a special case, and in fact it is; however, despite being the canonical case of a very important problem, it has been open since the work of \citet{Myerson:1981aa}, except for some specialized approaches which successfully resolved the problem for two and three items (mostly using complicated necessary conditions and rather involved computations). Our approach gives an elegant framework to solve these cases and provides a natural description and understanding of the solution. It also gives rise to beautiful problems; in particular, for the case of 2 items we know (but not include here; see, e.g., \cite{g2014_2}) at least five different solutions for the problem, each with its own merits.

Our dual formulation of the problem can be rephrased as follows (see \cref{th:weakdualityuniform}): In the unit hypercube of $m$ dimensions, we seek functions $z_j(\vecc x)$, one for each dimension; each $z_j$ starts at value $0$ on the edge $(0, \vecc x_{-j})$ of the hypercube and rises up to value $1$ at the opposite edge $(1, \vecc x_{-j})$ of the hypercube. Given that the functions cannot rise rapidly (more precisely, the sum of their slopes cannot exceed $m+1$ at each point of the hypercube), find the functions with minimum sum of integrals. Alternatively, we can view it as a problem in the $m+1$ hypercube: each function $z_j$ defines a hypersurface which starts at the edge $(0, \vecc x_{-j})$ of the hypercube, ends at the opposite edge $(0, \vecc x_{-j})$, and they collectively cannot grow rapidly; we seek to minimize the sum of volumes beneath these surfaces (\cref{fig:2dimdual}). The remaining dual constraints $s_{j}$ do not appear anywhere, since for this application to the case of uniform distribution, we make the choice to relax even further the primal \hyperref[eq:totalrevenue]{Program~\eqref{eq:totalrevenue}} by dropping the corresponding $(s_{i,j})$ constraints that require the derivatives to be nonnegative; as we'll see, this is again without loss for the revenue optimality.  

\subsection{The Straight-Jacket Auction (SJA)}
\label{sec:stra-jack-auct}

This dual system suggests in a natural way a selling mechanism, the Straight-Jacket Auction (SJA). We explain the intuition behind the mechanism and give a formal definition in \cref{sec:SJAdef}. SJA is defined so that for every bundle of items $A$ with $\card{A}=r$, \emph{the price $p_{r}$ for $A$ is determined by the requirement that the volume of the $r$-dimensional body in which the mechanism sells a nonempty subset of $A$ is exactly equal to $r/(m + 1)$}.

The aim of the remaining and more technical part of the paper is to develop the toolkit to prove that SJA is optimal for any number of items; however, we manage to prove optimality only for up to 6 items.

The straightforward way for proving the optimality of SJA would be to find a pair of primal and dual solutions that have the same value. Although we know such explicit solutions for the case of two items, there does not seem to exist a natural solution of the dual program which can be easily described for more that two items. How then can we show optimality in such cases? \emph{We do not give an explicit dual solution, but we only show that a proper solution exists and rely on complementarity to show optimality.}

\subsection{Proof of Optimality of SJA}
\label{sec:proof-optimality-sja}

A central notion in our development is the notion of deficiency: the $k$-deficiency of  a body $S$ in $m$ dimensions is $|S|-k\, (\sum_j |S_{[m]\setminus \ssets{j}}|)$, where $S_{[m]\setminus \ssets{j}}$ denotes the projection of $S$ on the hyperplane $x_j=0$ (this is an $(m-1)$-dimensional body). In particular, we are interested in the deficiency of the subsets of $\overline U_{\emptyset}$, the valuation subspace in which the auction sells a nonempty bundle. The main tool for proving the optimality of SJA is the following: \emph{To show that the SJA is optimal it suffices to show that no set $S$ of points inside $\overline U_{\emptyset}$ has positive $\frac{1}{m+1}$-deficiency} (\cref{theorem:weak-deficiency}). 

The fact that this is sufficient is based mainly on the observation that finding a feasible dual solution is, in disguise, a perfect matching problem between the hypercube and its boundaries (taken with appropriate multiplicities). If Hall's condition for perfect matchings (see, e.g., \citep[Theorem 1.1.3]{Lovasz:1986qf}) could apply to infinite graphs, under some continuity assumptions the sufficiency of the above would be evident. However, Hall's theorem does not hold for infinite graphs in general~\citep{Aharoni:1991ab} and, even worse, the continuity assumptions seem hard to establish. We bypass both problems by considering an interesting discretized version of the problem that takes advantage of Hall's theorem and the piecewise continuity; we then apply approximate complementarity to prove optimality.

The technical core in our proof for the optimality of SJA consists of establishing that no positive deficiency subset of $\overline U_{\emptyset}$ exists. Let us call such a set a \emph{counterexample}. To prove that no counterexample exists, we first argue that such a counterexample would have certain properties and then show that no counterexample with these properties exists. We first show that we can restrict our attention to special types of counterexamples, those that are upwards closed and symmetric (\cref{th:wlogsymmetric}). Ideally, we would like to restrict our attention even further to box-like counterexamples, those that are the intersection of an $m$-dimensional box and $\overline U_{\emptyset}$. This would restrict significantly the search of counterexamples, and in fact a well-known isoperimetric lemma by \citet{Loomis:1949ul} (see \cref{th:loomis}), and a generalization by \citet{Bollobas:1995fr} show that this is actually true when we remove the restriction that the counterexample must lie inside some fixed body (in our case, inside $\overline U_{\emptyset}$). 
Unfortunately, we can only establish this claim for 2 items.
Instead, we prove a weaker version of it: we show that if a counterexample exists, it must be closed under taking the convex hull of all symmetric images of a point (\cref{th:pclosure}). Furthermore, the requirement on deficiency provides a lower bound on the volume of the counterexample (\cref{lemma:sizeboundcounter,th:counterchainwidths}).

By exploiting these properties, we show that no counterexample exists for 6 or fewer items (\cref{theorem:nopositivesubSIMs}). The case of 4 or fewer items is straightforward, but the case of 5 items is qualitatively more challenging. The main reason for this difficulty is that the optimal mechanism for 5 items never sells a bundle of 4 items (equivalently, the price for 4 items is equal to the price of 5 items). The case of 6 items is similar to the case of 5 items; the optimal mechanism does not sell any bundles of 5 items. However, all these cases are being treated in a unified way in the proof of the theorem that avoids tiresome case analysis.
We must point out here that \cref{theorem:nopositivesubSIMs} is essentially \emph{the only ingredient of this paper whose proof does not work for more than $6$ items}.

\section{Duality}\label{sec:duality}

Motivated by traditional linear programming duality theory, we develop a duality theory framework that can be applied to the problem \eqref{eq:totalrevenue} of designing auctions with optimal expected revenue. By interpreting the derivatives as differences, we can view this as an (infinite) linear program and we can find its dual. The variables of the primal linear program are the values of the functions $u_i(\vecc x)$.  The labels $(z_j(\vecc x))$ and $(s_{i,j}(\vecc x))$ on the constraints of Program \eqref{eq:totalrevenue} are the analog of the dual variables of a linear program.

To find its dual program, we first rewrite the objective function in terms of the $u_i$'s instead of their derivatives. In particular, by integration by parts  we have
\begin{align*}
  \hspace{-0.8cm}\int_D \frac{\partial u_i(\vecc x)}{\partial x_{i,j}} x_{i,j} f(\vecc x) \, d\vecc x &= \int_{D_{-(i,j)}} \left[ u_i(\vecc x) x_{i,j} f(\vecc x) \right]_{x_{i,j}=L_{i,j}}^{x_{i,j}=H_{i,j}} \, d\vecc x_{-(i,j)} - \int_D u_i(\vecc x) \frac{\partial (x_{i,j} f(\vecc x))}{\partial x_{i,j}} \, d\vecc x\\ 
&= \int_{D_{-(i,j)}} \left[ u_i(\vecc x) x_{i,j} f(\vecc x) \right]_{x_{i,j}=L_{i,j}}^{x_{i,j}=H_{i,j}} \, d\vecc x_{-(i,j)} - \int_D u_i(\vecc x) f(\vecc x) \, d\vecc x - \int_D u_i(\vecc x) x_{i,j} \frac{\partial f(\vecc x)}{\partial x_{i,j}} \, d\vecc x
\end{align*}
to rewrite the objective of the primal program as
\begin{align}
\label{eq:primal objective}
  \sum_{i=1}^n \int_{D}\left(\nabla u_i(\vecc x)\cdot \vecc x_i - u_i(\vecc
  x)\right)\, d F(\vecc x)   =& \sum_{i=1}^n\sum_{j=1}^m \int_{D_{-(i,j)}}  H_{i,j} \, u_i(H_{i,j},  \vecc x_{-(i,j)}) \, f(H_{i,j}, \vecc x_{-(i,j)}) \, dx_{-(i,j)}  \\ 
& - \sum_{i=1}^n\sum_{j=1}^{m} \int_{D_{-(i,j)}}  L_{i,j} \,  u_i(L_{i,j}, \vecc x_{-(i,j)})  \, f(L_{i,j}, \vecc x_{-(i,j)}) \, dx_{-(i,j)}  \nonumber \\ 
& - \sum_{i=1}^n \int_D u_i(\vecc x) \left( (m+1) f(\vecc x)+ \vecc x_i \cdot \nabla_i f(\vecc x) \right) \, d \vecc x. \nonumber
\end{align}
Notice that some of the above expressions make sense only for bounded domains (i.e., when $H_{i,j}$ is not infinity), but it is possible to extend the duality framework to unbounded domains, by carefully replacing these expressions with their limits when they exist or by appropriately truncating the probability distributions. For the main results in this work we deal only with bounded domains, but for completeness and future reference we provide a treatment of the general case in \cref{sec:dualityinfinity}.

\emph{We also relax the original problem} by replacing the convexity constraint
by the much milder constraint of \emph{absolute continuity}; absolute continuity allows us to express functions as integrals of their derivatives. We can restate this as follows: truthfulness in general imposes two conditions on the solution of allocating the items to bidders (see \cref{thm:truthfulconvex}): the first condition is that the utility is convex; the second one is that the allocations must be gradients of the utility. It seems that in most cases, including the important Myersonian case of one item and regular distributions, when we optimize revenue the convexity constraint is redundant. Later on, when we will be applying the duality framework to the case of uniform distributions we will also drop the constraints of nonnegative allocation probabilities (i.e., the $(s_{j}(\vecc x))$ constraints in~\eqref{eq:totalrevenue}). 
In many cases, dropping these constraints might have no effect on the value of the program. However, there are cases in which these constraints are essential. In particular, they are needed even for the case of one item when the probability distributions are not regular. We give an in-depth discussion of this topic in \cref{sec:convexity-duality}.

To find the dual program, we have to take extra care on the boundaries of the domain, since the derivatives correspond to differences from which one term is missing (the one that corresponds to the variables outside the domain). This is a point where our approach differs from that of \citet{Daskalakis:2013vn}, which applies only to special distributions and in particular it does not apply to the uniform distribution.  Inside the domain, the dual constraint that corresponds to the primal variable $u_i(\vecc x)$ is $\sum_j \frac{\partial z_j(\vecc x)}{\partial x_{i,j}} \leq (m+1) f(\vecc x)+ \vecc x_i \cdot \nabla_i f(\vecc x) $. 

So, the dual program that we propose is
\begin{mdframed}
\begin{equation}
  \label{eq:dual}
  \inf_{z_1,\ldots,z_m} \int \sum_{j=1}^m z_j(\vecc x) \, d\vecc x
\end{equation}
subject to
\begin{align*}
  \tag{$u_i(\vecc x)$}   \sum_{j=1}^m \left(\frac{\partial z_j(\vecc x)}{\partial x_{i,j}}- \frac{\partial s_{i,j}(\vecc x)}{\partial x_{i,j}} \right)&\leq
   (m+1) f(\vecc x)+ \vecc x_i \cdot \nabla_i f(\vecc x)  \\
  \tag{$u_i(L_{i,j}, \vecc x_{-(i,j)})$}    z_j(L_{i,j}, \vecc x_{-(i,j)})-s_{i,j}(L_{i,j}, \vecc x_{-(i,j)}) &\leq
  L_{i,j} f(L_{i,j}, \vecc x_{-(i,j)}) \\
  \tag{$u_i(H_{i,j}, \vecc x_{-(i,j)})$}   z_j(H_{i,j}, \vecc x_{-(i,j)}) - s_{i,j}(H_{i,j}, \vecc x_{-(i,j)}) &\geq
  H_{i,j} f(H_{i,j}, \vecc x_{-(i,j)}) \\
  z_j(\vecc x),s_{i,j}(\vecc x) &\geq 0
\end{align*}
\end{mdframed}

The above intuitive derivation of this dual is used only for illustration and for explaining how we came up with it. None of the results rely on the actual way of coming up with the dual problem. However, the derivation is useful for intuition and for suggesting traditional linear programming machinery for these infinite systems; for example, although we don't directly use any results from the theory of linear programming duality, we are motivated by it to prove similar connections between our primal and dual programs.

One can interpret this dual as follows: For the sake of clarity, assume a single bidder and drop the $s_{i,j}$ constraints;  we seek $m$ functions $z_j$ defined inside the hyperrectangle $[L_1, H_1]\times\cdots\times [L_m, H_m]$ such that
\begin{itemize}
\item in the $j$-th direction, function $z_j$ starts at value (at most) $L_j f(L_j, \vecc x_{-j})$ and ends at value (at least) $H_j f(H_j, \vecc x_{-j})$; this must hold for all $\vecc x_{-j}$.
\item at every point of the domain, the sum of the derivatives of functions $z_j$ cannot exceed $(m+1) f(\vecc x)+ \vecc x\cdot\nabla f(\vecc x)$.
\item the sum of the integrals of these functions is minimized.
\end{itemize}

\begin{figure}[t]
\centering
\begin{subfigure}{0.485\textwidth}
\includegraphics[width=1\textwidth]{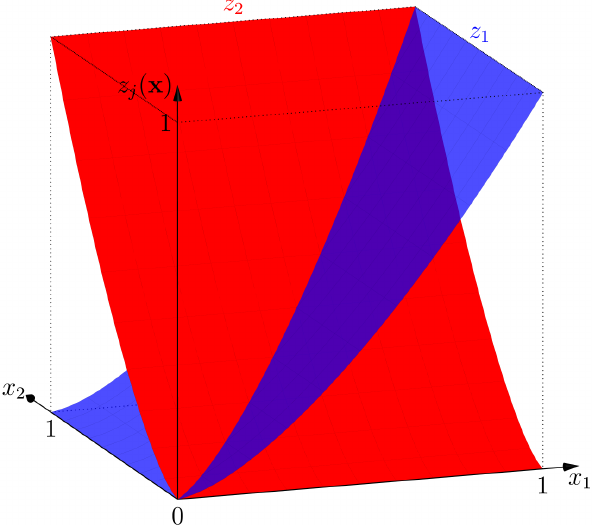}
\caption{\footnotesize Feasible solutions $z_1,z_2$ to the two-items dual program. Each function $z_j$ has to start at $0$ on the entire axis $x_j=0$ and rise to $1$. At no point of the $2$-dimensional cube the sum of their slopes is allowed to exceed $3$, and the objective is to keep them as low as possible, i.e., minimize the volume under their curves.}
\label{fig:2dimdual}
\end{subfigure}
~
\begin{subfigure}{0.485\textwidth}
\includegraphics[width=1\textwidth]{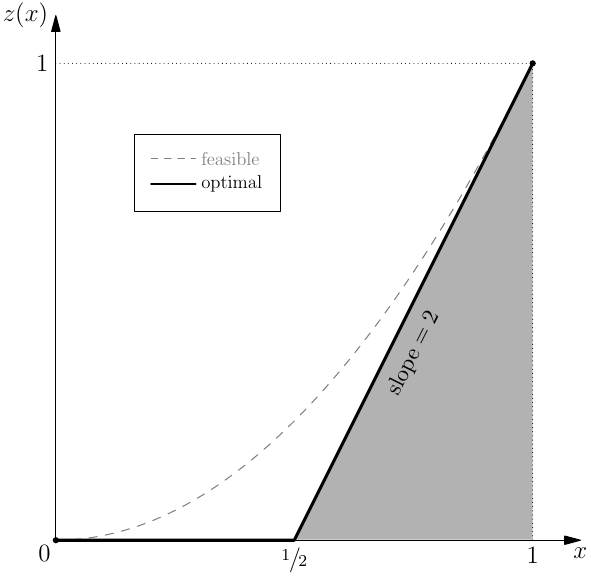}
\caption{\footnotesize For the special case of a single item, the dual feasible function $z$ has to start at $0$ and rise to $1$ or higher when $x=1$, with a slope of at most $2$. The optimal function minimizes the area below it. It is not difficult to see that the optimal solution is to remain at value 0 until $x=1/2$ and then increase steadily to $1$; the optimal dual objective is equal to the gray area. This corresponds exactly to the well-known optimal solution of Myerson with reserve price of $1/2$.}
\label{fig:1dimdual}
\end{subfigure}
\caption{\footnotesize Geometric interpretation of the dual \hyperref[eq:dual]{Program~\eqref{eq:dual}} for the case of a single bidder and $m=1,2$ uniform i.i.d.\ items.}
\label{fig:dualprogram}
\end{figure}

For a significant portion of this paper, we materialize this duality framework by applying it to the case of i.i.d.~uniform distributions over the unit interval $I^m$. Therefore, let's clearly state our dual constrains for ease of reference:

\begin{remark}[Duality for Uniform Domains]\label{th:weakdualityuniform}
The dual constraints (in \hyperref[eq:dual]{Program~\eqref{eq:dual}}) for the single-bidder $m$-items uniform i.i.d.\ setting over $I^m$ become

\begin{align*}\nopagebreak
  \tag{$u(\vecc x)$}   \sum_ {j=1}^m \frac{\partial z_j(\vecc x)}{\partial x_j} &\leq
   m+1 \\
  \tag{$u(0, \vecc x_{-j})$}    z_j(0, \vecc x_{-j}) &= 0 \\
  \tag{$u(1, \vecc x_{-j})$}   z_j(1, \vecc x_{-j}) &\geq 1 \\
  z_j(\vecc x) &\geq 0
\end{align*}
A geometric interpretation of this dual for the case of one and two items, based on the previous discussion, can be found in \cref{fig:dualprogram}. 
\end{remark}

Let us also mention parenthetically that one can derive Myerson's results by selecting as variables not the utilities $u_i(\vecc x)$, but their derivatives. In fact, since the allocation constraints involve exactly the derivatives, this is the natural choice of primal variables. Unfortunately, such an approach does not seem to work for more than one item because the derivatives are not independent functions. If we treat them as independent, we lose the power of the gradients constraint.

\subsection{Duality and Complementarity}
\label{sec:dual-compl}

The way that we derived the dual system does not yet provide any rigorous connection with the original primal system. We now prove that this is indeed a weak dual, in the sense that the value of the dual minimization \hyperref[eq:dual]{Program~\eqref{eq:dual}} cannot be less than the value of the primal program. 

\begin{mdframed}
\begin{lemma}[Weak Duality]\label{lemma:weakdualitymany}
  The value of every feasible solution of the primal \hyperref[eq:totalrevenue]{Program~\eqref{eq:totalrevenue}} does not exceed the value of any feasible solution of the dual \hyperref[eq:dual]{Program~\eqref{eq:dual}}.
\end{lemma}
\end{mdframed}

\begin{proof}
The proof is essentially a straightforward adaptation of the proof of traditional weak
duality for finite linear programs. Take a pair of feasible solutions for
the primal and the dual programs and consider the difference
between the dual objective~\eqref{eq:dual} and the primal objective~\eqref{eq:primal objective}:
\begin{align}
& \sum_{j=1}^m \int_D z_j(\vecc x) \, d\vecc x  + \sum_{i=1}^n \int_D u_i(\vecc
x) \, \left( (m+1) f(\vecc x)+ \vecc x_i \cdot \nabla_i f(\vecc x) \right) \,
d\vecc x \label{eq:weak_duality_ref3} \\
& - \sum_{i=1}^n\sum_{j=1}^m \int_{D_{-(i,j)}}  H_{i,j} \, u_i(H_{i,j},  \vecc
x_{-(i,j)}) \, f(H_{i,j}, \vecc x_{-(i,j)}) \, d\vecc x_{-(i,j)} \notag\\
& + \sum_{i=1}^n\sum_{j=1}^m \int_{D_{-(i,j)}}  L_{i,j} \,  u_i(L_{i,j}, \vecc
x_{-(i,j)})  \, f(L_{i,j}, \vecc x_{-(i,j)}) \, d\vecc x_{-(i,j)}  \notag
\end{align}
Using the constraints of the programs, the first two terms of this expression can be bounded from below by
\begin{multline}
\label{eq:weak_duality_ref1}
  \sum_{j=1}^m \int_D z_j(\vecc x) \sum_{i=1}^n \frac{\partial u_i(\vecc x)}{\partial x_{i,j}} \, d\vecc x - \sum_{i=1}^n\sum_{j=1}^m \int_D s_{i,j}(\vecc x) \frac{\partial u_i(\vecc x)}{\partial x_{i,j}} \, d\vecc x\\  
  + \sum_{i=1}^n \int_D u_i(\vecc x)\left( \sum_{j=1}^m\frac{\partial z_j(\vecc x)}{\partial x_{i,j}} - \sum_{j=1}^m\frac{\partial s_{i,j}(\vecc x)}{\partial x_{i,j}}\right) \, d\vecc x  
\end{multline}
which equals
$$
\sum_{i=1}^n\sum_{j=1}^m \int_D \frac{\partial \left[(z_j(\vecc x)-s_{i,j}(\vecc x)) u_i(\vecc x)\right]}{\partial x_{i,j}} \, d\vecc x.
$$
Similarly, the other two terms of the expression can be bounded from below by
\begin{multline}
\label{eq:weak_duality_ref2}
  - \sum_{i=1}^n\sum_{j=1}^m \int_{D_{-(i,j)}}  u_i(H_{i,j},  \vecc
  x_{-(i,j)}) \left(z_j(H_{i,j},  \vecc x_{-(i,j)}) - s_{i,j}(H_{i,j},  \vecc x_{-(i,j)})\right) \, d\vecc x_{-(i,j)}
  \\
+ \sum_{i=1}^n\sum_{j=1}^m \int_{D_{-(i,j)}}  u_i(L_{i,j},  \vecc x_{-(i,j)}) \left(z_j(L_{i,j},  \vecc x_{-(i,j)}) - s_{i,j}(L_{i,j},  \vecc x_{-(i,j)}) \right) \, d\vecc x_{-(i,j)} 
\end{multline}
and they cancel out the first two terms. Bringing everything together, the difference of the dual and primal objectives is bounded from below by zero.
\end{proof}

We can use weak duality to show optimality: it suffices to have a pair of feasible primal and dual solutions that give the same value. In many cases, such as the case of uniform distributions, computing the optimal value is not easy or it may not even be expressible in a closed form. In such a case, a useful tool to prove optimality is through complementarity. In fact, we will prove a slight generalization of traditional Linear Programming complementarity which will allow us later to discretize the domain and consider approximate solutions. Specifically, instead of requiring the product of primal and corresponding dual constraints to be zero, we generalize it to be bounded above by a constant:

\begin{mdframed}
\begin{lemma}[Complementarity]\label{lemma:complementaritymany}
Suppose that $u_i(\vecc x)$ is a feasible solution of the primal \hyperref
[eq:totalrevenue]{Program~\eqref{eq:totalrevenue}} and $z_j(\vecc
x)$, $s_{i,j}(\vecc x)$ is a feasible solution of the dual \hyperref[eq:dual]
{Program~\eqref{eq:dual}}. Fix some parameter $\varepsilon\geq 0$. If the
following complementarity constraints hold for all $i\in[n]$, $j\in[m]$ and a.e. $\vecc x \in D$,
\begin{align*}
  u_i(\vecc x) \cdot  \left( (m+1) f(\vecc x)+ \vecc x_i \cdot \nabla_i f(\vecc x)
    -\sum_{j =1}^m \frac{\partial z_j(\vecc x)}{\partial x_{i,j}}+\sum_{j=1}^m \frac{\partial s_{i,j}(\vecc x)}{\partial x_{i,j}}
  \right)& \leq \varepsilon f(\vecc x) \\
  u_i(L_{i,j}, \vecc x_{-(i,j)}) \cdot \left( L_{i,j}f(L_{i,j}, \vecc x_{-(i,j)})
    - z_j(L_{i,j}, \vecc x_{-(i,j)}) +s_{i,j}(L_{i,j}, \vecc x_{-(i,j)})\right) & \leq \varepsilon f(L_{i,j}, \vecc x_{-(i,j)}) \\
  u_i(H_{i,j}, \vecc x_{-(i,j)}) \cdot \left( z_j(H_{i,j}, \vecc x_{-(i,j)}) - s_{i,j}(H_{i,j}, \vecc x_{-(i,j)})
    -H_{i,j}f(H_{i,j}, \vecc x_{-(i,j)}) \right) & \leq \varepsilon f(H_{i,j}, \vecc x_{-(i,j)}) \\
  z_j(\vecc x) \cdot \left( 1 - \sum_{i=1}^n\frac{\partial
        u_i(\vecc x)}{\partial x_{i,j}} \right) & \leq \varepsilon f(\vecc x)\\
  s_{i,j}(\vecc x) \cdot \frac{\partial
        u_i(\vecc x)}{\partial x_{i,j}} & \leq \varepsilon f(\vecc x),  
\end{align*}
then the primal and dual objective values differ by at most
$(n+m+3nm)\varepsilon$. In particular, if the conditions are satisfied
with $\varepsilon=0$, both solutions are optimal.
\end{lemma}  
\end{mdframed}
\begin{proof}
  We take the sum of all complementarity constraints and integrate in the domain:
  \begin{multline*}
   \sum_{i=1}^n \int_D  u_i(\vecc x)  \left( (m+1) f(\vecc x)+ \vecc x_i \cdot \nabla_i f(\vecc x) -\sum_{j=1}^m \frac{\partial z_j(\vecc x)}{\partial x_{i,j}}+\sum_{j=1}^m \frac{\partial s_{i,j}(\vecc x)}{\partial x_{i,j}}
  \right) \, d\vecc x \\
  +\sum_{i=1}^n\sum_{j=1}^m \int_{D_{-(i,j)}}  u_i(L_{i,j}, \vecc x_{-(i,j)})  \left( L_{i,j}f(L_{i,j}, \vecc x_{-(i,j)})
    - z_j(L_{i,j}, \vecc x_{-(i,j)}) + s_{i,j}(L_{i,j}, \vecc x_{-(i,j)}) \right) \, d\vecc x_{-(i,j)} \\
   +\sum_{i=1}^n\sum_{j=1}^m \int_{D_{-(i,j)}}  u_i(H_{i,j}, \vecc x_{-(i,j)})  \left( z_j(H_{i,j}, \vecc x_{-(i,j)}) - s_{i,j}(H_{i,j}, \vecc x_{-(i,j)})
    -H_{i,j}f(H_{i,j}, \vecc x_{-(i,j)}) \right) \, d\vecc x_{-(i,j)}  \\
  +\sum_{j=1}^m \int_D z_j(\vecc x) \left( 1 - \sum_{i=1}^n\frac{\partial
        u_i(\vecc x)}{\partial x_{i,j}} \right) \, d \vecc x +\sum_{i=1}^n\sum_
        {j=1}^m\int_D s_{i,j}(\vecc x) \cdot \frac{\partial
        u_i(\vecc x)}{\partial x_{i,j}}\, d\vecc x\leq (n+m+3nm) \varepsilon
  \end{multline*}
  It suffices to notice that the left-hand side is just the sum of \eqref{eq:weak_duality_ref3},
  \eqref{eq:weak_duality_ref1} and \eqref{eq:weak_duality_ref2}, so by using the
  same transformations
  that we used to prove the Weak Duality \cref{lemma:weakdualitymany}, it is
  equal to the dual objective minus the primal objective.
\end{proof}

\section{The Straight-Jacket Auction (SJA)}
\label{sec:tight-slic-mech}
\label{sec:SJAdef}

\begin{figure}[t]
\centering
\includegraphics[width=0.44\textwidth]{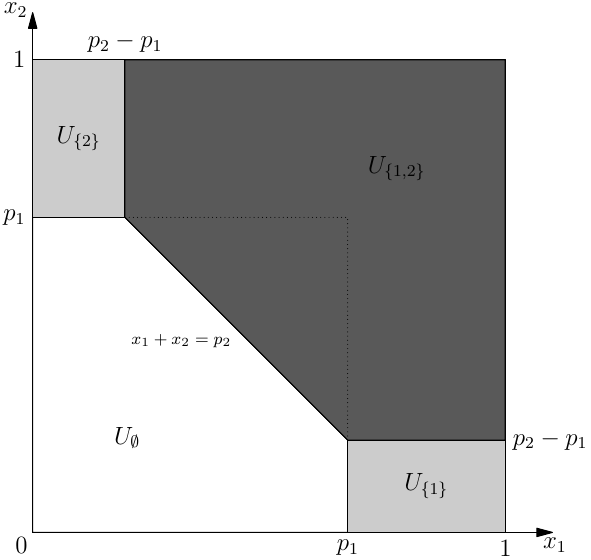}
\includegraphics[width=0.55\textwidth]{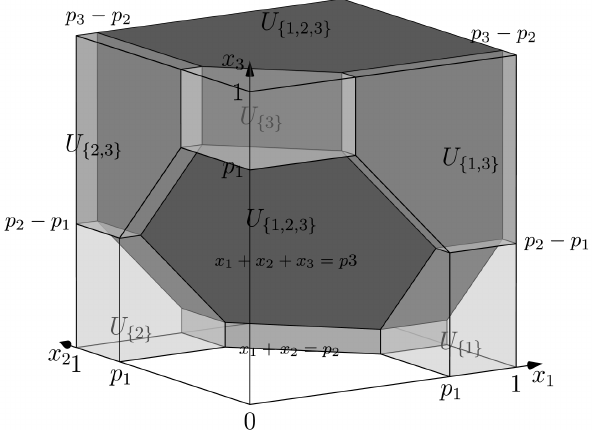}
\caption{\footnotesize The allocation spaces of the optimal SJA mechanisms for $m=2$ and $m=3$ items. The payments are given by $p_1=\frac{m}{m+1}$, $p_2=\frac{2m-\sqrt{2}}{m+1}$, and $p_3=3-\frac{7.0971}{m+1}$. The mechanism sells at least one item within the gray areas $\overline U_{\emptyset}$, and all items within the \emph{dark} gray areas $U_{[m]}$. If we flip around these dark gray areas by $\vecc x\mapsto \vecc 1-\vecc x$, so that $\vecc 1$ is mapped to the origin $\vecc 0$, they are exactly the SIM-bodies defined in \cref{sec:SIMbodies2}, for $k=\frac{1}{m+1}$. These SIM-bodies can be seen in \cref{fig:2dimSIM,fig:3dimSIM}, respectively.}
\label{fig:SJA2and3}
\end{figure}

In the rest of the paper we demonstrate the power and usage of the duality framework developed in \cref{sec:duality}, by applying it to the canonical open problem of revenue maximization in the economics literature: that of a single bidder setting where item valuations come i.i.d.~from a uniform distribution over $[0,1]$. Recall that now the general dual \hyperref[eq:dual]{Program~\eqref{eq:dual}} takes the simple form shown in \cref{th:weakdualityuniform}, where the $s_{i,j}$ variables do not appear since we have chosen to relax even further the primal \hyperref[eq:totalrevenue]{Program~\eqref{eq:totalrevenue}} by dropping the nonnegative derivatives constraint; this will end up being without loss to optimality. 

The duality conditions are not only useful in establishing optimality; they can in fact \emph{suggest the optimal auction in a natural way}. We illustrate this by considering the case of 2 items. Starting from \cref{fig:2dimdual}, we need to find two functions $z_1$ and $z_2$ that satisfy the boundary constraints and the slope constraint. If we had only one function, say $z_1$, the solution would be obvious and similar to the solution for one item (\cref{fig:1dimdual}): $z_1(\vecc x)$ would be 0 up to $x_1=2/3$ and then increase with a maximum slope of~3. But if we do the same for both functions $z_1$ and $z_2$, we obtain an infeasible solution: in the square $[2/3,1]\times[2/3,1]$ the total slope would be 6 instead of 3. This implies that the functions need more space to grow; in fact, the area of growth needs to be at least equal to the area of the square $[2/3,1]\times[2/3,1]$. The natural way to get this space is to add a triangle of area $1/9$ in the way indicated in the left part of \cref{fig:SJA2and3} (the triangle defined by the lines $x_j=p_1=2/3$, $j=1,2$, and $x_1+x_2=p_2$). We then seek a dual solution in which only $z_j$ grows in area $U_{\{j\}}$ and both functions grow in $U_{\{1,2\}}$ (\cref{fig:SJA2and3}). The corresponding primal solution is that only item $j$ is sold in $U_{\{j\}}$ and both items are sold in $U_{\{1,2\}}$.

The remarkable fact is that the optimal mechanism is completely determined by the obvious requirement that the area of the triangle must be (at least) equal to $1/9$. To put it in another way: suppose that we knew that the optimal mechanism is deterministic; then the dual program requires that
\begin{itemize}
	\item the price $p_1$ for one item must satisfy $p_1\leq 2/3$ so that $z_1$ has enough space to grow from 0 to 1 with the maximum slope 3
	\item the price $p_2$ for the bundle of both items must be such that the area of the region $\overline U_{\emptyset}=U_{\ssets{1}}\union U_{\ssets{2}}\union U_{\ssets{1,2}}$, in which the mechanism allocates at least one item, is at least $2/3$ so that both functions have enough space to grow to 1
\end{itemize}
The central point of this work is that these necessary conditions (which we call \emph{slice conditions}) are also sufficient. This intuition naturally extends to more items: the price for a bundle of $r$ items is determined by the slice condition that the $r$-dimensional volume in which the mechanism sells at least one item of the bundle is exactly equal to $r/(m+1)$.

Using this intuition, we define here the \emph{Straight-Jacket Auction (SJA)}. This selling mechanism is deterministic and symmetric; as such, it is defined by a payment vector $p^{(m)}=(p_1^{(m)},\ldots,p_m^{(m)})$; $p_r^{(m)}$ is the price offered by the mechanism to the bidder for every subset of $r$ items, $r\in [m]$. We will drop the superscript when there is no confusion about the number of available items. The utility of the bidder is then given by $u(x)=\max_{J\subseteq [m]} \left( \sum_{j\in J} x_j - p_{|J|} \right).$

The prices are defined by the slice conditions. For a subset of items $J \subseteq [m]$, let $\Pr(J, \vecc x_{-J})$ be the probability that at least one item in $J$ is sold when the remaining items have values $\vecc x_{-J}$. The $r$-th dimensional slice condition is that for every $J$ with $|J|=r$ and every $\vecc x_{-J}$: $\Pr(J, \vecc x_{-J}) \geq |J|/(m+1)$. The SJA is the deterministic mechanism which satisfies the slice conditions for all dimensions \emph{as tightly as possible} (hence its name), in the following sense: determine the prices $p_1, p_2, \ldots, p_m$ in this order so that, having fixed the previous ones, select $p_r$ as large as possible to satisfy all $r$-dimensional slice conditions. In particular, this guarantees that the $m$-dimensional slice is tight, or equivalently, that the probability that at least one item is sold is $m/(m+1)$.

\begin{mdframed}
\begin{definition}[Straight-Jacket Auction (SJA)]
\label{def:SJAmechanism}
  SJA for $m$ items is the deterministic symmetric selling mechanism whose
  prices $p_1^{(m)}$, \ldots, $p_m^{(m)}$, where $p_r^{(m)}$ is the
  price of selling a bundle of size $r$, are determined as follows:
  for each $r\in[m]$, having fixed $p_1^{(m)}$, \ldots, $p_{r-1}^{(m)}$, price
  $p_r^{(m)}$ is selected to satisfy
  \begin{equation}
    \label{eq:sja}
    \probabilityext{\vecc x\sim {\uniformunit}^m}{\bigland_{J\subseteq[r]}\sum_{j\in J} x_j
  <p_{|J|}^{(m)}}=1-r\cdot k,
  \end{equation}
  where $k=\frac{1}{m+1}$.  In words, $p_r^{(m)}$ is selected so that the
  probability of selling no item when $r$ values are drawn from the
  uniform probability distribution (and the remaining values of the
  $m-r$ items are set to 0) is equal to $1-r \cdot k$.
  We will refer to constraints~\eqref{eq:sja} as \emph{slice conditions}. 
\end{definition}
\end{mdframed}
If we take the complement of the above probability, an equivalent definition would be to ask for the probability of selling at least one of items $[r]$, when all other bids for items $[r+1...m]$ are fixed to zero, to be $rk$. That is, if for any dimension $m$ and positive $\alpha_1,\alpha_2,\dots, \alpha_m$ we define
  \begin{equation}
    \label{eq:disjunctionspace}
    V(\alpha_1,\ldots,\alpha_m)\equiv
    \sset{\vecc x\in I^m \fwh{\biglor_{J\subseteq[m]} \sum_{j\in J}  x_j \geq \alpha_{\card{J}}
   }},
  \end{equation}
the volume of the $r$-dimensional body $V(p_1^{(m)},\dots,p_r^{(m)})$, let's
denote it by $v(p_1^{(m)},\dots,p_r^{(m)})$, must be $r k$ (for all $r\in[m]$).
Notice also that it is not immediate that SJA is in general well-defined for any
dimension $m$:  there should exist prices $p_r^{(m)}$ that satisfy~\eqref
{eq:sja}.

The specific value on the right-hand side of~\eqref{eq:sja} depends on the parameter $k$, which, in turn, depends on the total number of items $m$; the exact dependence arises from the specific values of the primal and dual program. It is, however, useful in providing a unifying approach to carry out the discussion and analysis for an arbitrary (albeit small, $k\leq \frac{1}{m+1}$) parameter $k$ and to plug in the specific value $k=1/(m+1)$ only when this is absolutely necessary.

The main technical result of this work is showing that the SJA mechanism is
optimal for $m\leq 6$:
\begin{mdframed}
\begin{theorem}\label{th:mainresultSJA6}
The Straight-Jacket Auction is a revenue optimal mechanism for selling up to 6
goods to a single additive buyer having uniformly i.i.d.\ valuations over $
[0,1]$.
\end{theorem}
\end{mdframed}

Our proof of this theorem relies significantly on the geometry of these mechanisms. We conjecture that the theorem holds for any number of items:

\begin{conjecturenonum}\label{th:SJAoptimalany}
The Straight-Jacket Auction is a revenue optimal mechanism for selling \emph{any} number of
goods to a single additive buyer having uniformly i.i.d.\ valuations over $
[0,1]$.
\end{conjecturenonum}

Here is how to use the slice conditions~\eqref{eq:sja} to compute the prices of SJA: The $1$-dimensional condition on a $1$-dimensional hypercube simply means that $p_1^{(m)} = 1-1/(m+1)$, because we only have condition $x_1<p_1^{(m)}$. The $2$-dimensional  condition on a $2$-dimensional boundary requires that the region $\{\vecc x\, :\, x_1+x_2< p_2 \;\;\text{and}\;\; x_1< p_1 \;\;\text{and}\;\; x_2 < p_1\}$ inside the unit square must have area equal to $1-2/(m+1)$. In other words, we want to find where to move the line $x_1+x_2=p_2$ so that the area that it cuts satisfies the slice condition (in the left part of \cref{fig:SJA2and3}, $U_{\ssets{1}}$, $U_{\ssets{2}}$, and $U_{\ssets{1,2}}$ have total volume $2/(m+1)$); this gives $p_2=2-(2+\sqrt{2})/(m+1)$. We can proceed in the same way to higher dimensions: fix some dimension $m$ and an order $r>1$. If the prices $p_1,p_2,\dots,p_r$ are such that $p_{j}-p_{j-1}$ is a nonnegative and (weakly) decreasing sequence, then
\begin{align}
v\left(p_1,\dots,p_r\right) &=\int_0^{p_{r}-p_{r-1}}v\left(p_1,\dots,p_{r-1}\right)\,dt+\int_{p_{r}-p_{r-1}}^{p_{r-1}-p_{r-2}}v\left(p_1,\dots,p_{r-2},p_{r}-t\right)\,dt\notag\\
+\dots &+\int_{p_{2}-p_{1}}^{p_{1}}v\left(p_2-t,\dots,p_{r-1}-t,p_{r}-t\right)\,dt+\int_{p_{1}}^1 1\,dt.\label{eq:recsliceSJAvol}
\end{align}
This is a recursive way to compute the expressions for the volumes $v\left(p_1,\dots,p_r\right)$. In case that the sequence $p_1,p_2,\dots,p_r$ of the prices up to order $r$ breaks the requirement to be increasing at the last step, i.e.~$p_{r}<p_{r-1}$, then simply $v\left(p_1,\dots,p_r\right)=v\left(p_1,\dots,p_{r-2},p_r,p_r\right)$  and we can still deploy the previous recursion.

An exact, analytic computation of these values for up to $r=6$ using the above recursion is given in \cref{sec:exactcompute6}, but we also list them below for quick reference. In the following we will often use the transformation
\begin{equation}\label{eq:transfpm}
p_r=r-\frac{\mu_r}{m+1}
\end{equation} 
so that prices will be determined with respect to some parameters $\mu_r$. It will be algebraically convenient to also assume $p_0=0$. 

\begin{itemize}
\item For $r\leq 4$ and \emph{any} number of items $m\geq r$:
\begin{align}
p_1 &= \frac{m}{m+1} &	p_2 &= \frac{2m-\sqrt{2}}{m+1}&	p_3 &\approx 3-\frac{7.0972}{m+1}&	p_4 &\approx 4-\frac{11.9972}{m+1} \label
{eq:prices_numeric_1_4}\\
\mu_1 &=1 &	\mu_2 &=2+\sqrt{2}&	\mu_3 &\approx 7.0972&	\mu_4 &\approx 11.9972
\notag
\end{align}

\item For $r=5,6$:
\begin{align}
p_5^{(5)} &\approx 1.9856 &	p_5 &\approx 5-\frac{18.0843}{m+1}\;\; 
(m\geq 6)&	p_6^{(6)} &\approx 2.3774 \label
{eq:prices_numeric_5_6}\\
\mu_5^{(5)} &\approx 18.0865 &	\mu_5 &\approx 18.0843\;\; (m\geq 6)&	\mu_6^{(6)}
&\approx 25.3585 \notag
\end{align}

\end{itemize}

\subsection{Optimality of SJA}
\label{sec:optimality-SJA-gather}
In this section we gather the key elements that form the backbone of our proof for the optimality of the SJA mechanism. 

\begin{definition}
  We denote by $U_J^{(m)}$ the subdomain in which SJA allocates exactly the
  bundle $J\subseteq [m]$ of items:
  \begin{equation}
    \label{eq:def-U_A}
    U_J^{(m)}\equiv\sset{ \vecc x\in I^m \fwh{\bigland_{L\subseteq[m]} \sum_{j\in J} x_j - p_{|J|}^{(m)} \geq \sum_{j\in
          L} x_j - p_{|L|}^{(m)}}}.
  \end{equation}
\end{definition}

Let $\slice{U_J^{(m)}}{-J}{\vecc t}$ denote the $|J|$-dimensional slice of
  $U_J^{(m)}$ when we fix the values of the remaining $[m]\setminus J$ items to $\vecc t$:
  \begin{equation*}
    \label{eq:def-slice}
    \slice{U_J^{(m)}}{-J}{\vecc t}  = \{ x_J \,:\, (\vecc x_J,\vecc t) \in U_J \}. 
  \end{equation*}
For example, the slices of
$U_{\{1\}}^{(m)}$ are the horizontal ($1$-dimensional) intervals; when $J=[m]$,
$U_J^{(m)}$ has only one slice: itself. \Cref{fig:SJA2and3} shows the various subdomains $U_J^{(m)}$ for $m=2,3$.

We next define the notion of deficiency of a body, which is one of the key geometric ingredients in this paper. It captures how large an $m$-dimensional body $A$ is with respect to its $(m-1)$-dimensional projections, which are denoted by $A_{[m]\setminus \ssets{j}}$ (see the beginning of the next section for a formal definition). The $k$-deficiency is the difference of the volume $\card{A}$ of the body $A$ minus the volumes $k \card{A_{[m]\setminus \ssets{j}}}$ of each $m$-dimensional prism that results when we extend a $(m-1)$-dimensional projection by height $k$. This is inspired by the deficiency notion in bipartite graphs defined by~\citet{Ore:1955fk}.  

\begin{definition}[Deficiency]\label{def:deficiency} For any $k>0$, we will call \emph{$k$-deficiency} of an $m$-dimensional body $A\subseteq \R_+^m$ the quantity
\begin{equation}
\label{eq:deficiency}
\delta_{k}(A)\equiv \card{A}-k\sum_{j=1}^m\card{A_{[m]\setminus \ssets{j}}}.
\end{equation}
\end{definition}
From now on we will sometimes drop the subscript in the deficiency notation $\delta_k$ whenever it is clear from the context what $k$ we are referring to, or if we want to make a general statement that holds for all values of parameter $k$ (see, e.g., \cref{th:deficienciesunioninters}).

\begin{definition}[SIM-bodies]
\label{def:sim-body}
  For positive $\alpha_1\leq \cdots\leq \alpha_r$, let
  \begin{equation}
    \label{eq:sim-body}
    \Lambda(\alpha_1,\ldots,\alpha_r)\equiv
    \sset{\vecc x\in \R^r_+ \fwh{\bigland_{J\subseteq[r]} \left(\sum_{j\in J}  x_j \;\;\leq \sum_{j=r-|J|+1}^{r}
    \alpha_j \right) }}.
  \end{equation}
  We call these \emph{SIM-bodies}\footnote{The naming is inspired by the
  familiar, characteristic shape of mobile phone SIM cards; see \cref
  {fig:2dimSIM} for the apparent resemblance in $2$-dimensional space.}. We will
  also use the following notation:
$$q\cdot\Lambda(\alpha_1,\ldots,\alpha_r)\equiv\Lambda(q \cdot \alpha_1,\ldots,q \cdot \alpha_r) 
$$
for any positive real $q$.
\end{definition}
It turns out that SIM-bodies (see \cref{fig:dimSIM}) are essentially the building blocks of the allocation space of SJA:
\begin{lemma}
\label{lemma:SIMallocspaceisom}
 Every nonempty slice $\slice{U_J^{(m)}}{-J}{\vecc t}$ of SJA is isomorphic to the
 SIM-body $k\cdot \Lambda(\lambda_1^{(m)},\ldots,\lambda_{|J|}^{(m)})$, where $k=1/(m+1)$. The parameters $\lambda_r^{(m)}$ depend on the payments of SJA as follows:
\begin{equation}
\label{eq:lambdasdef}
\lambda_r^{(m)}\equiv\mu_r^{(m)}-\mu_{r-1}^{(m)},
\end{equation}
where $\mu_r^{(m)}$ is defined in \eqref{eq:transfpm}.
\end{lemma}
This essentially establishes a correspondence between SJA subdomains  $U^{(m)}_{J}$ and SIM-bodies $\Lambda(\lambda_1^{(m)},\dots,\lambda_{\card{J}}^{(m)})$, for every $J\subseteq [m]$.
The main geometric property of SJA is captured by the following theorem, the proof of which appears in \cref{sec:optimality-sja}.
\begin{theorem}
\label{theorem:nopositivesubSIMs}
  For $m\leq 6$ and for the $\lambda$'s defined in \eqref{eq:transfpm}, no SIM-body $\Lambda(\lambda_1^{(m)},\ldots,\lambda_{\card{J}}^{(m)})$ corresponding to a nonempty subdomain $U_J^{(m)}$
  contains positive 1-deficiency sub-bodies.
\end{theorem}
  Using this geometric property, we  prove in \cref{sec:optimality-sja} the optimality of SJA:
\begin{theorem}
\label{theorem:weak-deficiency}
If for every nonempty subdomain $U^{(m)}_{J}$ of SJA the corresponding SIM-body $\Lambda(\lambda_1^{(m)},\dots,\lambda_{\card{J}}^{(m)})$ contains no sub-bodies of positive $1$-deficiency, then SJA is optimal. 
\end{theorem}
Notice that the last theorem applies to any number of items, but the proof of optimality is restricted to 6 items by \cref{theorem:nopositivesubSIMs}.  The rest of the paper focuses in formalizing these notions and proving the above theorems.

\section{Bodies and Deficiencies}
\label{sec:geometric}
In this section we develop the geometric theory that captures the critical structural properties of SJA mechanisms and  use this to prove our main result, \cref{th:mainresultSJA6}, that shows their optimality. First we will need to establish some notation and formally define some notions.

For any positive integer $m$, an $m$-dimensional \emph{body} $A$ is any compact subset of the nonnegative orthant $A\subseteq\R_+^m$. We will denote its volume simply by $\card{A}\equiv\mu(A)$ (where $\mu$ is the standard $m$-dimensional Lebesgue measure). For any index set $J\subseteq[m]$, the \emph{projection} of $A$ with respect to the $J$ coordinates is defined as
$$
A_{[m]\setminus J}\equiv\sset{\vecc x_{-J}\fwh{\vecc x\in A}}
$$
and is the remaining body of $A$ if we ``delete'' coordinates $J$.
For any $r\in[m]$, index set $J\subseteq[m]$ with $\card{J}=m-r$ and $\vecc t\in\R_+^{m-r}$ we define the  \emph{slice} of $A$ above the point $\vecc t$ with respect to coordinates $J$ as
$$
\slice{A}{J}{\vecc t}\equiv\sset{\vecc x_{-J}\fwh{\vecc x\in A\;\land\; \vecc x_{J}=\vecc t}}.
$$
It is the remaining of the body $A$ if we fix a vector $\vecc t$ at coordinates $J$. 
The operations of projecting and slicing bodies commute with each other, that is, $\slice{A_{[m]\setminus I}}{J}{\vecc t}=(\slice{A}{J}{\vecc t})_{[m]\setminus I}$ for all disjoint sets of indices $I,J\subseteq[m]$ and $\card{J}$-dimensional vector $\vecc t$.

For any set of points $S\subseteq \R_+^m$ we denote their convex hull by $\chull(S)$ and for any vector $\vecc x$ we will denote by $\permuts(\vecc x)$ the set of all permutations of $\vecc x$. We will say that a body $A$ is \emph{downwards closed} if, for any point of $A$, all points below it are also contained in A: 
$\vecc y\in A$ for all $\vecc y\in\R_+^m$ with $\vecc y\leq\vecc x\in A$. Body $A$ will be called \emph{symmetric} if it contains all permutations of its elements: $\permuts(\vecc x)\subseteq A$ for all $\vecc x\in A$. If an $m$-dimensional body $A$ is symmetric then one can define its \emph{width} to be the length of its projection towards any axis: $w(A)\equiv \cards{A_{\ssets{j}}}$ for any $j\in[m]$. In a similar way, if $A\subseteq S$ we will say that $A$ is \emph{upwards closed} (with respect to $S$) if, for any $\vecc x\in A$, we have $\vecc y\in A$ for any $\vecc x\leq \vecc y\in S$. For any set of points $S\subseteq\R_+^m$, its \emph{downwards closure} is defined to be all points below it: 
$
\dclosure(S)=\sset{\vecc x\in \R_+^m\fwh{\exists \vecc y\in S:\vecc x\leq \vecc y}}.
$
Finally, we describe a property that will play a key role in the following:
\begin{definition}[p-closure]
\label{def:p-closure}
We will say that a body $A$ is \emph{p-closed} if it contains the convex hull of the permutations of any of its elements. Formally:
$
\chull(\permuts(\vecc x))\subseteq A
$
for all $\vecc x\in A$.
\end{definition}
Notice that any p-closed body must be symmetric (but not necessarily convex) and that any convex symmetric body is p-closed.

A useful, trivial to prove property of the deficiency function (see \cref{def:deficiency}) is that it is \emph{supermodular}:
\begin{lemma}\label{th:deficienciesunioninters}
For any bodies $A_1,A_2$,
$$
\delta(A_1\union A_2)+\delta(A_1\inters A_2)\geq \delta(A_1)+\delta(A_2).
$$
\end{lemma}

The next lemma tells us that ``leaving gaps'' between the points of bodies and the orthant's faces can only reduce the deficiency.
\begin{lemma}\label{th:eliasdown}
For any bodies $A,B$ such that $B\subseteq A$ and $A$ is downwards closed, there exists a downwards closed sub-body $\tilde B\subseteq A$ such that $\delta(\tilde B)\geq\delta(B)$. 
\end{lemma}
Instead of proving this lemma, we provide a stronger construction, given by the following \cref{lemma:pushingdown}. 

\begin{lemma}\label{lemma:pushingdown}
Let $\mathcal A_m$ be the set of $m$-dimensional bodies and $\mathcal
K_m\subseteq\mathcal A_m$ be the set of downwards closed ones. There is a mapping $\chi\,:\, {\cal A}_m \rightarrow {\cal K}_m$ such that for any 
$A,B\in\mathcal{A}_m$:
\begin{enumerate}
\item \label{item:chi1} $\card{\chi(A)}=\card{A}$ and, for every $J\subseteq [m]$, $\card{\chi(A)_J} \leq
 \card{A_J}$.
\item \label{item:chi2} $\chi(A) \cup \chi(B) \subseteq \chi(A \cup B)$. Equivalently,
 $A\subseteq B$ implies $\chi(A)\subseteq \chi(B)$.
\item \label{item:chi3}  
if $A\in\mathcal K_m$ then $\chi(A)=A$.
\end{enumerate}
\end{lemma}

It is straightforward to see how \cref{lemma:pushingdown} implies \cref{th:eliasdown}, by taking $\tilde B=\chi(B)$. Then, $\tilde B$ has the same volume as $B$ and (weakly) smaller projections (\hyperref[item:chi1]{Property~\ref*{item:chi1}}). This directly implies that $\delta(\tilde B)\geq\delta(B)$. It is also a subset of $A$ (by \hyperref[item:chi2]{Property~\ref*{item:chi2}}): $\tilde B=\chi(B)\subseteq\chi(A)=A$; the last equality follows from the fact that $A$ is already downwards closed and thus invariant under $\chi$ (\hyperref[item:chi3]{Property~\ref*{item:chi3}}). 

\begin{proof}[Proof of \cref{lemma:pushingdown}]
The lemma is proved by induction on $m$. For $m=1$ it is trivial: $\chi(A)$ is the interval starting at 0 with length equal to $|A|$.

Fix now a coordinate $j\in[m]$ and consider the $(m-1)$-dimensional slices $\slice{A}{\ssets{j}}{t}$ of $A$, ranging over $t$. Apply the lemma recursively (that is, use function~$\chi$ by the induction hypothesis from the previous dimension) to each such slice to obtain a body $A'$. Let $\chi'$ be this map from ${\cal A}_m$ to ${\cal A}_m$, i.e. $\chi'(A)=A'$. Notice that $A'$ may not be downwards closed, but we argue that $\chi'$ satisfies all three properties.

Indeed, for \hyperref[item:chi1]{Property~\ref*{item:chi1}}, we have two cases: If $j\in J$ then, by using \hyperref[item:chi1]{Property~\ref*{item:chi1}}, we get
$$\card{A_J'}=\int_t  \card{\slice{A_J'}{\ssets{j}}{t}}=\int_t \card{\left(\slice{A'}{\ssets{j}}{t}\right)_J} =\int_t \card{\left(\chi'(\slice{A}{\ssets{j}}{t})\right)_J} \leq\int_t \card{\left(\slice{A}{\ssets{j}}{t}\right)_J} =\int_t  \card{\slice{A_J}{\ssets{j}}{t}}=\card{A_J}.$$
In particular, the above holds with equality when $J=[m]$.
Otherwise, if $j\not\in J$, we have
\begin{equation*}
\card{A_J'}
= \card{ \left(\bigcup_t\slice{A'}{\ssets{j}}{t}\right)_J}
= \card{ \left(\bigcup_t\chi'\left(\slice{A}{\ssets{j}}{t}\right)\right)_J}
\leq \card{ \left(\chi'\left(\bigcup_t\slice{A}{\ssets{j}}{t}\right)\right)_J}
\leq \card{ \left(\bigcup_t\slice{A}{\ssets{j}}{t}\right)_J}
=\card{A_J},
\end{equation*}
the first inequality holding due to \hyperref[item:chi2]{Property~\ref*{item:chi2}}
and the second one due to the inequality at
\hyperref[item:chi1]{Property~\ref*{item:chi1}}.

Property 2 is also satisfied because, if $A\subseteq B$, then for every
$t$ it is $\slice{A}{\ssets{j}}{t}\subseteq \slice{B}{\ssets{j}}{t}$, and thus by induction $\chi'(\slice{A}{\ssets{j}}{t}) \subseteq 
\chi'(\slice{B}{\ssets{j}}{t})$, therefore
$$
\vecc x\in \chi'(A)\then \vecc x_{-j}\in \chi'(\slice{A}{\ssets{j}}{x_j}) \then \vecc x_{-j}\in \chi'(\slice{B}{\ssets{j}}{x_j})\then \vecc x\in \chi'(B).
$$

\hyperref[item:chi3]{Property~\ref*{item:chi3}} is satisfied, since if $A$ is already downwards closed, its slices are also downwards closed and, by induction, they will remain
unaffected by $\chi'$.

If $A$ is downwards closed with respect to some coordinate $i\in [m]$, then $\chi'(A)$ will remain closed downwards with respect to $i$: It is obvious by induction that $\chi'$ preserves downwards closure for every coordinate $i\neq j$. For coordinate $i=j$, it suffices to notice that downwards closure of $A$ is equivalent to $\slice{A}{\ssets{j}}{t}\subseteq \slice{A}{\ssets{j}}{t'}$ for all $t\geq t'$. Since $\chi'$ satisfies \hyperref[item:chi2]{Property~\ref*{item:chi2}}, the same holds for their images: $\chi'(\slice{A}{\ssets{j}}{t})\subseteq \chi'(\slice{A}{\ssets{j}}{t'})$.

Map $\chi'$ is not the desired map because if $A$ is not already downwards closed with respect to $j$, the result may not be downwards closed. However, we can select another coordinate $j'\neq j$ to create another map $\chi''$ similar to $\chi'$. Since $\chi''$ will satisfy all properties and preserve the downwards closure of coordinate $j'$, we conclude that $\chi=\chi''\circ\chi'$ has all the desired properties.
\end{proof}

The supermodularity of deficiency functions (\cref{th:deficienciesunioninters}) immediately implies that if bodies $A_1,A_2\subseteq S$ are of maximum deficiency (within $S$), then both their union and intersection are also of maximum deficiency. Based on this, the following can be shown:

\begin{lemma}\label{th:wlogsymmetric}
For any downwards closed and symmetric body $A$, there is a maximum volume sub-body of $A$ of maximum deficiency, which is also downwards closed and symmetric.
\end{lemma}

\begin{proof}
Let $B\subseteq A$ be of maximum deficiency. Then, by \cref{th:eliasdown}  there exists a downwards closed  $\tilde B\subseteq A$ such that $\delta(\tilde B)\geq \delta (B)$, and, due to the maximum deficiency of $B$, it must be that $\delta(\tilde B)= \delta (B)$. Now, let $\tilde B_1,\tilde B_2,\dots,\tilde B_{m!}$ be all possible permutations of the body $\tilde B$ (within the $m$-dimensional space) and take their union  $\hat B=\bigunion_{i=1}^{m!}\tilde B_i$. This new body $\hat B$ is clearly symmetric.
Also, because of the symmetry of $A$, all $\tilde B_i$ remain within $A$, so $\hat B\subseteq A$.

Now notice that all $B_i$'s have $\delta(\tilde B_i)=\delta(\tilde B)$, so they also have maximum deficiency within $A$. Remember that the deficiency function is supermodular (\cref{th:deficienciesunioninters}), so the union of maximum deficiency sets must also be of maximum deficiency. Thus, $\hat B$ is indeed of maximum deficiency.
Finally, it is not difficult to see that union preserves downwards closure and also, trivially, $\cards{\hat B}\geq\cards{\tilde B}$.
\end{proof}

The next lemma describes how global maximum deficiency implies also a kind of local one:
\begin{lemma}\label{th:counterslices}
Let $A\subseteq S$ be a maximum deficiency body (within $S$). Then, \emph{every} slice of $A$ must have nonnegative deficiency and must not contain subsets with higher deficiency.
\end{lemma}
\begin{proof}
To get to a contradiction, suppose that there exists such a slice $B=\slice{A}{J}{\vecc t}$ of $A$, such that $\delta(B)<0$. Then, let's remove the entire slice $B$ above $\vecc t$ from body $A$, to get a new body $A'$. This $(m-1)$-dimensional slice though is of measure $0$  in the larger $m$-dimensional space, so what we should really do is to remove an $\varepsilon$-neighborhood of $B$ (around $\vecc t$) within $A$, of ``parallel'' slices. This neighborhood has a volume of positive measure and is arbitrarily close to the slice\footnote{For ease of presentation, in the following we will use that procedure without making explicit mention to the underlying technicalities.}. This section removed from the body had the property of having volume strictly less than $k$ times its projections with respect to the coordinates not in $J$, i.e., the ``active'' coordinates in $B$ (because we are working close to $B$ for which $\delta(B)<0$). Regarding the other remaining projections with respect to the coordinates in $J$, by removing points they cannot possibly be increased. Since volumes have positive sign effect at the expression~\eqref{eq:deficiency} of the deficiency function, and projections have negative, we can deduce that the resulting body has strictly higher deficiency than $A$, which contradicts the maximum deficiency of $A$ within $S$.

The proof for subsets of the slice with higher deficiency is similar: replace the entire slice with its subset of higher deficiency, and the total deficiency must increase.
\end{proof}

As a consequence of \cref{th:counterslices} we get the following properties of maximum deficiency sub-bodies, which imply that these bodies must be ``large enough'' (\cref{lemma:sizeboundcounter,th:counterchainwidths}) and also demonstrate some kind of ``symmetric convexity'' (p-closure \cref{th:pclosure}, \cref{def:p-closure}). But first we will need an inequality that brings together volumes and projections of bodies, due to \citet{Loomis:1949ul}. An easy proof of this can be found in~\citep{Allan:1983pd}. 

\begin{lemma}[Loomis--Whitney]\label{th:loomis} For any $m$-dimensional body $A$,
$$\card{A}^{m-1}\leq \prod_{j=1}^m\card{A_{[m]\setminus \ssets{j}}}.$$
\end{lemma}

\begin{lemma}\label{lemma:sizeboundcounter}
Let $A\neq\emptyset$ be an $m$-dimensional body with nonnegative  $k$-deficiency. Then
$$
\card{A}\geq (km)^m.
$$
As a consequence, if $A$ is also symmetric and downwards closed, its width must be at least
$$
w(A)\geq km.
$$
\end{lemma}

\begin{proof}
Since $\delta_k(A)\geq 0$,
we know that
$
\card{A}\geq k\sum_{j=1}^m\card{A_{[m]\setminus\ssets{j}}},
$
or equivalently
\begin{equation}\label{eq:lb1}
\sum_{j=1}^m\card{A_{[m]\setminus \ssets{j}}}\leq\frac{\card A}{k}.
\end{equation}
Also, by the Loomis--Whitney inequality (\cref{th:loomis}),
$
\card{A}^{m-1}\leq \prod_{j=1}^m\card{A_{[m]\setminus \ssets{j}}}
$;
so, by using the arithmetic-geometric means inequality we can derive that
$\card{A}^{m-1}\leq\left(\frac{1}{m}\sum_{j=1}^m\card{A_{[m]\setminus \ssets{j}}}\right)^m$,
or equivalently
\begin{equation}\label{eq:lb2}
\sum_{j=1}^m\card{A_{[m]\setminus \ssets{j}}}\geq m\card{A}^{\frac{m-1}{m}}.
\end{equation}

Combining~\eqref{eq:lb1} and~\eqref{eq:lb2} we get
$
m\card{A}^{\frac{m-1}{m}}\leq\frac{\card A}{k},
$
which completes the proof of the lemma (since $\card{A}\neq 0$). The inequality involving the body's width follows immediately from the observation that every symmetric and downwards closed body $A$ is included in the $m$-dimensional hypercube with edge length $w(A)$.
\end{proof}

\begin{lemma}\label{th:counterchainwidths}
If $A$ is a nonempty, symmetric, downwards closed body with nonnegative
$k$-deficiency then it must contain the point $(k, 2 k,\dots, m k)$. More generally, it must contain the point $(k, 2 k,\dots,(m-1)k,w(A))$.
\end{lemma}

\begin{proof}
We will recursively utilize \cref{th:counterslices,lemma:sizeboundcounter} to show that points $$(mk,\vecc 0_{m-1}),(mk,(m-1)k, \vecc 0_{m-2}),  \dots, (mk,(m-1)k,\dots,k)$$ belong to $\hat A$, where $\hat A$ is a symmetric, downwards closed sub-body of $A$ of maximum deficiency (see \cref{th:wlogsymmetric}).  By \cref{lemma:sizeboundcounter} it must be that that $w(\hat A)\geq mk$, thus $(mk,\vecc 0_{m-1})\in \hat A$ by downwards closure. For the next dimension, consider the slice $\slice{\hat A}{\ssets{j}}{mk}$ (for some $j\in[m]$). It is $(m-1)$-dimensional, of nonnegative deficiency by \cref{th:counterslices}, and so it must have width at least $(m-1)k$ (\cref{lemma:sizeboundcounter}). That means that point $(mk,(m-1)k,\vecc 0_{m-2})$ must be in $\hat A$. We can continue like this all the way down to single-dimensional lines.
\end{proof}

\begin{lemma}[p-closure]\label{th:pclosure}
Let $A\subseteq S$ be a maximum volume sub-body of $S$ of maximum deficiency and
let $S$ be p-closed and downwards closed. Then \emph{every} slice of $A$ 
(including $A$ itself) must be p-closed (see \cref{def:p-closure}).
\end{lemma}

\begin{proof}
Without loss (by \cref{th:wlogsymmetric}) $A$ can be assumed to be symmetric and downwards closed.
We need to prove that, for any $r\in [m]$ ($r$ is the dimension of the slice) and any $r$-dimensional vector $\vecc x$ and $\vecc z\in\chull(\permuts(\vecc x))$ in the convex hull of its permutations, 
$$
\text{for all}\;\;\vecc t:\qquad(\vecc x,\vecc t)\in A\quad\then\quad (\vecc z,\vecc t)\in A.
$$

The proof is by induction on $r$. For the base case of $r=1$, it is $\chull(\permuts(\vecc x))=\ssets{\vecc x}$ so the proposition follows trivially.
For the induction step, assume the proposition is true for some $r\leq m-1$ and we will prove it for $r+1$. So, take $(r+1)$-dimensional vectors $\vecc x$ and $\vecc z$ such that $\vecc z\in\chull(\permuts(\vecc x))$ and fix some $\vecc t\in\R_+^{m-r-1}$. To complete the proof we need to show that the slice of $A$ above $\vecc x$, with respect to the first $r+1$ coordinates, is included within the one above $\vecc z$, i.e.~$\slice{A}{[r+1]}{\vecc x}\subseteq \slice{A}{[r+1]}{\vecc z}$. For simplicity, let's abuse notation for the remaining of this proof and just use $A_{\vecc x}$ and $A_{\vecc z}$ for these slices.

So, to arrive at a contradiction, let's assume that, $A_\vecc x\setminus A_\vecc z\neq\emptyset$. First notice that since $A_{\vecc x}\inters A_{\vecc z}\subseteq A_{\vecc x}$ and $A_\vecc x$ is a slice of a maximum deficiency body, by \cref{th:counterslices} it must be that $\delta(A_{\vecc x}\inters A_{\vecc z})\leq \delta(A_{\vecc x})$. So, by the supermodularity of deficiencies (\cref{th:deficienciesunioninters}) we get that
$$
\delta(A_{\vecc x}\union A_{\vecc z})\geq \delta(A_{\vecc z}).
$$
This means that if we replace (an $\varepsilon$-neighborhood around $\vecc z$ of) slice $A_{\vecc z}$ by its superset $A_{\vecc x}\union A_{\vecc z}$ and we can also show that no new projections are created with respect to the first $r+1$ coordinates, then the overall deficiency of the body would not decrease and its volume would increase strictly (since we have assumed that $A_\vecc x\setminus A_\vecc z\neq\emptyset$), which is a contradiction to the maximum deficiency of $A$ within $S$. Notice a subtle point here: How do we know that this extension can fit within $S$ above point $\vecc z$? It does, because we have assumed $S$ to be p-closed and the new elements added are convex combinations of permutations of elements already known to be in $S$. The remainder of the proof is dedicated to proving that this extension indeed does not create new projections with respect to the first $r+1$ coordinates. 

Without loss, due to symmetry, we can take $x_1\leq x_2\leq\dots\leq x_{r+1}$. We argue that, if we remove any  one of the coordinates of the vector $\vecc z$, it can be dominated by a convex combination of permutations of the vector $\vecc x_{-1}$ (i.e., the vector $\vecc x$ if we remove its \emph{smallest} coordinate). To see that, remember that $\vecc z$ is at the convex hull of the permutations of $\vecc x$, so there exist nonnegative real parameters $\ssets{\xi_\pi}$ such that 
$$
\vecc z=\sum_{\vecc\pi\in\permuts(\vecc x)}\xi_\pi\vecc \pi\quad\text{and}\quad \sum_{\vecc\pi\in\permuts(\vecc x)}\xi_{\vecc \pi}=1.
$$
But that means that 
\begin{equation}\label{eq:pclosure3}
\vecc z_{-j}=\sum_{\vecc\pi\in\permuts(\vecc x)}\xi_{\vecc\pi}\vecc \pi_{-j}
\end{equation}
for any coordinate $j$.

Let's define a transformation $\phi$ over all vectors $\sset{\vecc \pi_{-j}\fwh{\pi\in\permuts(\vecc x)\;\;\text{and}\;\;j\in[r+1]}}$  such that $\phi(\vecc \pi_{-j})=\vecc \pi_{-j}$ if the $j$-th coordinate removed from $\vecc \pi$ to get $\vecc \pi_{-j}$ was $x_1$, and otherwise $\phi(\vecc \pi_{-j})$ is the $r$-dimensional vector that we get if we replace $x_1$ in $\vecc \pi_{-j}$ by the coordinate $\vecc \pi_j$ that was removed. It follows that for all  $j$

$$
\vecc \pi_{-j}\leq \phi(\vecc \pi_{-j})\quad \text{and}\quad \phi(\vecc \pi_{-j})\in\permuts(\vecc x_{-1}),
$$
so by~\eqref{eq:pclosure3},
$$
\vecc z_{-j}\leq \sum_{\vecc\pi\in\permuts(\vecc x)}\xi_{\vecc \pi}\vecc \phi(\pi_{-j})\in\chull(\permuts(\vecc x_{-1})).
$$ 
By the induction hypothesis and downwards closure for $A$ it can be deduced that 
$$
(\vecc x_{-1},0,\vecc t)\in A\;\;\then\;\; (\vecc z_{-j},0,\vecc t)\in A\quad\text{for all}\;\;j\in[r+1].
$$
Thus in particular for every $\vecc t \in A_\vecc x$, due to symmetry of $A$, we have that 
$
\left((\vecc z_{-j},0),\vecc t\right)\in A,
$
which means that indeed every projection of $(\vecc z,\vecc t)$ with respect to a coordinate in $[r+1]$ was already included in $A$.
\end{proof}

\section{Decomposition of SJA into SIM-bodies}\label{sec:SIM}

\subsection{SIM-bodies}\label{sec:SIMbodies2}

Remember that in \cref{def:sim-body} we introduced the notion of a SIM-body $\Lambda(\alpha_1,\ldots,\alpha_r)$: for parameters  $\alpha_1\leq \cdots\leq \alpha_r$, it is the set of all vectors $\vecc x\in\R^r_+$ satisfying conditions $ \sum_{j\in J}  x_j \leq \sum_{j=r-|J|+1}^{r}\alpha_j$ for all $J\subseteq [r]$.
 
The geometry of the allocation space of the SJA mechanisms (see \cref{fig:SJA2and3}) naturally gives rise to this family of bodies. Their importance and connection with the structure of the SJA mechanisms will become evident in \cref{sec:proof-no-positive-deficiency}, where we prove \cref{lemma:SIMallocspaceisom}. The intuition behind the naming becomes obvious if one looks at \cref{fig:2dimSIM}. By the way SIM-bodies are defined, one can immediately see that they are downwards closed, symmetric, and convex polytopes. Thus, they are also p-closed. Each one of its faces corresponds to a defining hyperplane $$\sum_{j\in J}x_j=\alpha_{r+1-\card{J}}+\dots+\alpha_{r}$$ 
for some $J\subseteq [r]$ or, of course, to a side of the $r$-dimensional positive orthant $\R^m_+$.

SIM-bodies demonstrate some inherently recursive and symmetric properties, captured by the following lemma. They are made clear in \cref{fig:dimSIM}. 

\begin{lemma}
\label{lemma:SIMbodiesprops2}
For any SIM-body $\Lambda= \Lambda(\alpha_1,\ldots,\alpha_r)$:
\begin{enumerate}
\item \label{prop:SIM1} $w(\Lambda)=\alpha_r$
\item \label{prop:SIM2} $\Lambda=\dclosure(\chull(\permuts(\alpha_1,\ldots,\alpha_r)))$
\item \label{prop:SIM3} $\slice{\Lambda}{\ssets{j}}{\alpha_r}=\Lambda(\alpha_1,\ldots,\alpha_{r-1})$ for any $j\in[r]$
\item \label{prop:SIM6} $\Lambda_{[r]\setminus\ssets{j}}=\Lambda(\alpha_2,\ldots,\alpha_{r})$ for any $j\in[r]$
\item \label{prop:SIM4} $\delta_{q\cdot k}(q\cdot\Lambda)=q^r\cdot \delta_k(\Lambda)$ for any $q,k>0$
\end{enumerate}
\end{lemma}
\begin{proof}
\hyperref[prop:SIM1]{Property~\ref*{prop:SIM1}} is trivial: by the definition of SIM-bodies~\eqref{eq:sim-body}, a point $(x,\vecc 0_{r-1})\in \Lambda$ if and only if $x\leq \alpha_r \land \dots \land x\leq \alpha_1+\dots+\alpha_r$, i.e., $x\leq \alpha_r$.

 For \hyperref[prop:SIM2]{Property~\ref*{prop:SIM2}}, let $E$ be the set of the
 extreme points of the  polytope $\Lambda$. It is convex, thus $\Lambda=\chull
 (E)$. But it is also downwards closed, so we can just focus on the extreme
 points $\overline E\subseteq E$ that belong to the ``full'' facet of the
 hyperplane $x_1+\dots x_r=\alpha_1+\dots+\alpha_r$, since the entire polytope
 can be recovered as the downwards closure $\Lambda=\dclosure(\chull(\overline E))$. By taking intersections with the other hyperplanes and keeping in mind that the $\alpha_j$'s are nondecreasing, we get that these extreme points in $\overline E$ are $(\alpha_1,\alpha_2,\dots,\alpha_r)$ and all its permutations. So, we can recover the entire SIM-body as $\Lambda=\dclosure(\chull(\permuts(\alpha_1 ,\dots,\alpha_r )))$.

For \hyperref[prop:SIM3]{Property~\ref*{prop:SIM3}}, notice that an $(r-1)$-dimensional vector $\vecc x$ belongs in the slice $\slice{\Lambda}{\ssets{j}}{\alpha_r}$ if and only if $(\vecc x,\alpha_r)\in \Lambda$, which by using~\eqref{eq:sim-body} is equivalent to 
$$
\bigland_{J\subseteq [r-1]} \left(\sum_{i\in J}x_i\leq \alpha_{r+1-\card{J}}+\dots+\alpha_{r}\right)
\quad\text{and}\quad 
\bigland_{J\subseteq [r-1]}\left(\alpha_r+\sum_{i\in J}x_i\leq \alpha_{r-\card
{J}}+\dots+\alpha_{r}\right).
$$
The second set of conditions can be rewritten simply as
\begin{equation}
\label{eq:helperghk} 
\bigland_{J\subseteq [r-1]}\;\sum_{i\in J}x_i\leq \alpha_{r-\card{J}}+\dots+\alpha_{r-1},
\end{equation}
which makes the first set of constraints redundant since 
$
\alpha_{r-\card{J}}+\dots+\alpha_{r-1}\leq \alpha_{r+1-\card{J}}+\dots+\alpha_{r}
$
from the monotonicity of the sequence of $\alpha_r$'s. The constraints~\eqref{eq:helperghk} that we are left with, exactly define $\Lambda(\alpha_1,\dots,\alpha_{r-1})$ (see~\eqref{eq:sim-body}).

\hyperref[prop:SIM6]{Property~\ref*{prop:SIM6}} can be shown in a very similar way: due to downwards closure, any projection $\Lambda_{[r]\setminus\sset{j}}$ is just the slice $\slice{\Lambda}{\sset{j}}{0}$.

Finally, \hyperref[prop:SIM4]{Property~\ref*{prop:SIM4}} is a result of scaling: $q\cdot\Lambda$ and $\Lambda$ are similar by a scaling factor of $q$, so the ratio of their volumes is $q^{r}$ and the ratio of their projections is $q^{r-1}$. In formula~\eqref{eq:deficiency} that defines deficiencies, the volumes of the projections are also multiplied by the parameter $k$ of the deficiency, resulting in an overall ratio of $q^{r}$ between the two deficiencies.
\end{proof}

\begin{figure}[t]
\centering
\begin{subfigure}{0.485\textwidth}
\includegraphics[width=1\textwidth]{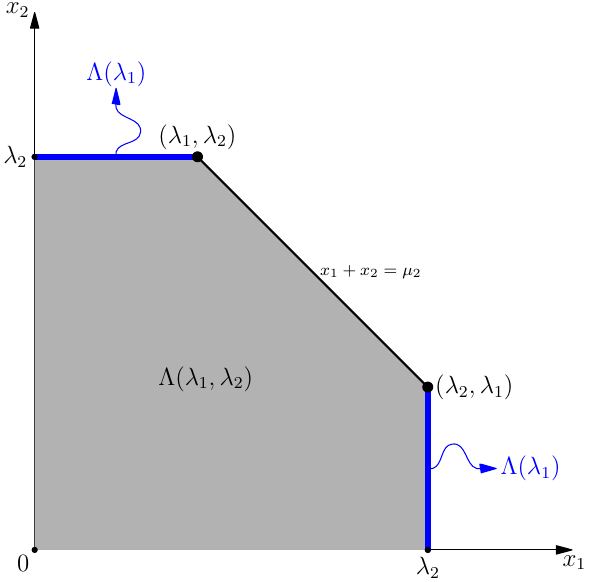}
\caption{\footnotesize The $2$-dimensional SIM-body $\Lambda(\lambda_1,\lambda_2)$}
\label{fig:2dimSIM}
\end{subfigure}
~
\begin{subfigure}{0.485\textwidth}
\includegraphics[width=1\textwidth]{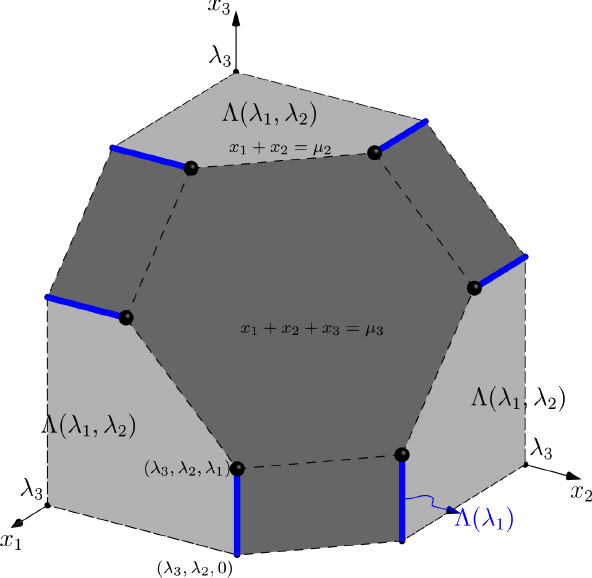}
\caption{\footnotesize The $3$-dimensional SIM-body $\Lambda(\lambda_1,\lambda_2,\lambda_3)$}
\label{fig:3dimSIM}
\end{subfigure}
\caption{\footnotesize SIM-bodies for dimensions $m=2,3$. Notice the recursive nature of these constructions: a SIM-body encodes in it the SIM-bodies of lower dimensions as extreme slices (\hyperref[prop:SIM3]{Property~\ref*{prop:SIM3}} of \cref{lemma:SIMbodiesprops2}). In this figure, these $1$-dimensional critical  bodies are denoted by thick lines (blue in the color version of the paper) and the $2$-dimensional ones in light gray.}
\label{fig:dimSIM}
\end{figure}

\subsection{Decomposition of SJA}
\label{sec:proof-no-positive-deficiency}

In this section we bring together all the necessary elements needed to prove \cref{lemma:SIMallocspaceisom}. We study the structure of the allocation space of SJA that reveals an elegant decomposition which demonstrates that the SIM-bodies essentially act as building blocks for SJA.

First of all, we need to get a closer look at SJA's payments and demonstrate
some of their interesting characteristics. The way in which the SJA payments are
constructed makes them satisfy a kind of ``contraction'' property:
\begin{lemma}
\label{lemma:decrdiffs}
The prices of the SJA mechanism have nonincreasing differences, i.e.,
$$p^{(m)}_r-p^{(m)}_{r-1}\leq p^{(m)}_{r-1}-p^{(m)}_{r-2}$$
for all $r=2,\dots,m$.
\end{lemma}

\begin{proof}
Fix some dimension $m$ and assume that we have computed prices of SJA up to
$p_1,p_2,\dots, p_{r-1}$ for some $2\leq r\leq m$. First we will show that 
\begin{equation}
\label{eq:helperlem5}
(r-1)p_r\leq rp_{r-1},
\end{equation}
i.e., that the price $p_r$ must be in $[0,\frac{r}{r-1}p_{r-1}]$. We will do that
by showing that otherwise this price would be redundant, in the sense that for any $p_r> \frac{r}{r-1}p_{r-1}$ the sub-body of $I^r$ defined by
$$
\bigland_{J\subseteq [r]}\;\;\sum_{j\in J}x_j<p_{\card{J}},
$$
and whose volume must be exactly $1-rk$ in \cref{def:SJAmechanism}, would remain unchanged  and equal to the one defined by 
\begin{equation}
\label{eq:helper589}
\bigland_{\substack{J\subseteq [r]\\ \card{J}\leq r-1}}\;\sum_{j\in J}x_j<p_{\card{J}}.
\end{equation}

Indeed, the body defined from~\eqref{eq:helper589} is a downwards closed, symmetric convex polytope and for the newly inserted hyperplane $x_1+\dots+x_r=p_r$ to have any effect on it, i.e., to have a nonempty intersection with it, it must be that this hyperplane's ``symmetric point'' $(p_r/r,\dots,p_r/r)$ belongs already to the interior of the body in~\eqref{eq:helper589} (this is due to the symmetry and convexity of the body). So, this point must satisfy the $(r-1)$-dimensional condition $x_1+\dots+x_{r-1}\leq p_{r-1}$, thus $(r-1)(p_r/r)\leq p_{r-1}$ which is exactly property~\eqref{eq:helperlem5}.

To show that $p_{r}-p_{r-1}\leq p_{r-1}-p_{r-2}$ for all $2\leq r\leq m$, or equivalently $p_{r}\leq 2p_{r-1}-p_{r-2}$, by~\eqref{eq:helperlem5} it is enough to show that $\frac{r}{r-1}p_{r-1}\leq 2p_{r-1}-p_{r-2} $. But this is equivalent to $(r-1)p_{r-2}\leq(r-2)p_{r-1}$ which we know holds, also from~\eqref{eq:helperlem5}.
\end{proof}

\paragraph{Normalized payments} By the procedure of defining SJA payments (\cref{def:SJAmechanism}), it can be the case that price $p_r$ is smaller than $p_{r-1}$, i.e., $p_r\in[p_{l},p_{l+1}]$ for some $l\leq {r-2}$. This is perfectly acceptable, and it just means that essentially we render older prices that are above $p_r$ redundant, in the sense that setting $p_{j}\gets p_r$ for all $j<r$ with $p_j\geq p_r$ would not have an effect on the sub-body $\bigland_{J\subseteq [r]}\sum_{j\in J}x_j<p_{\card{J}}$ of $I^r$ used in the definition of SJA in~\eqref{eq:sja}. This because $x_1+\dots+x_r\leq p_r\then x_1+\dots+x_j\leq p_j$ (since $j<r$ and $p_r\leq p_j$), so old conditions $x_1+\dots+x_j\leq p_j$ have become useless.
In particular, notice how this is the case for the full-bundle price $p_m^{(m)}$
when $m=5,6$: from \cref{eq:prices_numeric_1_4,eq:prices_numeric_5_6} we see
that indeed 
$p^{(5)}_5\approx 1.9856 < 2.0005\approx p^{(5)}_4$ 
and 
$p^{(6)}_5\approx 2.3774 < 2.4165\approx p^{(6)}_6$. 
This means that no bundle of $m-1$ items is ever going to be
sold under the SJA mechanism for $m=5$ or $m=6$:
\begin{equation}
\label{eq:special_empty_allocations_5_6}
U^{(5)}_{[4]}=U^{(6)}_{[5]}=\emptyset.
\end{equation}

Furthermore, by the nonincreasing differences property of the SJA payments 
(\cref{lemma:decrdiffs}), every new payment after $r$ will continue to fall below the previous one. So, at the end the situation will be in the form of
\begin{equation}
\label{eq:beforenormpayments} 
p_1\leq\dots\leq p_{l}\leq p_m\leq \dots
\end{equation}
for some $l<m$ and, as we discussed above, there will be absolutely no effect on the mechanism if we update all older payments that have ended up above $p_m$ to ``collapse'' to $p_m$, i.e.,
\begin{equation}
\label{eq:afternormpayments}  
p_1\leq\dots\leq p_{l}\leq p_m=p_{m-1}=p_{m-2}=\dots=p_{l+1}.
\end{equation}
Rigorously, we redefine
$$
p^{(m)}_{j}\gets p^{(m)}_m\quad\text{for all}\;\; j\in[m-1]\;\;\text{with}\;\;p_j\geq p_m.
$$
While this \emph{normalization} has no effect on the SJA mechanism itself, it makes sure that payments are now given in a \emph{nondecreasing} order, which is an elegant property that will simplify our exposition later on.

An important observation is that this normalization of payments does not break the property of the nonincreasing differences of the payments of SJA, i.e.\ \cref{lemma:decrdiffs} continues to hold: having a look at the transition before and after the normalization process from~\eqref{eq:beforenormpayments} to~\eqref{eq:afternormpayments} we see that all the differences up to the $l$-th payment remain unchanged, $p_{l+1}-p_{l}$ can only decrease and all differences above the $(l+1)$-th payment have just collapsed to $0$.

From now on and for the remaining of this paper we will assume that SJA payments
are normalized. The only difference that this makes, for up to $m=6$ dimensions,
to the values of the payments we have already computed at \cref
{eq:prices_numeric_1_4,eq:prices_numeric_5_6} in \cref
{sec:SJAdef} is that
for $m=5,6$ we have that
$$ p_{4}^{(5)}\gets p_{5}^{(5)}\quad\text{and}\quad p_{5}^{(6)}\gets p_{6}^{(6)}, $$
which gives by~\eqref{eq:transfpm} that also the $\mu^{(m)}_r$ parameters are updated to $\mu_{m-1}^{(m)}\gets\mu_{m}^{(m)}-(m+1)$:
$$
\mu^{(5)}_4\approx 12.0865\quad\quad \mu^{(6)}_5\approx 18.3585.
$$

Recall the definition $\lambda_r\equiv\mu_r-\mu_{r-1}$ from \eqref{eq:lambdasdef}. These are the critical parameters of the SIM-bodies used in all the key theorems for the optimality of SJA.
Equation \eqref{eq:lambdasdef} is equivalent to saying that $\mu_r=\lambda_1+\dots+\lambda_r$. Taking the $\mu_r^{(m)}$ values into account (see~\eqref{eq:transfpm}) the $\lambda_r^{(m)}$'s for up to $m=6$ items are, for $m\leq 4$,
\begin{equation}
\label{eq:lambdas_def1}
\lambda_1 = 1	\qquad	\lambda_2 = 1+\sqrt{2}	\qquad	\lambda_3 \approx 3.6830	\qquad	\lambda_4 \approx 4.9000
\end{equation}
and for $m=5,6$ the only modifications are
\begin{equation}
\label{eq:lambdas_def2}
\lambda_4^{(5)} \approx 4.9894	\qquad	\lambda_5^{(5)} = 6 \qquad	\lambda_5^{(6)}
\approx 6.3613	\qquad	\lambda_6^{(6)} = 7.
\end{equation}

The nonincreasing differences property of the SJA payments makes these parameters monotonic:
\begin{lemma}
\label{lemma:decrdiffs2}
The $\lambda_r^{(m)}$ parameters are nondecreasing and upper-bounded by $m+1$, i.e.,
$$
\lambda_{r-1}^{(m)}\leq \lambda_r^{(m)}\leq m+1,
$$
for all $r=2,\dots,m$.
\end{lemma}

\begin{proof}
Using the transformations~\eqref{eq:transfpm} and~\eqref{eq:lambdasdef} we have
$$
 p_r-  p_{r-1}\leq   p_{r-1}-  p_{r-2}\then \mu_{r-1}-\mu_{r-2}\leq \mu_{r}- \mu_{r-1} \then \lambda_{r-1}\leq \lambda_r
$$
and 
$$
 p_{r-1}\leq  p_{r}\then \mu_{r}-\mu_{r-1}\leq m+1\then \lambda_{r}\leq m+1,
$$
which concludes the proof since the SJA payments are nondecreasing with nonincreasing differences (\cref{lemma:decrdiffs}).
\end{proof}

An algebraic manipulation of~\eqref{eq:def-U_A}, using the nonincreasing
differences property of the SJA payments, can give us the following
characterization: 
\begin{lemma}
\label{lemma:allocsubspacealgebra}
For any subset of items $J\subseteq [m]$,
$$
U_{J}^{(m)} =\sset{\vecc x\in I^m \fwh{\bigland_{L\subseteq J}\;\;\sum_{j\in L}x_j\geq p_{\card{J}}^{(m)}-p_{\card{J}-\card{L}}^{(m)}\bigland_{L\subseteq [m]\setminus J}\;\;\sum_{j\in L}x_j\leq p_{\card{J}+\card{L}}^{(m)}-p_{\card{J}}^{(m)}}}.
$$
\end{lemma}
\begin{proof}
Fix some positive integer $m$, $J\subseteq [m]$ and an arbitrary $\vecc x\in
I^m$. We
need to show that $\vecc x$ satisfies the constraints in the description of set
$U_J^{(m)}$ at the statement of \cref{lemma:allocsubspacealgebra} if and only if
it satisfies the constraints in~\eqref{eq:def-U_A}. To be more precise,
and after moving all
$x_j$'s in the constraints of~\eqref{eq:def-U_A} at the left side of the
inequalities and deleting the ones that cancel out, we need to show that
\begin{equation}
\label{eq:allocsubspacealgebra_1}
\sum_{j\in J\setminus L} x_j - \sum_{j\in L\setminus J} x_j \geq p_{\card{J}}-p_
{\card
{L}}
\quad\text{for all}\;\; L\subseteq [m]
\end{equation}
if and only if 
\begin{equation}
\label{eq:allocsubspacealgebra_2}
\sum_{j\in L_1} x_j \geq p_{\card{J}}-p_
{\card
{J}-\card{L_1}}
\quad\text{for all}\;\; L_1\subseteq J
\end{equation}
and
\begin{equation}
\label{eq:allocsubspacealgebra_3}
\sum_{j\in L_2} x_j \leq p_{\card{J}+\card{L_2}}-p_
{\card
{J}}
\quad\text{for all}\;\; L_2\subseteq \bar{J},
\end{equation}
where for simplicity we drop the superscript $(m)$ from the prices and denote
$\bar{J}=[m]\setminus J$.

Indeed, first assume that $\vecc x$ satisfies \eqref{eq:allocsubspacealgebra_1}
and pick any $L_1\subseteq J$, $L_2\subseteq \bar{J}$. Then, since
$J\setminus (J\setminus L_1)=L_1$ and $(J\setminus L_1)\setminus J=\emptyset$,
by using $L\gets J\setminus L_1$ in \eqref{eq:allocsubspacealgebra_1} we get
$$
\sum_{j\in L_1} x_j \geq p_{\card{J}}-p_
{\card
{J\setminus L_1}}
= p_{\card{J}}-p_
{\card
{J}-\card{L_1}},
$$
which proves that $\vecc x$ satisfies \eqref{eq:allocsubspacealgebra_2}. In a
similar way, since $J\setminus(J \union L_2)=\emptyset$ and $(J\union
L_2)\setminus J=L_2$, if we use $L\gets J\union L_2$ in \eqref
{eq:allocsubspacealgebra_1}
we get
$$
-\sum_{j\in L_2} x_j \geq p_{\card{J}}-p_
{\card
{J\union L_2}}=p_{\card{J}}- p_{\card{J}+\card{L_2}},
$$
which is the same as \eqref{eq:allocsubspacealgebra_3}.

For the opposite direction, assume now that $\vecc x$ satisfies \eqref
{eq:allocsubspacealgebra_2} and \eqref{eq:allocsubspacealgebra_3}, and pick any
$L\subseteq [m]$. Since $J\setminus L\subseteq J$ and $L\setminus
J\subseteq\bar{J}$, if we use $L_1\gets J\setminus L$ and $L_2\gets L\setminus
J$
in \eqref
{eq:allocsubspacealgebra_2} and \eqref{eq:allocsubspacealgebra_3}, respectively,
we get
\begin{align*}
\sum_{j\in J\setminus L} x_j &\geq p_{\card{J}}-p_
{\card
{J}-\card{J\setminus L}},\\
\sum_{j\in L\setminus J} x_j & \leq p_{\card{J}+\card{L\setminus J}}-p_{\card
{J}}.
\end{align*}
By subtracting these inequalities and taking into consideration that
$\card{J}-\card{J\setminus L} = \card{J\inters L}$
and
$\card{J}+\card{L\setminus J}=\card{J\union L}$
we have
$$
\sum_{j\in J\setminus L} x_j - \sum_{j\in L\setminus J} x_j
\geq
2p_{\card{J}}-p_
{\card
{J\inters L}}- p_{\card{J\union L}}.
$$
So, in order to show that \eqref{eq:allocsubspacealgebra_1} holds and conclude
the proof of the lemma, it is enough to show that $p_{\card{J}}-p_
{\card
{J\inters L}}- p_{\card{J\union L}}\geq - p_{\card{L}}$, or equivalently that
$$
p_{\card{J\union L}} - p_{\card{J}} \leq p_{\card{L}} - p_{\card{J\inters L}}.
$$
But since $\card{J\union L}-\card{J}=\card{L}-\card{J\inters L}$ (they are both
equal to $\card{L\setminus J}$) and $\card{J\union L}\geq \card{L}$, the above
inequality indeed holds due to the nonincreasing differences property of the
SJA payments (\cref{lemma:decrdiffs}).
\end{proof}

Notice here that, due to symmetry, every slice $\slice{U^{(m)}_J}{-J}{\vecc t}$ with $\card{J}=r\leq m$ is isomorphic to $\slice{U^{(m)}_{[r]}}{[r+1...m]}{\vecc t}$ and so, from the characterization in \cref{lemma:allocsubspacealgebra}, this slice is \emph{invariant} with respect to the specific value of the (($m-r$)-dimensional) vector $\vecc t$. In particular, if it's nonempty, then
\begin{equation}
\label{eq:sliceszeroenough}
\slice{U^{(m)}_J}{-J}{\vecc t}=\slice{U^{(m)}_J}{-J}{\vecc 0_{m-\card{J}}}.
\end{equation}
The following lemma essentially gives an alternative definition of SJA, in terms of the deficiencies of its allocation components $U_J^{(m)}$. In particular, it requires every $\card{J}$-dimensional slice of any subdomain $U_J^{(m)}$ to have zero deficiency:

\begin{lemma}
\label{lemma:zeroslicedefSJA}
Every slice
  $\slice{U^{(m)}_J}{-J}{\vecc t}$ of SJA has zero $k$-deficiency, where
  $k=\frac{1}{m+1}$.
\end{lemma}

\begin{proof}[Proof sketch]
Fix some dimension $m$ and let $k=1/(m+1)$. By the definition of SJA~\eqref{eq:sja}, the domain $\overline U_\emptyset$ where at least one item is sold must have volume $m/(m+1)$: the probability of selling at least an item is $mk=m/(m+1)$ which corresponds to the volume of this domain because the valuations' space is the unit cube $I^m$. Every projection $(\overline U_\emptyset)_{\ssets{j}}$ of this body towards any coordinate $j\in[m]$ has volume~$1$: it is the $(m-1)$-dimensional side of the cube; just set the valuation of item $j$ to $x_j=1$ and trivially notice that, no matter what the remaining valuations $\vecc x_{-j}\in I^{m-1}$ are, at least one item is being sold by SJA, namely item $j$, since $x_j=1\geq p_1$. Bringing the above together, this means that the $k$-deficiency of $\overline U_\emptyset$ is $m/(m+1)- k\cdot m\cdot 1=0$.

This valuations subdomain $\overline U_\emptyset$ where at least one item is sold can be decomposed in its various components $U_J$, where $\emptyset\neq J\subseteq [m]$. Its volume is just the sum of the volumes of these components. Also, its projections (i.e. the sides of the unit cube $I^m$) can be covered by taking the projection of any such component $U_J$ with respect to its ``active'' coordinates in $J$. This tells us that the deficiency of the entire body $\overline U_\emptyset$ is essentially reduced to the sum of the deficiencies of its subdomains. But this body has zero $k$-deficiency, so all its components must also have zero deficiencies (by using an inductive argument).

A complete, formal proof of this characterization can be found in \cref{append:fullproofzeroslicedef}.
\end{proof}

Now we are ready to prove \cref{lemma:SIMallocspaceisom}, which makes rigorous the correspondence between the various components $U_J^{(m)}$ of the allocation space of SJA and SIM-bodies. It is the motivation behind introducing SIM-bodies in the first place. Essentially, the entire allocation space of SJA is made up by slices of SIM-bodies: 

\begin{lemmanonum}[\cref{lemma:SIMallocspaceisom}]
 Every nonempty slice $\slice{U_J^{(m)}}{-J}{\vecc t}$ is isomorphic to the
  SIM-body $k\cdot \Lambda(\lambda_1^{(m)},\ldots,\lambda_{|J|}^{(m)})$, where $k=\frac{1}{m+1}$. 
\end{lemmanonum}

\begin{proof}
Let $\card{J}=r$. Then, due to symmetry, the slice $\slice{U_J}{-J}{\vecc t}$ is isomorphic to $\slice{U_{[r]}}{[r+1...m]}{\vecc t}$. An $r$-dimensional vector $\vecc y$ belongs to this slice if and only if $(\vecc y,\vecc t)\in U_{[r]}$, which by \cref{lemma:allocsubspacealgebra} means that $\vecc y\in I^r$ and
$$
\bigland_{L\subseteq[r]}\;\sum_{j\in L}y_j\geq p_{r}-p_{r-\card{L}}.
$$
By~\eqref{eq:transfpm} this can be written as 
$$
\bigland_{L\subseteq[r]}\;\sum_{j\in L}y_j\geq \card{L}- (\mu_{r}-\mu_{r-\card{L}})k.
$$
So this slice is an upwards closed body within the $r$ dimensional unit-hypercube $I^r$, and if we apply the isomorphism $\vecc y\mapsto\vecc 1_r-\vecc y$ it is flipped around and mapped to the downwards closed body around the origin $\vecc 0_r$ defined by $\vecc y\in I^r$ and $\bigland_{L\subseteq[r]}\sum_{j\in L}y_j\leq (\mu_{r}-\mu_{r-\card{L}})k$. By taking into consideration~\eqref{eq:lambdasdef}  this becomes
\begin{equation}
\label{eq:helper41}
\bigland_{L\subseteq[r]}\;\sum_{j\in L}y_j\leq \lambda_{r-\card{L}+1}k+\dots+\lambda_rk.
\end{equation}
It is easy to see that the extra condition $\vecc y\in I^r$ can be replaced by the weaker one $\vecc y\in \R_+^r$, since the upper bounds $y_j\leq 1$ are already captured by~\eqref{eq:helper41}: for $L=\ssets{j}$ it gives
\begin{equation*} 
\label{eq:ineqboundm+1}
y_j\leq\lambda_rk=\frac{\lambda_r}{m+1}\leq 1,
\end{equation*}
the last inequality holding from \cref{lemma:decrdiffs}. So, we end up with exactly the definition of   
 $\Lambda(k\lambda_1,\ldots,k\lambda_{r})$. We must note here that this SIM-body
 is well-defined, since the $\lambda_r$'s are nondecreasing (\cref{lemma:decrdiffs}).
\end{proof}

\section{Proof of Optimality}
\label{sec:optimality-sja}

In this section we conclude the proof of our main result about the optimality of SJA (\cref{th:mainresultSJA6}). We do that by showing \cref{theorem:nopositivesubSIMs,theorem:weak-deficiency}.

In addition to the SIM-bodies $\Lambda(\lambda_1,\dots,\lambda_r)$ being essentially the building blocks of the allocation space of the SJA, the particular choice of the $\lambda_r$ parameters makes them satisfy another property: they have zero $1$-deficiency:

\begin{lemma}
\label{lemma:SIMpositivedef}
For any dimension $m$, if a subdomain $U_J^{(m)}$ of SJA is nonempty then the corresponding SIM-body $\Lambda(\lambda_1^{(m)},\dots,\lambda_{\card{J}}^{(m)})$ has zero $1$-deficiency.
\end{lemma}

\begin{proof}
Fix some $m$ and let $k=1/(m+1)$. For any nonempty subdomain $U_J$, the slice $\slice{U_J}{-J}{\vecc 0_{m-\cards{J}}}$ is nonempty (by downwards closure), so by \cref{lemma:zeroslicedefSJA} it has zero $k$-deficiency. But from \cref{lemma:SIMallocspaceisom} it is also isomorphic to the SIM-body 
$k\cdot\Lambda(\lambda_1^{(m)},\dots,\lambda^{(m)}_r)$, thus 
$\delta_k(k\cdot\Lambda(\lambda_1^{(m)},\dots,\lambda^{(m)}_r))=0$.
By \hyperref[prop:SIM4]{Property~\ref*{prop:SIM4}} of \cref{lemma:SIMbodiesprops2}, this means that indeed
$
\delta_1(\Lambda(\lambda_1^{(m)},\dots,\lambda^{(m)}_r))\linebreak[3]=0.
$
\end{proof}

Now we are ready to prove \cref{theorem:nopositivesubSIMs}. \emph{It is essentially the only ingredient of this paper whose proof does not work for more than $6$ items} (condition~\eqref{eq:SJAlambdasbound3}, specifically). In a way it demonstrates the maximality of the deficiency of the particular critical SIM-bodies $\Lambda(\lambda_1,\dots,\lambda_r)$, in the sense that they cannot contain subsets that have greater deficiency than themselves.

\begin{theoremnonum}[\cref{theorem:nopositivesubSIMs}]
  For up to $m\leq 6$, no SIM-body $\Lambda(\lambda_1^{(m)},\ldots,\lambda_{r}^{(m)})$ corresponding to a nonempty subdomain $U_{[r]}^{(m)}$
  contains positive 1-deficiency sub-bodies.
\end{theoremnonum}

\begin{proof}
We will prove the stronger statement that for all $r\leq m\leq 6$ no SIM-body $\Lambda(\lambda_1^{(m)},\ldots,\lambda_{r}^{(m)})$ contains a sub-body with nonnegative $1$-deficiency greater than its own, i.e.,
\begin{equation}
\label{eq:monotnicitysubsetSIM}
\emptyset\neq A\subseteq \Lambda(\lambda_1^{(m)},\ldots,\lambda_{r}^{(m)})\;\; \land \;\; \delta_1(A)\geq 0\quad\then\quad\delta_1(A)\leq\delta_1(\Lambda(\lambda_1^{(m)},\ldots,\lambda_{r}^{(m)}))
\end{equation}
This is enough to establish the theorem, because of \cref
{lemma:SIMpositivedef}. We will use induction on $r$. At the basis, whenever $r=1$, for any number of items $m$ the SIM-body is just the line segment $\Lambda(\lambda_1^{(m)})=[0,\lambda_1^{(m)}]$ and it is easy to see that every (nonempty) subset of it will have smaller volume but the same projection, resulting in smaller deficiency.

Moving on to the inductive step,  for simplicity denote $\Lambda=\Lambda(\lambda_1^{(m)},\ldots,\lambda_{r}^{(m)})$ and let $A\subseteq\Lambda$ be a maximum volume sub-body of maximum nonnegative deficiency within $\Lambda$. Without loss (by \cref{th:wlogsymmetric}) $A$ can be assumed to be symmetric and downwards closed. By \cref{th:pclosure}, this tells us that every slice of it must be p-closed (since $A$ is within $\Lambda$ which is a SIM-body and thus p-closed). We will prove that $A=\Lambda$ which is enough to establish~\eqref{eq:monotnicitysubsetSIM}. 

We start by showing that the outmost $(r-1)$-dimensional slice of $A$, namely $\slice{A}{\ssets{1}}{w(\Lambda)}$, cannot be empty. Notice that, by \hyperref[prop:SIM1]{Property~\ref*{prop:SIM1}} of \cref{lemma:SIMbodiesprops2}, $w(\Lambda)=\lambda_r^{(m)}$.  The choice of coordinate $1$ here is arbitrary; due to symmetry any slice $\slice{A}{\ssets{j}}{w(\Lambda)}$ with $j\in{[r]}$ would work in exactly the same way. If this slice was empty, we could add in this free space of $A$  (an $\varepsilon$-neighborhood of) the $(r-1)$-dimensional SIM-body $B$ defined by 
\begin{equation}
\label{eq:m_prime_def}
B= \Lambda(\lambda_1^{(m')},\dots,\lambda_{r-1}^{(m')})
\quad\quad\text{where}\;\;
m'=
\begin{cases}
r-1, &\text{if}\; U_{[r-1]}^{(m)}=\emptyset,\\
m, &\text{otherwise}.
\end{cases}
\end{equation}
Observe here, that by taking into consideration the values of the $\lambda_j^{
(m)}$ parameters of the SJA mechanism (see \cref
{eq:lambdas_def1,eq:lambdas_def2}) we can see that the following
properties are satisfied for all $r\leq m\leq 6$:
\begin{equation}
\label{eq:SJAlambdasbound3}
\lambda_{j}^{(m')}\leq\lambda_{j}^{(m)}\quad\text{and}\quad \lambda_{j}^{(m')}\leq j+1,
\quad\text{for all}\;\; j\in[r-1]
\end{equation}
and
\begin{equation}
\label{eq:SJAlambdasbound3_1}
\lambda_{j}^{(m)}\leq j+1,
\quad\text{for all}\;\; j\in[r-2].
\end{equation}
In particular, for $m=r=6$, notice that although 
$\lambda^{(m)}_{r-1}=\lambda^{(6)}_5\approx 6.3613>6=(r-1)+1$ 
(and that is why Property~\eqref{eq:SJAlambdasbound3_1} above cannot be extended
to $j=r-1$), 
it's still the case\footnote{Due to the fact that $U^{(6)}_{[5]}=\emptyset$ (see
\eqref{eq:special_empty_allocations_5_6}) and the definition of $m'$
in \eqref{eq:m_prime_def}.} that
$\lambda^{(m')}_{r-1}=\lambda^{(r-1)}_{r-1}=\lambda^{
(5)}_5=6=(r-1)+1$ and so \eqref{eq:SJAlambdasbound3} holds.
So, it must be that
$$
B\subseteq \Lambda(\lambda_1^{(m)},\dots,\lambda_{r-1}^{(m)})=\slice{\Lambda}{\ssets{1}}{\lambda_{r}^{(m)}},
$$
the first inclusion being a result of~\eqref{eq:SJAlambdasbound3} and the last equality being from \hyperref[prop:SIM3]{Property~\ref*{prop:SIM3}} of \cref{lemma:SIMbodiesprops2}.  This means that $B$ indeed fits in the exterior space $\Lambda$ at distance $x_1=\lambda_r^{(m)}$, which is exactly where we put it.

We will now show that this addition caused no decrease at the $1$-deficiency of $A$, which would contradict the maximality of the volume of $A$. Equivalently, we need to show that the increase we caused in the volume by extending $A$ was at least equal to the increase in the total volume of its projections. First, we show that no new projections were created with respect to coordinate $1$, i.e., $B$ was already included in $A_{[r]\setminus\ssets{1}}=\slice{A}{\ssets{1}}{0}$. Indeed, it is

$$B=\dclosure(\chull(\permuts(\lambda^{(m')}_1,\dots,\lambda^{(m')}_{r-1})))\subseteq \dclosure(\chull(\permuts(2,\dots,r)))\subseteq \slice{A}{\ssets{1}}{0}.$$ 
The first equality comes from \hyperref[prop:SIM2]{Property~\ref*{prop:SIM2}} of the SIM-bodies in \cref{lemma:SIMbodiesprops2}, the second inclusion is from~\eqref{eq:SJAlambdasbound3}, and the last inclusion is by \cref{th:counterchainwidths} and the p-closure of $ \slice{A}{\ssets{1}}{0}$. What is left to show is that the sum of the new projections created with respect to the remaining coordinates $[2...r]$ was at most equal to the increase in the volume. But this comes directly from the fact that the slice $B$ we added has zero $1$-deficiency: it is a SIM-body corresponding to a subdomain $U_{[r-1]}^{(m')}\neq\emptyset$ (see \cref{lemma:SIMpositivedef}).

So, in the following we can indeed assume that body $A\subseteq\Lambda$ is of maximum width 
$w(A)=\lambda_r^{(m)}.$
Then we will show that, at $x_1=w(A)$, $A$ must in fact include the \emph{entire} corresponding slice of $\Lambda$. This slice is $\slice{\Lambda}{\ssets{1}}{\lambda_{r}^{m}}=\Lambda(\lambda_1^{(m)},\dots,\lambda_{r-1}^{(m)})$, so that would mean that the extreme point $(\lambda_1^{(m)},\dots,\lambda_{r-1}^{(m)},\lambda_r^{(m)})$ is in $A$, and thus by p-closure (\cref{th:pclosure}) the body $\dclosure(\chull(\permuts(\lambda_1^{(m)},\dots,\lambda_r^{(m)})))$ must be included within $A$. But from \hyperref[prop:SIM2]{Property~\ref*{prop:SIM2}} of \cref{lemma:SIMbodiesprops2} this body is exactly the entire external body $\Lambda$, which concludes the proof. So let's show that indeed $\slice{A}{\ssets{1}}{\lambda_{r}^{m}}=\slice{\Lambda}{\ssets{1}}{\lambda_{r}^{m}}$. It is enough to show that removing this slice of $A$ and replacing it with the full slice of $\Lambda$ would result in a non-decrease of the $1$-deficiency: that would contradict the maximality of the volume of $A$.

First, notice that $\slice{A}{\ssets{1}}{\lambda_{r}^{m}}$ is within $\slice{\Lambda}{\ssets{1}}{\lambda_{r}^{m}}$, where $\slice{\Lambda}{\ssets{1}}{\lambda_{r}^{m}}$ is the SIM-body $\Lambda(\lambda_1^{(m},\dots,\lambda_{r-1}^{(m)})$ and also slice $\slice{A}{\ssets{1}}{\lambda_{r}^{m}}$ must have nonnegative deficiency (by \cref{th:counterslices}). So, by the induction hypothesis it must be that the full slice $\slice{\Lambda}{\ssets{1}}{\lambda_{r}^{m}}$ has at least the deficiency of the slice $\slice{A}{1}{\lambda_{r}^{m}}$ it replaces. That means that, taking into consideration only projections in the directions $[2...r]$, the overall change in the deficiency is indeed nonnegative. So, to conclude the proof it is enough to show that no new projections with respect to coordinate $1$ are created by this replacement, i.e.~that $\Lambda(\lambda_1^{(m},\dots,\lambda_{r-1}^{(m)})$ was already included in $A_{[r]\setminus\ssets{1}}=\slice{A}{\ssets{1}}{0}$. Indeed,
\begin{align*}
\Lambda(\lambda_1^{(m},\dots,\lambda_{r-1}^{(m)}) &=\dclosure(\chull(\permuts(\lambda^{(m)}_1,\dots,\lambda^{(m)}_{r-1})))\\
		&\subseteq \dclosure(\chull(\permuts(\lambda^{(m)}_1,\dots,\lambda^{(m)}_{r-2},\lambda^{(m)}_{r})))\\
		&\subseteq \dclosure(\chull(\permuts(2,\dots,r-1,w(A)))),
\end{align*}
by~\eqref{eq:SJAlambdasbound3_1} and the fact that $w(A)=\lambda_r^{(m)}$,
which concludes the proof since slice $\slice{A}{\ssets{1}}{0}$ is p-closed and $(2,\dots,r-1,w(A))$ belongs to it, because $(1,2,\dots,r-1,w(A))$ belongs to $A$ by \cref{th:counterchainwidths}.
\end{proof}

 We now present our main tool to prove that SJA is optimal. It utilizes the fact that the allocation space of SJA has no positive deficiency subsets in a combinatorial way.

\begin{theoremnonum}[\cref{theorem:weak-deficiency}]
If for every nonempty subdomain $U^{(m)}_{J}$ of SJA the corresponding SIM-body $\Lambda(\lambda_1^{(m)},\dots,\lambda_{\card{J}}^{(m)})$ contains no sub-bodies of positive $1$-deficiency, then SJA is optimal. 
\end{theoremnonum}

\begin{proof}
The proof of \cref{theorem:weak-deficiency} is done via a
combinatorial detour to a discrete version of the problem, which is
interesting in its own right and highlights the connection of the dual program
with bipartite matchings. The nonpositive deficiencies property allows us to utilize Hall's marriage condition.
Let us denote by $I_j\equiv\ssets{(\vecc x_{-j},1)\fwh{\vecc x\in I^m}}$ the side on the boundary of the $I^m$ cube which is perpendicular to axis $j$, for $j\in[m]$.

We start by restricting the search for an appropriate feasible dual
solution to those functions $z_j(\vecc x)$ that have the following form:
\begin{quote}
  Fix some integer $N$ which is a multiple of $m+1$ and let
  ${\varepsilon}'=1/N$. We discretize the space by taking a fine grid
  partition of the hypercube $I^m$ into small hypercubes of side ${\varepsilon}'$
  and we require that inside each small hypercube the derivatives
  $\partial z_j(\vecc x)/\partial x_j$ are constant and take either value 0
  or value $m+1$.
\end{quote}
We must point out here that this discretization is used only in the analysis
and it is not part of the optimal selling mechanism which is given just by its
prices $p^{(m)}_r$. 

With the discretization, the combinatorial nature of the dual
solutions emerges: a dual solution is essentially a coloring of all the $\varepsilon'$-hypercubes of $I^m$
into colors $0, 1, \dots,m$. The interpretation of the coloring is
the following: the derivative $\partial z_j(\vecc x)/\partial x_j$ 
has a positive value $m+1$ if and only if the corresponding hypercube (at which $\vecc x$ belongs to) has color $j$, otherwise it is zero (i.e., $z_j(\vecc x)$ is constant with respect to the direction of the $j$-axis);
color $0$ is used exactly for the points where \emph{all} $z_j$ functions are
constant. A \emph{feasible} dual solution corresponds to a coloring in which
every line of hypercubes parallel to some axis, say axis $j$, contains at least
$N/(m+1)$ hypercubes of color $j$. To see this, notice that
function $z_j(\vecc x)$ must increase in a fraction of (at least) $1/(m+1)$
of those small hypercubes (because it starts at value $0$ and has to increase
to a value of at least $1$; see the dual constraints in \cref{th:weakdualityuniform}). 
\Cref{fig:2discretecoloring} illustrates such a coloring for $m=2$ items.

\begin{figure}
\centering
\includegraphics[width=0.4\textwidth]{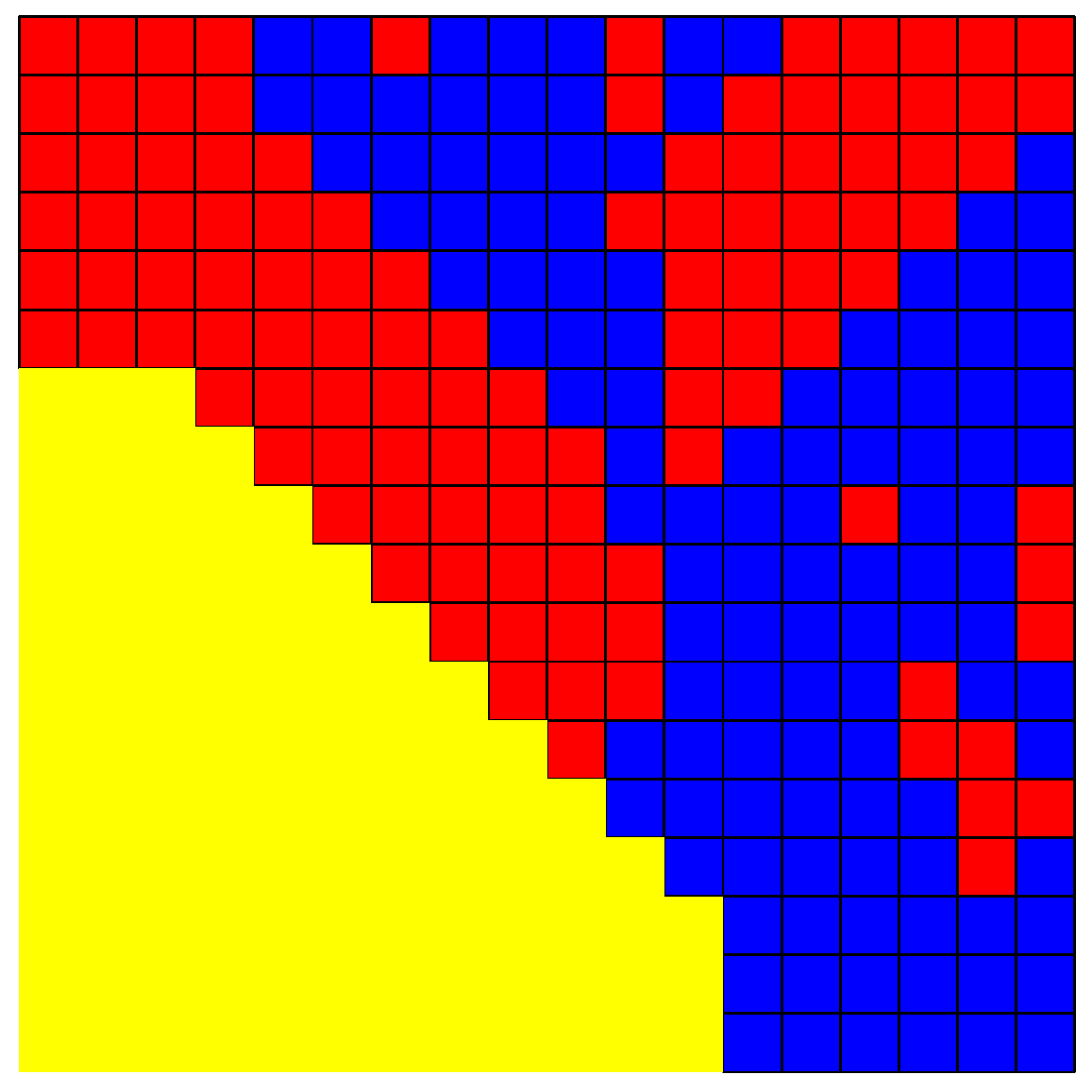}
~
\includegraphics[width=0.4\textwidth]{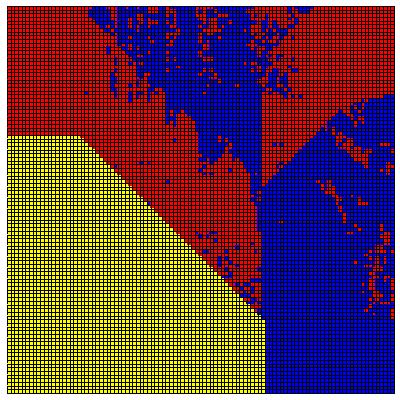}
\caption{\footnotesize Proper colorings of the allocation space $\overline U_{\emptyset}$ of the SJA mechanism for $m=2$ items and different discretization factors $N=18$ (left) and $N=105$ (right). Blue corresponds to the direction of the horizontal axis and red to the vertical axis. The zero region $U_{\emptyset}$ where no item is allocated (white region in \cref{fig:SJA2and3}) is colored in yellow. Notice how the entire region $U_{\ssets{1}}$ is colored blue and the entire $U_{\ssets{2}}$ red. The critical and technically involved part of the coloring for two items is the one of region $U_{\ssets{1,2}}$ where both items are allocated. Interpreting this in the realm of the dual program and the language of the proof of \cref{theorem:weak-deficiency}, blue is color $1$ and corresponds to the points where function $z_1$ increases with ``full'' derivative $m+1=3$ (with respect to the coordinate $x_1$)  while $z_2$ remains constant (with respect to coordinate $x_2$). Red is color $2$ and denotes the reverse situation where $z_2$ increases with derivative $3$ (with respect to coordinate $x_2$) and $z_1$ remains constant (with respect to variable $x_1$). Yellow is color $0$ where both $z_1$ and $z_2$ are constant.}
\label{fig:2discretecoloring}
\end{figure}

To formalize this let us discretize the unit cube $I^m$ in $\varepsilon'$-hypercubes 
$[(i_1-1)\cdot\varepsilon',i_1\cdot\varepsilon']\times\dots\times[(i_m-1)\cdot\varepsilon',i_m\cdot\varepsilon']$, where $i_j\in[N]$ for all $j\in[m]$ (see \cref{fig:coloring-graph}). To keep notation simple, we will sometimes identify hypercubes by their center points, i.e., refer to the $\varepsilon'$-hypercube $\vecc x$ instead of the cube $[x_1-\varepsilon'/2,x_1+\varepsilon'/2]\times\dots\times[x_m-\varepsilon'/2,x_m+\varepsilon'/2]$. In that way, $I^m$ is essentially an $m$-dimensional lattice of points 
$$\left((i_1-1)\cdot\varepsilon'+\varepsilon'/2,\dots,(i_m-1)\cdot\varepsilon'+\varepsilon'/2\right),\quad i_j\in[N],j\in[m].$$
Based on this, for any $S\subseteq I^m$ we will denote by $\discube{S}$ the set of lattice points in $S$.

Next, consider the subdomain $U_J$ where SJA sells exactly the items that are in $J\subseteq [m]$. For any one of these ``active'' coordinates $j\in J$ take $U_J$'s boundary at side $I_j$ of the unit cube and ``inflate'' it to have a width of $k=1/(m+1)$. Formally, for all $J\subseteq [m]$ and $j\in J$ define 
\begin{equation}
 B_{J,j}\equiv \sset{(t,\vecc x_{-j})\fwh{\vecc x\in U_J\;\land\;t\in[1,1+k]}}.
 \end{equation}
$B_{J,j}$ is isomorphic to $(U_J)_{[m]\setminus\ssets{j}}\times[0,k]$.
For any subset of items $J\subseteq[m]$ denote $B_J=\bigunion_{j\in J}B_{J,j}$ and $B=\bigunion_{J\subseteq[m]}B_J$ the entire external layer on all sides.

Notice that $\overline U_{\emptyset}$ cannot be perfectly discretized: the small hypercubes do not fit exactly inside ${\overline U_{\emptyset}}$ because its
boundaries are not rectilinear\footnote{The solution of partitioning
  the unit hypercube into small simplices instead of small
  hypercubes does not work either; although simplices have more
  appropriate boundaries, we cannot guarantee that there exists an 
  ${\varepsilon}'$ for which all the boundaries of ${\overline U_{\emptyset}}$ coincide with
  some boundaries of the small simplices.}. To fix this, we will take
a cover ${\overline U_{\emptyset}^{*}}$ of ${\overline U_{\emptyset}}$ which \emph{can} be partitioned into
${\varepsilon}'$-hypercubes. More precisely, define
${\overline U_{\emptyset}^{*}}$ to be the union of all ${\varepsilon}'$-hypercubes of $I^m$ that intersect ${\overline U_{\emptyset}}$. Finally, let's also extend the boundary region $B$ by adding on top of every boundary component $B_{J,j}$ a thin strip 
\begin{equation}
B_{J,j}^*=\ssets{(t,\vecc x_{-j})\fwh{\vecc x\in U_J\;\land\; t\in[1+k,1+k+g(m)\cdot\varepsilon']}},
\end{equation}
where $g(m)=\lceil\sqrt{m}+1\rceil$, 
and extend notation in the obvious way:
$B^*_J=\bigunion_{j\in J}B^*_{J,j}$ and $B^*=\bigunion_jB^*_j$.

Now it's time to fully reveal the combinatorial structure of our construction by defining a bipartite graph $G(\discube{\overline U_{\emptyset}^{*}}\union\discube{B\union B^*},E)$, which has as nodes the $\varepsilon'$-hypercubes of the cover $\overline U_{\emptyset}^{*}$ and the boundary $B\union B^*$ (see \cref{fig:coloring-graph}). 
Intuitively, the edges $E$ will connect all lattice points of a subdomain $U_J$ with the nodes of its corresponding boundary $B_J\union B^*_{J}$ that agree on $m-1$ coordinates; each $U_J$ is projected onto the sides $I_j$ of the cube that correspond to active items $j\in J$. To be precise, for any $\vecc x\in \discube{\overline U_{\emptyset}^{*}}$ and $\vecc y\in \discube{B\union B^*}$,
$$
(\vecc x,\vecc y)\in E \ifif \vecc x_{-j}=\vecc y_{-j} \;\text{for some $j\in J$, $J\subseteq [m]$, with $\varepsilon'$-hypercube $\vecc x$ intersecting $U_J$}
$$
Another way to view this is that edges start from a node on a side $j$ of the external layer $B\union B^*$, are perpendicular to that side of the unit cube (i.e., parallel to axis $j$) and run towards its interior body $\overline U_{\emptyset}^{*}$, excluding the areas where $j$ is not sold.
\begin{figure}
\centering
\includegraphics[width=0.99\textwidth]{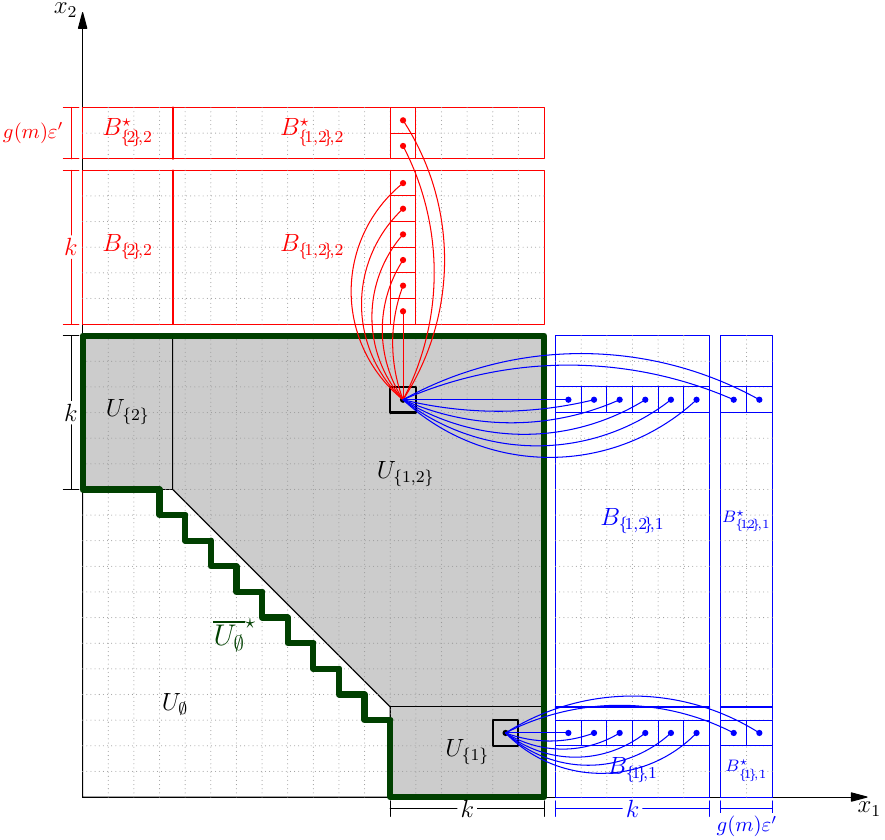}
\caption{\footnotesize The discretization of the allocation space and the structure of graph $G$ used in the proof of \cref{theorem:weak-deficiency}, for $m=2$ items. The space $\overline U_{\emptyset}$ where SJA sells at least one item (colored gray) does not properly align with the $\varepsilon'$-discretization grid so we have to take a cover $\overline U_{\emptyset}^*$ (outlined with the thick line, green in the color version of this paper). The boundaries $B_j\union B_j^{*}$ have width $k+g(m)\varepsilon'$. The one on the right (perpendicular to the vertical axis) consists of $\varepsilon'$-cubes holding color $1$ (blue at the color version of the paper) and the one at the top color $2$ (red). Edges run from every internal $\varepsilon'$-cube, vertically towards the red exterior and horizontally towards the blue exterior. Notice, however, how the cube within the allocation subspace $U_{\ssets{1}}$ has only horizontal (blue) edges running out of it, since it is not allowed to use color $2$ (red). That is due to the fact that item $2$ is not sold within $U_{\ssets{1}}$.}
\label{fig:coloring-graph}
\end{figure}

By this construction, \emph{a bipartite matching of graph $G$ that matches \emph
{completely} the initial boundary $\discube{B}$ corresponds to a proper coloring of the $\varepsilon'$-hypercubes of $I^m$}: an internal cube matched to a node in side $B_j$ is assigned color $j$ and all unmatched cubes are assigned color $0$; every line parallel to an axis $j\in [m]$ contains at least $k/\varepsilon '=N/(m+1)$ distinct hypercubes in the boundary $B_j$.

We will use standard graph-theoretic notation and for any set of nodes $X$, $N(X)$ will denote its set of neighbors, i.e., $N(X)=\sset{ y\fwh{(x,y)\in E\;\;\text{for some}\;\; x\in X}}$. Hall's condition tells us that in any bipartite graph $G=(X\union Y, E)$ there is a matching that completely matches $X$ if and only if
\begin{equation}
\card{S}\leq \card{N(S)}\quad\text{for all}\; S\subseteq X.\tag{Hall's condition}
\end{equation} 

What does the nonpositive $1$-deficiency property of all SIM-bodies $\Lambda(\lambda_1,\dots,\lambda_r)$, $r\leq m$, tell us about graph $G$? Remember (\cref{lemma:SIMallocspaceisom}) that these SIM-bodies correspond to slices $\slice{U_J}{-J}{\vecc t}$ of the allocation space, so (using also \hyperref[prop:SIM4]{Property~\ref*{prop:SIM4}} of \cref{lemma:SIMbodiesprops2}) for any $J\subseteq[m]$, $\vecc t\in\R_+^{m-\card{J}}$ and $S\subseteq \slice{U_J}{-J}{\vecc t}$:
\begin{equation}
\label{eq:nopositivecountercolor}
\card{S}\leq k\sum_{j\in J}\card{S_j},
\end{equation}
where $S_j\equiv S_{[m]\setminus\ssets{j}}$.
Using the fact that every such slice $\slice{U_J}{-J}{\vecc t}$ has zero
$k$-deficiency (\cref{lemma:zeroslicedefSJA}), if we take complements $\overline S=\slice{U_J}{-J}{\vecc t}\setminus S$ and $\overline S_j=\left(\slice{U_J}{-J}{\vecc t}\right)_{[m]\setminus\ssets{j}}\setminus S_j$ the above relation gives
\begin{equation}
\label{eq:nopositivecountercolorinv}
k\sum_{j\in J}\card{\overline S_j}\leq \card{\overline S}.
\end{equation}

First we will show that there is a matching on the bipartite graph $G$ we defined, which completely matches all nodes in $\discube{\overline U_\emptyset}$. By \eqref{eq:nopositivecountercolor} and the fact that every $(U_J)_{[m]\setminus\ssets{j}}\times[0,k]$ is isomorphic to $B_{J,j}$, Hall's theorem tells us that we can completely match $\discube{\slice{U_J}{-J}{\vecc t}}$ into $\discube{B_J}$. By the way we have constructed the edge set $E$, this directly means that there is a complete matching of $\discube{\overline U_\emptyset}$ into $\discube{B}$. 
So, to extend this into a complete matching of the cover $\discube{\overline
U_\emptyset^*}$, using this time additional points in the extended thin-stripe
boundary
$B^*$ at the other side of the bipartite graph $G$, 
it is enough to show that the extra lattice points in 
$\overline{U}_\emptyset^*\setminus \overline U_\emptyset$
of any line parallel to some axis $j$ are at most $g(m)$,  the number of
neighbors in $B^*$. Indeed, any point in $\overline U_{\emptyset}^{*}$ cannot have distance more than $\sqrt{m}\varepsilon'\leq g(m)\varepsilon'$ from a point in $\overline U_{\emptyset}$, because every $\varepsilon'$-hypercube of $\overline U_{\emptyset}^{*}$ intersects with $\overline U_{\emptyset}$ and the diameter of such a hypercube (with respect to the Euclidean metric) is exactly $\sqrt{m}\varepsilon'$.
We will now show that there is also a complete matching of $\discube{B}$ into $\discube{\overline U_\emptyset^*}$. By the way we constructed the edge set, it is enough to show that every slice $\discube{\slice{B_J}{-J}{\vecc t}}$ of the boundary can be completely matched into the corresponding internal slice $\discube{\slicesmall{\overline U_\emptyset^*}{-J}{\vecc t}}$. Fix some nonempty $J\subseteq [m]$ and $\vecc t\in\R_+^{m-\card{J}}$. By Hall's theorem it is enough to prove that, for any family of sets of $\ssets{T_j}_{j\in J}$ of lattice points $T_j\subseteq \discube{\slice{B_{J,j}}{-J}{\vecc t}}$, $\sum_{j\in J}\cards{T_j}\leq \cards{\bigunion_{j\in J} N(T_j)}$. We will prove the stronger $\sum_{j\in J}\cards{T_j}\leq \cards{\bigunion_{j\in J} N(T_j)\inters \slicesmall{\overline U_\emptyset}{-J}{\vecc t}}$; that is, we will just count neighbors in the initial set $\overline U_\emptyset$ and not the cover $\overline U^*_\emptyset$. The continuous analogue of this is to take $T_j$'s be subsets of $\slice{B_{J,j}}{-J}{\vecc t}$ and consider the natural extension of the neighbor function $N$ when  we now have a infinite graph of edges
$$
\sset{(\vecc x,\vecc y)\fwh{\vecc x\in \slice{U_J}{-J}{\vecc t}\;\land\;\vecc y\in\slice{B_J}{-J}{\vecc t}\;\land\; \vecc x_{-j}=\vecc y_{-j} \;\text{for some $j\in J$}}}
$$
Let $S=\slice{U_J}{-J}{\vecc t}\setminus \bigunion_{j\in J} N(T_j)$ be the set of points \emph{not} being neighbors of any node in $\bigcup_{j\in J}T_j$ of the boundary. Then by~\eqref{eq:nopositivecountercolorinv} it is enough to show that $\sum_{j\in J}\cards{T_j}\leq k\sum_{j\in J}\card{\overline S_j}$, where $S_j=S_{[m]\setminus\ssets{j}}$. Every point in the boundary $\slice{B_{J,j}}{-J}{\vecc t}$ that has neighbors in $\bigunion_{j\in J} N(T_j)$ projects (with respect to $j$) inside $\overline S_j$.
But, for any point $\vecc y$ in $T_j$ the only other points that can have the same projection with respect to coordinate $j$ are all points of the line segment of $\slice{B_{J,j}}{-J}{\vecc t}$  which is parallel to the $j$-axis and passes through $\vecc y$, and this segment has length $k$.

Combining the existence of the above two matchings,  a straightforward use of the classic Cantor--Bernstein theorem from Set Theory  ensures the existence of a matching in graph $G$ that \emph{completely} matches both $\discube{\overline U_{\emptyset}^{*}}$ and $\discube{B}$. But, as we discussed before, this means that $\overline U_{\emptyset}^{*}$ is properly colorable and thus this coloring induces a feasible dual solution. Let's denote this solution by $z_j(\vecc x)$, $j\in[m]$ and also let $u(\vecc x)$ be the primal solution given by SJA, i.e., $u$ is the utility function of the SJA mechanism. To prove the optimality of $u$, 
we will take advantage of the
\emph{approximate} complementarity:
we claim that this primal-dual pair of solutions satisfies the approximate
complementarity conditions in \cref{lemma:complementaritymany} for $\varepsilon=g(m) m(m+1)\cdot\varepsilon'$: 
\begin{align}
  u(\vecc x) \cdot  \left(m+1-\sum_{j \in [m]} \frac{\partial z_j(\vecc x)}{\partial x_{j}}\right)& \leq \varepsilon \label{eq:colormatchdul1}\\
 - u(0, \vecc x_{-j}) \cdot  z_j(0, \vecc x_{-j}) & \leq \varepsilon \label{eq:colormatchdul2}\displaybreak[2]\\
  u(1, \vecc x_{-j}) \cdot \left( z_j(1, \vecc x_{-j})
    -1 \right) & \leq \varepsilon\label{eq:colormatchdul3} \\
  z_j(\vecc x) \cdot \left( 1 - \frac{\partial
        u(\vecc x)}{\partial x_{j}} \right) & \leq \varepsilon. \label{eq:colormatchdul4} 
\end{align}
If that is true, then the proof of \cref{theorem:weak-deficiency} is complete, since by the approximate complementarity \cref{lemma:complementaritymany} the primal and dual objectives differ by at most $(3m+1)\varepsilon=(3m+1)g(m)m(m+1)\varepsilon'$, and if we take the limit of this as $\varepsilon'\to 0$, these values must be equal. So let's prove that \eqref{eq:colormatchdul1}--\eqref{eq:colormatchdul4} indeed hold.

Condition~\eqref{eq:colormatchdul2} is satisfied trivially, since both the primal and the dual variables are nonnegative. 
Regarding~\eqref{eq:colormatchdul3}, for any line parallel to some axis $j$ the length of its segment intersecting the boundary $B\union B^*$ (which is the one contributing the critical colors $j$ to that direction) is $k+g(m)\varepsilon'$. So, given that the derivative of $z_j(\vecc x)$ in sections colored with $j$ is $m+1$ we can upper-bound the value of $z_j(1,\vecc x_j)$ by $(k+g(m)\varepsilon')(m+1)=1+g(m)(m+1)\varepsilon'$. This means that $z_j(1,\vecc x_j)-1\leq g(m)(m+1)\varepsilon'$, and given the fact that the utility function has the property that $u(\vecc x)\leq m$ (because its derivatives are at most $1$ at every direction), we finally get the desired $u(1, \vecc x_{-j}) \cdot( z_j(1, \vecc x_{-j})-1 )\leq g(m)m(m+1)\varepsilon'=\varepsilon$.

For condition~\eqref{eq:colormatchdul1}, assume that $u(\vecc x)>0$ (otherwise it is satisfied). That means that SJA sells at least one item, thus $\vecc x\in{\overline U_{\emptyset}}\subseteq{\overline U_{\emptyset}^{*}}$; but ${\overline U_{\emptyset}^{*}}$ is completely matched, thus all points of ${\overline U_{\emptyset}^{*}}$
  are colored with some color in $[m]$ (not with color $0$); this is
  equivalent to the fact that some derivative of the $z_j$ functions is $m+1$ and all others are zero, meaning that the corresponding slack variable $m+1- \sum_{j\in
    [m]} \partial z_j(\vecc x)/\partial x_j$ is zero.

Finally, for condition~\eqref{eq:colormatchdul4}, fix some direction $j\in[m]$
and assume that $\partial u(\vecc x)/\partial x_{j}\neq 1$ (otherwise the condition is satisfied). SJA is deterministic, so it must be that $\partial u(\vecc x)/\partial x_{j}=0$, i.e., item $j$ is not allocated. 
That means that $\vecc x$ belongs to a subdomain $U_J$  with $j\notin J$, and the same is true for all points before it parallel to axis $j$ (that is, all points $(t,\vecc x_{-j})$ with $t\in[0,x_j]$). Thus, by the way that the edge set $E$ of the graph $G$ was defined, $\vecc x$'s $\varepsilon'$-hypercube, as well as all hypercubes before it and parallel to axis $j$, cannot have been colored with color $j$ unless they happen to intersect with a neighboring subdomain $U_{J^*}$ with $j\in J^*$. But it is a simple geometric argument to see that point $(x_j-\varepsilon'm,\vecc x_{-j})$ is at distance at least $\frac{\varepsilon' m}{\sqrt{\card{J^*}}}\geq \frac{\varepsilon' m}{\sqrt{m}}=\sqrt{m}\varepsilon'$ below the boundary $\sum_{j\in J^*}x_j=p_{\card{J^*}}$ of $U_{J^*}$ (since we already know that $\vecc x$ is below it), which is exactly the diameter of the $\varepsilon$-hypercubes. So, at most $m$ such hypercubes below $\vecc x$'s could intersect with $U_{J^*}$, and thus be colored with color $j$, meaning that 
$z_j(\vecc x)$ cannot have increased more than $(m+1)\varepsilon' \cdot (m+1)$ from zero. This proves that indeed
$
z_j(\vecc x)( 1 - \partial u(\vecc x)/\partial x_{j})= z_j(\vecc x)\leq (m+1)^2\varepsilon'\leq \varepsilon.
$
\end{proof}

\section{Conclusion}
\label{sec:conclusion}
Our main goal in this paper was to design revenue maximizing auctions when
\emph{many} heterogeneous items are to be sold to a buyer whose values are
independently, 
uniformly distributed over the unit interval $[0,1]$. 
This is the ``canonical'' multidimensional monopolist problem in Economics that
still remains unsolved, four decades after the seminal work of \citet
{Myerson:1981aa} for the special case of a single item.
We design and analyze a natural mechanism (SJA), and prove its optimality for up to 6 items. Interestingly, it turns out that the optimal mechanism is deterministic.
Prior to our work only
solutions for 2 or 3 items were known~\citep{Manelli:2006vn,Pavlov:2011fk}, and
they were achieved mostly
through a direct, case-specific optimization approach rather than a clear,
unifying viewpoint that could help us towards a more fundamental understanding
of multi-item auction settings.

The cornerstone of our approach is the use of duality and
complementarity. Towards this end, we present a much more general weak duality
theory framework that can be formulated for many buyers and arbitrary joint
distributions which we hope will prove useful in future attacks on
generalizations of our problem: After the conference version of this paper, our
duality framework has already successfully been applied to deal with rather
general classes of distributions~\citep{gk2015,Fiat2016a} and in the domain of
approximately optimal auctions as well~\citep{g2014}.
Our solution illuminates the rich geometric ideas underlying
the problem, and in the process we formally develop some novel geometric
machinery that might be of independent interest.

The most obvious direction for future work is validating the conjecture that SJA
is indeed optimal for any number of items and
not just up to 6.
Another fundamental open problem is that of finding exact optimal
solutions, like the ones provided in this paper, for more than one bidder: for
example,
the seemingly simple case of two items and two buyers with i.i.d.\ uniform
valuations over $[0,1]$ is still wide open. Is determinism
still powerful enough for revenue maximization when multiple bidders are
involved? And if
not, how well can deterministic auctions that generalize SJA approximate the
optimal revenue?

\paragraph{Acknowledgements:} We thank Anna Karlin, Amos Fiat, Sergiu Hart, and Constantinos Daskalakis for useful discussions. We are also very grateful to the anonymous reviewers for their careful and thorough reading of our manuscript and for their valuable feedback.

\nocite{McAfee:1988nx}
\nocite{Wang:2013ab}
\nocite{Thanassoulis:2004ys}
\nocite{Vohra:2011fk}
\nocite{g2014}
\bibliographystyle{abbrvnat} 
\bibliography{DualityUniform}

\appendix

\section{Full Proof of \cref{lemma:zeroslicedefSJA}}
\label{append:fullproofzeroslicedef}
Recall the definition of body $V(p_1^{(m)},\dots,p_r^{(m)})$ in~\eqref{eq:disjunctionspace}. By the definition of SJA in~\eqref{eq:sja}, the volume of this body must be equal to $rk$. Then, as we discussed in the proof sketch of \cref{lemma:zeroslicedefSJA} in \cref{sec:proof-no-positive-deficiency}, this translates to its deficiency being zero:
\begin{equation}
\label{eq:SJAequivdef2}
\delta_{\frac{1}{m+1}}(V(p_1^{(m)},\dots,p_r^{(m)}))=0\quad\text{for all}\;\; r\leq m.
\end{equation}
Before giving the formal proof of \cref{lemma:zeroslicedefSJA}, we will need the following lemma that shows how the deficiency of any such subdomain of the valuation space is the sum of the deficiencies of its ``critical'' subslices of lower dimensions: 

\begin{lemma}
\label{eq:decomposedeficiencies}
For any subset of items $J\subseteq [m]$, the $k$-deficiency of any slice of the subdomain where at least one of the items in $J$ is sold, when all other items' bids are fixed to zero, is the sum of the $k$-deficiencies of all its sub-slices $\slicesmall{(\slicesmall{U^{(m)}_L}{-J}{\vecc 0})}{J\setminus L}{\vecc t}$, where $\emptyset\neq L\subseteq J$ and $k=\frac{1}{m+1}$. Formally,
$$
\delta_{k}(V(p_1^{(m)},\dots,p^{(m)}_{\card{J}}))=\sum_{\emptyset\neq L\subseteq J}\int_{I^{\card{J}-\card{L}}}\delta_k\left(\slice{\left(\slice{U^{(m)}_L}{-J}{\vecc 0}\right)}{J\setminus L}{\vecc t}\right)\,d\vecc t.
$$
\end{lemma}

\begin{proof}
Fix some $m$. For the sake of clarity we will prove the proposition for $J$ having full dimension $J=m$. All the arguments easily carry on to the more general case where $J\subseteq [m]$ if one takes all valuations of items not in $J$ to be $0$, i.e., ``slicing'' $\slice{\left(\;\;\;\cdot\;\;\;\right)}{-J}{\vecc 0}$, since they are valid for any selling mechanism with nonincreasing price differences and SJA specifically; essentially, the case of $\card{J}=m'\leq m$ directly translates to the case of an $m'$-dimensional mechanism.

So, it is enough to show that
\begin{align}
\card{V} &=\sum_{\emptyset \neq L\subseteq [m]}\int_{I^{m-\card{L}}} \card{\slice{U_L}{-L}{\vecc t}}\,d\vecc t \label{eq:nonzeroallocdecomp3} \\
\card{V_{[m]\setminus\ssets{j}}} &=\sum_{\substack{\emptyset \neq L\subseteq [m]\\j\in L}}\int_{I^{m-\card{L}}} \card{\left(\slice{U_L}{-L}{\vecc t}\right)_{[m]\setminus\ssets{j}}}\,d\vecc t&&\text{for all}\;\;j\in [m], \label{eq:nonzeroprojallocdecomp3}
\end{align}
where for simplicity we have denoted the space $V(p_1,\dots,p_m)$ where mechanism allocates at least one item with $V$. Equation~\eqref{eq:nonzeroallocdecomp3} is a result of the fact that $V$ can be decomposed as $V=\sum_{\emptyset \neq L\subseteq [m]}U_L$ and every allocation subspace $U_L$ is isomorphic to the \emph{disjoint} union of all its slices $\slice{U_L}{-L}{\vecc t}$. In a similar way, to prove that~\eqref{eq:nonzeroprojallocdecomp3} holds, it is enough to show that for some fixed $j\in[m]$, the projection $V_{[m]\setminus\ssets{j}}$ can be covered by the union of all the projections of the subspaces $U_L$ with respect to coordinate $j$ and that all these projections $(U_L)_{[m]\setminus\ssets{j}}$ are disjoint almost everywhere, i.e., they can only intersect in a set of measure zero.

For the former, let $\vecc x_{-j}\in V_{[m]\setminus\ssets{j}}$. Then $(\vecc x_{-j},1)\in V$ (by only increasing the components of a valuation profile, items that were sold to the buyer are still going to be sold). So, there is a nonempty set of items $L\subseteq [m]$ such that $(\vecc x_{-j},1)\in U_L$ and $j\in L$ (item $j$ is sold since $x_{j}=1\geq p_1$), meaning that indeed $\vecc x_{-j}\in (U_L)_{[m]\setminus\ssets{j}}$ with $j\in L$. For the latter, consider distinct sets $L,L'\subseteq [m]$ with $j$ belonging to both $L$ and $L'$, and let a valuation profile $\vecc x\in U_L\inters U_{L'}$. Then, by the characterization in \cref{lemma:allocsubspacealgebra} it must be that
$$
\sum_{l\in L'\setminus L}x_l\geq p_{\card{L'}}-p_{\card{L'}-\card{L'\setminus L}}\quad\text{and}\quad \sum_{l\in L'\setminus L}x_l\leq p_{\card{L}+\card{L'\setminus L}}-p_{\card{L}},
$$
the first inequality being from the fact that $\vecc x\in U_{L'}$ and the second from $\vecc x\in U_{L}$, taking into consideration that $L'\setminus L\subseteq L'$ and $L'\setminus L\not\subseteq L$. As a result, the sum $\sum_{l\in L'\setminus L}x_l$ can range at most over only a single value, namely $p_{\card{L'}}-p_{\card{L'}-\card{L'\setminus L}}=p_{\card{L}+\card{L'\setminus L}}-p_{\card{L}}$ (and only if these two values are of course equal), otherwise by merging these two inequalities together we would have gotten that
$$
p_{\card {L'}}-p_{\card{L'}-\card{L'\setminus L}}
< p_{\card{L}+\card{L'\setminus L}}-p_{\card{L}},
$$
which contradicts the nonincreasing payment differences property, since both differences are between payments that differ at exactly $\card{L'\setminus L}$ ``steps'' but $\card{L}+\card{L'\setminus L}\geq \card{L'}$.
\end{proof}

\begin{lemmanonum}[\cref{lemma:zeroslicedefSJA}]
Every slice
  $\slice{U^{(m)}_J}{-J}{\vecc t}$ of SJA has zero $k$-deficiency, where
  $k=\frac{1}{m+1}$.
\end{lemmanonum}
\begin{proof}
Fix some $m$ and let $k=1/(m+1)$. We use induction on the cardinality of $J$. 
At the base of the induction $\card{J}=1$, and due to symmetry it is enough to prove the proposition for slices of the form $\slice{U_{\ssets{1}}}{[2...m]}{\vecc t}$. By~\eqref{eq:sliceszeroenough} this is equal to the slice $\slice{U_{\ssets{1}}}{[2...m]}{\vecc 0_{m-1}}$, which is the single-dimensional interval $[p_1,1]$, thus having $k$-deficiency $1-p_1-\frac{1}{m+1}\cdot 1=\frac{m}{m+1}-p_1=0.$

For the inductive step, fix some $r\leq m$ and assume the proposition holds for all $J\subseteq [m]$ with $\card{J}\leq r-1$. We will show that it is true also for $\card{J}=r$. Again, due to symmetry, it is enough to prove that the $k$-deficiency of slice $\slice{U_{[r]}}{[r+1...m]}{\vecc 0_{m-r}}$ is zero (taking into consideration~\eqref{eq:sliceszeroenough}). By \cref{eq:decomposedeficiencies} we deduce that the subdomain where at least one of items $[r]$ is sold, given that the remaining $[r+1...m]$ bids are fixed to zero, has
\begin{align*}
\delta_{k}(V(p_1,\dots,p_{r})) &=\sum_{\emptyset\neq L\subseteq [r]}\int_{I^{r-\card{L}}}\delta_k\left(\slice{\left(\slice{U_L}{[r+1...m]}{\vecc 0}\right)}{[r]\setminus L}{\vecc t}\right)\,d\vecc t\\
		&=\delta_k\left(\slice{U_{[r]}}{[r+1...m]}{\vecc 0}\right) +\sum_{\substack{\emptyset\neq L\subseteq [r]\\ \card{L}\leq r-1}}\int_{I^{r-\card{L}}}\delta_k\left(\slice{U_L}{[m]\setminus L}{(\vecc t,\vecc 0)}\right)\,d\vecc t\\
		&=\delta_k\left(\slice{U_{[r]}}{[r+1...m]}{\vecc 0}\right),
\end{align*}
by the induction hypothesis. 
But from the definition of SJA, and in particular~\eqref{eq:SJAequivdef2}, we have that $\delta_{k}(V(p_1,\dots,p_{r}))=0$, which concludes the proof.
\end{proof}

\section{Convexity and Duality}
\label{sec:convexity-duality}
In this section we discuss the convexity constraint of the utility functions. For clarity, we focus on the simple case of a single bidder and a single item.  We show that the convexity constraint is not necessary for regular\footnote{A distribution is called \emph{regular} when $x-\frac{1-F(x)}{f(x)}$ is nondecreasing.}  distributions. And in the opposite direction, we exhibit a nonregular distribution for which the convexity constraint cannot be dropped without affecting optimality.

The primal \hyperref[eq:totalrevenue]{Program~\eqref{eq:totalrevenue}} (taking into consideration~\eqref{eq:primal objective}) for this case is
$$\sup_u\;\;u(H) H f(H) - u(L) L f(L) -\int_L^H u(x) (f(x)+(xf(x))') \, dx$$
subject to
\begin{align*}
\tag{$z(x)$}  u'(x) &\leq 1  \\
\tag{$s(x)$}  u'(x) &\geq 0  \\
\tag{$w(x)$}  u''(x)& \geq 0 \\
                      u(x) & \geq 0 
\end{align*}
Notice that there is no reason to include $u(L)=0$ since this holds for the optimal solution; that is because if $u(x)$ and $u(x)-c$ are both feasible solutions and $c$ is a positive constant, then the corresponding objectives differ by $c Hf(H)-c Lf(L)- \int_L^H c (f(x)+(xf(x))') \, dx=-c(H\cdot F(H)-L\cdot F(L))<0$; this shows that the optimal solution has $u(x)=0$ for some $x$.

In our treatment of the subject in the main text of the paper, we dropped the constraints labelled by
the dual variables $w(x)$ and for the most part we also dropped the ones corresponding to $s(x)$. We did this to keep the primal
and dual systems simple. More importantly, there is a strong reason
for ignoring the constraints corresponding to $w(x)$ for
multi-parameter domains: the convexity constraints $\nabla^2 u(x)
\succeq 0$ (that is, the Hessian of $u$ being positive semidefinite) are not linear in $u$ (unlike the one-dimensional case, in
which the constraint $u''(x) \geq 0$ is linear in $u$).

In the rest of this subsection, we investigate when the simplified
systems are optimal. We first drop the constraints corresponding to convexity to get the dual

$$\inf_{z,s} \int_L^H z(x) \, dx$$

subject to

\begin{align*}
  \tag{$u(x)$}  z'(x)-s'(x) & \leq f(x) +(xf(x))' \\
  \tag{$u(H)$}  z(H)-s(H) & \geq H f(H) \\
  \tag{$u(L)$} z(L)-s(L) & \leq L f(L)\\
  					z(x), s(x) & \geq 0.
\end{align*}

\begin{lemma}
   For regular distributions, the above primal and dual programs give the optimal value. 
\end{lemma}
\begin{proof}
We have three linear programs here: the original primal program, the relaxed primal program in which we dropped the constraints labeled $w(x)$, and the dual program. The values of the three programs are clearly in a nondecreasing order, due to relaxation and weak duality.

Therefore, it suffices to give a feasible dual solution which achieves the same value with the original primal problem. From~\citet{Myerson:1981aa}, we know that the original primal program has value $\int_L^H \max(0,\varphi(x))f(x)\, dx$, where $\varphi(x)=x-\frac{1-F(x)}{f(x)}$ is the virtual value function. Since we consider regular distributions, $\varphi(x)$ is nondecreasing. We define the following dual solution:
\begin{align*}
  z(x) &=\max(0, \varphi(x)) f(x) \\
  s(x) &=-\min(0, \varphi(x)) f(x).
\end{align*}
We first show that this constitutes a feasible dual solution. Clearly $z(x)$ and $s(x)$ are nonnegative and
$z(x)-s(x)=\varphi(x) f(x)$. This gives,
\begin{align*}
  z'(x)-s'(x) & = (\varphi(x) f(x))' = f(x) +(xf(x))' \\ 
  z(H)-s(H) & = \varphi(H) f(H) = H f(H) \\
  z(L)-s(L) & = \varphi(L) f(L)= L f(L) - 1 \leq L f(L),
\end{align*}
which shows that it is a feasible dual solution. The lemma follows by observing that the dual objective is $\int_L^H z(x)\, dx = \int_L^H \max(0,\varphi(x))f(x)\, dx$, equal to the value obtained by Myerson's optimal mechanism.
\end{proof}

We now exhibit a (nonregular) distribution for which the relaxed primal and the dual give suboptimal solutions. Consider the probability distribution with cdf 
\begin{equation*}
\label{eq:non-regular-distr}
  F(x)=1-(1-x)(1+x(2.7x -2.9))=x (2.7 x^2-5.6 x+3.9),
\end{equation*}
over the unit interval $I$.
Distribution $F$ is drawn on the left part of \cref{fig:convexity1}; it is not regular and its revenue function $R(x)=x(1-F(x))$ is not concave.  The revenue curve is shown on the right part of the figure. The points $x_0$, $x_2$, and $x_3$ are extrema; the point $x_1$ induces the same revenue with $x_3$.

\begin{figure}[t]
\centering
\begin{subfigure}{0.46\textwidth}
\includegraphics[width=1\textwidth]{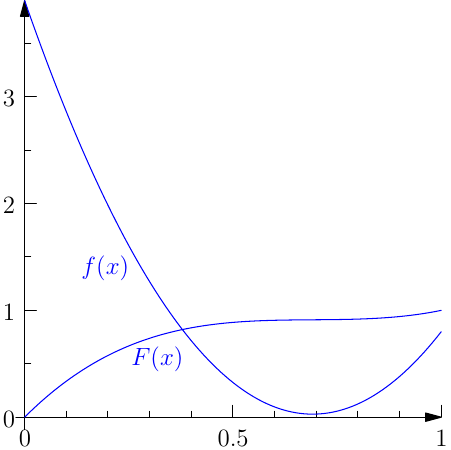}
\caption{\footnotesize The distribution}
\label{fig:convexity1-left}
\end{subfigure}
~
\begin{subfigure}{0.51\textwidth}
\includegraphics[width=1\textwidth]{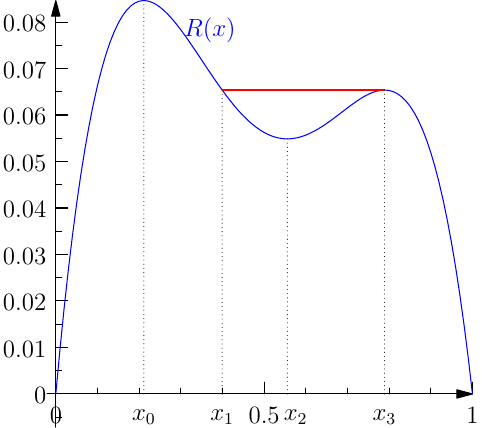}
\caption{\footnotesize The revenue curve $R(x)=x(1-F(x))$}
\label{fig:convexity1-right}
\end{subfigure}
\caption{\footnotesize The probability distribution in (\subref{fig:convexity1-left}) is not regular
      and does not have concave revenue curve, i.e., $F(x)+xf(x)-1$ is
      not monotone (increasing). Its revenue curve is shown in (\subref{fig:convexity1-right}). The
      points $x_0$, $x_2$, and $x_3$ are extrema; the point $x_1$ has
      the same revenue with $x_3$.}
\label{fig:convexity1}
\end{figure}
\begin{figure}
\centering
\begin{subfigure}{0.485\textwidth}
\includegraphics[width=1\textwidth]{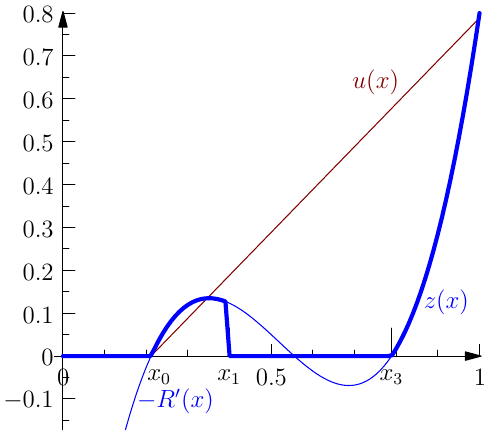}
\caption{\footnotesize Optimal solutions}
\label{fig:convexity2-left}
\end{subfigure}
~
\begin{subfigure}{0.485\textwidth}
\includegraphics[width=1\textwidth]{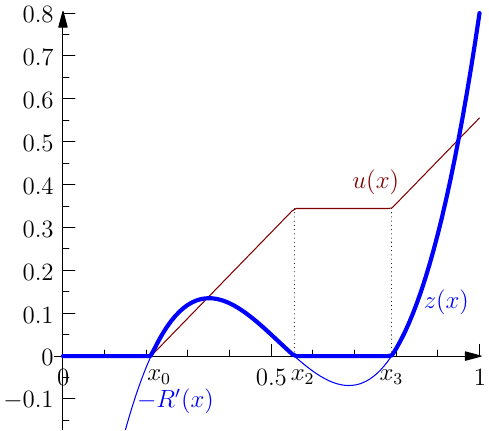}
\caption{\footnotesize Relaxed optimal solutions}
\label{fig:convexity2-right}
\end{subfigure}
\caption{\footnotesize Figure~(\subref{fig:convexity2-left}) shows the optimal solutions for the
     distribution of \cref{fig:convexity1}. The function $z(x)$ is part of the optimal dual solution when we include the convexity constraint. Figure~(\subref{fig:convexity2-right}) shows the optimal solutions when we drop the
      convexity constraint.}
\label{fig:convexity2}
\end{figure}
\Cref{fig:convexity2-left} shows the optimal solutions. 
The optimal primal solution is $u(x)=\max(0, x-x_0)$ and corresponds to the deterministic mechanism with reserve price $x_0$. The figure also shows $z(x)$ of the corresponding dual solution. 
This was computed by taking the dual program including the $w(x)$ constraints corresponding to convexity; showing how we computed this $z(x)$ is beyond the scope of this appendix note and we only provide it so that the reader can compare it with the relaxed solution in the right part of the figure.

\Cref{fig:convexity2-right} shows the optimal solutions for the relaxed primal program and its dual. The dual solution has $z(x)=\max(0, \varphi(x)) f(x)$ and $s(x)=-\min(0,\varphi(x)) f(x)$. As it was shown in the proof of the above lemma, this is a feasible dual solution. It corresponds to the primal solution shown in the figure.

The value of the dual solution is
$$
\hspace{-0.3cm}
\int_0^1 z(x)\, dx = \int_{x_0}^{x_2} -R'(x) \, dx + \int_{x_3}^1
-R'(x) \, dx= R(x_0)-R(x_2)+R(x_3)-R(1)=R(x_0)-R(x_2)+R(x_3).
$$ 
It is straightforward to verify that the indicated primal solution gives the same value, which shows that they are both optimal. However, the primal solution is not convex. Furthermore, its value $R(x_0)-R(x_2)+R(x_3)$ is strictly higher than the value $R(x_0)$ of the valid optimal solution, because $R(x_3)>R(x_2)$. Therefore \emph{the convexity constraint is essential to obtain the optimal solution}.

\section{Duality for Unbounded Domains}
\label{sec:dualityinfinity}
As we mentioned in the presentation of the duality framework in \cref{sec:duality}, for it to make sense as it is we need the integrals in the basic transformation of the primal revenue-maximization objective in expression~\eqref{eq:primal objective} to be well-defined. This is definitely the case when we have bounded domains, i.e., when the upper-boundary $H_{i,j}$ of each interval $D_{i,j}=[L_{i,j},H_{i,j}]$ is finite: all integrals in~\eqref{eq:primal objective} are finite and the integration by parts is valid. This of course includes the special case of uniform distributions which is the main topic in this paper. We will now discuss how one can still use this duality framework in cases where the domain $D$ is not bounded.

For the sake of clarity, let's assume for the remaining of this section that we have a single bidder and that item valuations are i.i.d.\ from some distribution $F$ with pdf $f$ over an interval $[L,H]$, $L\geq 0$. First notice that, even for unbounded domains where $H=\infty$, the critical integral

\begin{equation}
\label{eq:infiniteintervalhelp1}
\int_{D_{-j}}  H_{j} \, u(H_{j},  \vecc x_{-j}) \, f(H_{j}, \vecc x_{-j}) \, d\vecc x_{-j}=Hf(H)\int_{[L,H]^{m-1}} u(H,  \vecc x_{-j}) \prod_{l\neq j} f(x_l) \, d\vecc x_{-j}
\end{equation}
in \eqref{eq:primal objective} may still converge as $H\to\infty$. In such a case, the duality framework from \cref{sec:duality} can be applied as it is: one just has to take the limit of $H\to\infty$ wherever  $H$ appears, and in particular the Weak Duality \cref{lemma:weakdualitymany} is still valid if one replaces condition $z_j(H, \vecc x_{-j}) \geq Hf(H, \vecc x_{-j})$ by its natural limiting version of $$\lim_{H\to\infty}\left( z_j(H,\vecc x_{-j})-Hf(H)\prod_{l\neq j}f(x_l)\right)\geq 0.$$ For example, a sufficient condition for distributions with unbounded support to still induce bounded values in~\eqref{eq:infiniteintervalhelp1} is to have \emph{finite expectation}. This is a rather natural assumption to make and is standard for example in the works of \citet{Myerson:1981aa} and \citet{Krishna:2009ul}. To see why~\eqref{eq:infiniteintervalhelp1} is finite, it can be rewritten as $Hf(H)\expect_{\vecc x_{-j}\sim F^{m-1}}\left[u(H,\vecc x_{-j})\right]$ and so, due to the derivatives constraint $\nabla u(\vecc x)\leq \vecc 1_m$, it is upper-bounded by $Hf(H)\expect_{\vecc x_{-j}\sim F^{m-1}}\left[H+\sum_{l\neq j}x_j\right]=H^2f(H)+(m-1)Hf(H)\expect[X]$. Now, if we take into consideration that any bounded-expectation distribution must have $f(x)=o(1/x^2)$ since $\expect[X]=\int xf(x)\,dx$ must converge, then it is easy to see that this expression converges as $H\to\infty$, and in fact vanishes to zero. 

However, this might not be true for distributions with infinite expectation, for example the equal revenue distribution where $f(x)=1/x^2$ over the interval $[1,\infty)$. In such a case, we can follow a different path in order to use our duality framework. One can take the \emph{truncated} version of the distribution within a finite interval, i.e., consider the distribution $F_b(x)\equiv\frac{1}{F(b)}F(x)$ over the interval $[L, b]$ for any $b\geq L$, apply the duality theory framework in this finite case, and then study the behavior as $b\to\infty$. As the next \cref{th:trunclimit} proves, this process will be without loss. 

In the following we will use the notation $\rev(F)$ from~\citep{Hart:2012uq} to
denote the optimal revenue when item valuations follow the joint distribution $F$. One last remark before stating the theorem is that, whenever one deals with a specific case of the optimal revenue problem, he has to make sure that it is \emph{well-defined}, i.e., that $\rev(F)<\infty$ for the particular distributional priors $F$. This might seem obvious, but let us note here that it is not the case for \emph{any} probability distribution. For example, if we consider i.i.d.~valuations from the Pareto distribution $f(x)=\frac{1}{2}x^{-3/2}$ , $x\in[1,\infty)$, the expected (Myerson) revenue by selling a single item at a price of $t$ is $t(1-F(t))=t(1-1+t^{-1/2})=t^{1/2}$ which tends to infinity. Some simple sufficient conditions for bounded optimal revenue in the i.i.d.\ case where the valuations come from a product distribution $F^m$ are the bounded expectation of the distribution $F$, since by individual rationality (IR) one trivially gets the bound $\rev(F^m)\leq m\expect[X]$, and the bounded Myerson revenue $\rev(F)$ for the single-item case, since from the work of \citet{Hart:2012uq} we know that  there exists a constant $c>0$ such that
$
\frac{c}{\log^2 m}\rev(F^m)\leq \srev(F^m),
$
 where $\srev$ denotes the optimal revenue by selling the items independently, so we get 
$\rev(F^m)\leq\frac{m\log^2 m}{c}\rev(F)$. The former condition is stronger. For example, the equal revenue distribution does not have finite expectation but it does induce a finite Myerson revenue of $1$.

\begin{theorem}
\label{th:trunclimit}
Let $F$ be a probability distribution over $[a,\infty)$, $a\geq 0$, such that $\rev(F^m)<\infty$. Then, if $F_b$ denotes the truncation of $F$ in $[a,b]$, $b\geq a$, and $\lim_{b\to\infty} \rev(F_b^m)$ converges, it must be that
$$
\lim_{b\to\infty} \rev(F_b^m)=\rev(F^m).
$$
\end{theorem}

\begin{proof}
Let $u$ be the utility function of an optimal selling mechanism when valuations are drawn i.i.d.\ from $F$. The restriction of $u$ in $[a,b]$ is a valid utility function for the setting where valuations are drawn i.i.d.\ from $F_b$ and also we know that $F(x)=F(b)F_b(x)$ for all $x\in[a,b]$. Combining these we get
\begin{align*}
\rev(F^m) &=\int_{[a,\infty)^m}\vecc x\cdot\nabla u(\vecc x)-u(\vecc x)\,dF^m(\vecc x)\\
		&= \int_{[a,b]^m}\vecc x\cdot\nabla u(\vecc x)-u(\vecc x)\,dF^m(\vecc x)+\int_{[a,\infty)^m\setminus[a,b]^m}\vecc x\cdot\nabla u(\vecc x)-u(\vecc x)\,dF^m(\vecc x)\\
		&= F^m(b)\int_{[a,b]^m}\vecc x\cdot\nabla u(\vecc x)-u(\vecc x)\,dF_b^m(\vecc x)+\int_{[a,\infty)^m\setminus[a,b]^m}\vecc x\cdot\nabla u(\vecc x)-u(\vecc x)\,dF^m(\vecc x)\\
		&\leq F^m(b)\rev(F_b^m)+\int_{[a,\infty)^m\setminus[a,b]^m}\vecc x\cdot\nabla u(\vecc x)-u(\vecc x)\,dF^m(\vecc x).
\end{align*}

Next, for any $b\geq a$, let $u_b$ be the utility function of an optimal selling mechanism when valuations are drawn i.i.d.\ from $F_b$. This utility function can be extended to a valid utility function $\overline u_b$ over the entire interval $[a,\infty)$ in the following way:
$$
\overline u_b(\vecc x)=u_b(\gamma_b(\vecc x))+(\vecc x-\gamma_b(\vecc x))\cdot \nabla u_b(\gamma_b(\vecc x)),\quad \vecc x\in[a,\infty)^m,
$$
where $\gamma_b(\vecc x)$ the pointwise minimum of $\vecc x$ and $(b)^m$, i.e., is the $m$-dimensional vector  whose $j$-th coordinate is $\min\ssets{x_j,b}$.

Since $u_b$ is a convex function with partial derivatives in $[0,1]$, so is the extended $\overline u_b$. This means that we immediately get
$$
F^m(b)\rev(F_b^m)\leq \rev(F^m).
$$ 
Now the theorem follows from the facts that $\lim_{b\to\infty}F^m(b)=1$ and
$$
\lim_{b\to\infty}\int_{[a,\infty)^m\setminus[a,b]^m}\vecc x\cdot\nabla u(\vecc x)-u(\vecc x)\,dF^m(\vecc x)=0.
$$
The last equality is due to the fact that $\int_{[a,\infty)^m}\vecc x\cdot\nabla u(\vecc x)-u(\vecc x)\,dF^m(\vecc x)$ is bounded by assumption.
\end{proof}

\section{Exact Computation of the Prices for up to 6 Dimensions.}
\label{sec:exactcompute6}
To decongest notation, we will drop the subscript $(m)$ because it will be clear in which dimension we are working in and we denote 
$v(\alpha_1,\dots,\alpha_r)=\card{V(\alpha_1,\dots,\alpha_r)}$ 
where this body is defined in~\eqref{eq:disjunctionspace} and, as we mentioned after \cref{def:SJAmechanism}, we can use the equivalent to~\eqref{eq:sja} condition
$$
v(p_1,\dots,p_r)=rk
$$ 
to determine the SJA payments. Also, we set $k=\frac{1}{m+1}$ throughout this section.

\begin{itemize}
\item\underline{$r=1$ and any $m$:} As we said before, it is easy to see that for any dimension $m$, 
\begin{equation}\label{eq:volsliceSJA1}
v(p_1)=1-p_1.
\end{equation} 
From this, and applying the transformation~\eqref{eq:transfpm}, we solve
$$
v(p_1)=1\cdot k\ifif 1-(1-\mu_1k)=k\ifif \mu_1=1.
$$
So
$$
p_1^{(m)}=\frac{m}{m+1}\quad\text{and}\quad \mu_1=1.
$$
\end{itemize}
For higher orders $r>1$ we can utilize the recursive way of computing the expressions for the volumes $v_r$, given by formula~\eqref{eq:recsliceSJAvol} and the initial condition~\eqref{eq:volsliceSJA1}.
\begin{itemize}
\item \underline{$r=2$ and any $m$:}
Using the recursive formula~\eqref{eq:recsliceSJAvol} and~\eqref{eq:volsliceSJA1} we can compute that for every $p_2$ such that $0\leq p_2-p_1\leq p_1$ it would be
$$
v(p_1,p_2)=\int_0^{p_2-p_1}v(p_1)\,dt+\int_{p_2-p_1}^{p_1}v(p_2-t)\,dt+\int_{p_1}^1\,dt=p_1^2+\frac{p_2^2}{2}-2p_1p_2+1
$$
and by applying the transformation~\eqref{eq:transfpm} and plugging in the already computed value $\mu_1=1$ from the previous order $r=1$, we get
\begin{equation}\label{eq:volsliceSJA2}
v(1-\mu_1k,2-\mu_2k)=2k\ifif \mu_2^2-4 \mu_2+2=0
\end{equation}
If we pick the largest root of this equation $\mu_2=2+\sqrt{2}$ we can see that indeed condition $0\leq p_2-p_1\leq p_1$ is respected (it is equivalent to $0\leq k\leq 1/(1+\sqrt{2})$ which holds since $k\leq\frac{1}{r+1}$ and $r\geq 2$), so we have computed that for any $m$
$$
p_2^{(m)}=\frac{2m-\sqrt{2}}{m+1}\quad\text{and}\quad \mu_2=2+\sqrt{2}\approx 3.41421.
$$ 
\item \underline{$r=3$ and any $m$:}
In the same way, using again recursive formula~\eqref{eq:recsliceSJAvol} and the volume of the previous order $r=2$ from~\eqref{eq:volsliceSJA2} we can compute that for every $p_3$ such that $0\leq p_3-p_2\leq p_2-p_1$ it would be
\begin{multline}
\ \hspace{-1cm}
v(p_1,p_2,p_3) =\int_0^{p_3-p_2}v(p_1,p_2)\,dt+\int_{p_3-p_2}^{p_2-p_1}v(p_1,p_3-t)\,dt+\int_{p_2-p_1}^{p_1}v(p_2-t,p_3-t)\,dt+\int_{p_1}^1\,dt\\
		=\frac{1}{6} \left(-3 p_1^3+9 p_1^2 p_3+9 p_1 \left(2 p_2^2-4 p_2 p_3+p_3^2\right)-6 p_2^3+9 p_2^2 p_3-p_3^3+6\right)\label{eq:eq:volsliceSJA3}
\end{multline}
and by applying the transformation~\eqref{eq:transfpm} and plugging in the already computed values for $\mu_1=1$ and $\mu_2=2+\sqrt{2}$ from the previous orders, we get
\begin{equation}\label{eq:mu3defpoly}
v(1-\mu_1k,2-\mu_2k,3-\mu_3k)=3k\ifif \mu_3^3-9 \mu_3^2+9 \mu_3+12 \sqrt{2}+15=0
\end{equation}
If we pick the largest again root of this equation 
$$\mu_3=\sqrt[3]{6-6 \sqrt{2}+6 i \sqrt{3+2 \sqrt{2}}}+\frac{6^{2/3}}{\sqrt[3]{1-\sqrt{2}+i \sqrt{3+2 \sqrt{2}}}}+3\approx 7.09717$$ 
we can see that indeed condition $0\leq p_3-p_2\leq p_2-p_1$ is respected (it is equivalent to $0\leq k\leq 0.271521$ which holds since $k\leq\frac{1}{r+1}$ and $r\geq 3$), so we have computed that for any $m$
$$
p_3^{(m)}\approx 3-\frac{7.09717}{m+1}\quad\text{and}\quad \mu_3\approx 7.09717 .
$$ 
\item \underline{$r=4$ and any $m$:}
Continuing up the same way, we compute that for every $p_4$ such that $0\leq p_4-p_3\leq p_3-p_2$ it is
\begin{multline*}
v(p_1,p_2,p_3,p_4) =\int_0^{p_4-p_3}v(p_1,p_2,p_3)\,dt+\int_{p_4-p_3}^{p_3-p_2}v(p_1,p_2,p_4-t)\,dt\\+\int_{p_3-p_2}^{p_2-p_1}v(p_1,p_3-t,p_4-t)\,dt
+\int_{p_2-p_1}^{p_1}v(p_2-t,p_3-t,p_4-t)\,dt+\int_{p_1}^1\,dt	
\end{multline*}
which equals
\begin{multline}
\ \hspace{-1.3cm}
\frac{1}{24} \left(4 p_1^4-16 p_1^3 p_4-24 p_1^2 \left(3 p_3^2-6 p_3 p_4+2 p_4^2\right)-16 p_1 \left(3 p_2^3-9 p_2^2 p_4-9 p_2 \left(2 p_3^2-4 p_3 p_4+p_4^2\right)\right.\right.\\
\ \hspace{-1.3cm}
\left.\left.+6 p_3^3-9 p_3^2 p_4+p_4^3\right)+18 p_2^4-48 p_2^3 p_4-36 p_2^2 \left(2 p_3^2-4 p_3 p_4+p_4^2\right)+12 p_3^4-16 p_3^3 p_4+p_4^4+24\right).\label{eq:eq:volsliceSJA4}
\end{multline}
By applying the transformation~\eqref{eq:transfpm}, plugging in the already computed values for $\mu_1=1$ and $\mu_2=2+\sqrt{2}$ and using the fact that $\mu_3$ is the root of Equation~\eqref{eq:mu3defpoly}, we get that equation $v(1-\mu_1k,2-\mu_2k,3-\mu_3k,4-\mu_4k)=4k$ is equivalent to
\begin{equation}\label{eq:mu4defpoly}
\mu_4^4-16 \mu_4^3+24 \mu_4^2+96 \sqrt{2} \mu_4+128 \mu_4+72 \mu_3^2-144 \sqrt{2} \mu_3-288 \mu_3+48 \sqrt{2}+88=0.
\end{equation}
If we pick the largest again root $\mu_4\approx 11.9972$ of this equation we can see that indeed condition $0\leq p_4-p_3\leq p_3-p_2$ is respected (it is equivalent to $0\leq k\leq 0.204082$ which holds since $k\leq\frac{1}{r+1}$ and $r\geq 4$), so we have computed that for any $m$
$$
p_4^{(m)}\approx 4-\frac{11.9972}{m+1}\quad\text{and}\quad \mu_4\approx 11.9972 .
$$ 

\item \underline{$r=5$ and $m=5$:}
At this point we need to modify a little bit our procedure of computing the volumes in the usual recursive way, and consider the case where the new $p_5$ price is such that $p_3\leq p_5\leq p_4$  instead of $p_5\geq p_4$ (and in fact the even stronger condition that $p_5-p_4\leq p_4-p_3$). This is again a straightforward calculation, since as we argued before, $v(p_1,p_2,p_3,p_4,p_5)=v(p_1,p_2,p_3,p_5,p_5)$ and so
\begin{multline*}
v(p_1,p_2,p_3,p_4,p_5) =\int_0^{p_5-p_3}v(p_1,p_2,p_3,p_5-t)\,dt+\int_{p_5-p_3}^{p_3-p_2}v(p_1,p_2,p_5-t,p_5-t)\,dt\\+\int_{p_3-p_2}^{p_2-p_1}v(p_1,p_3-t,p_5-t,p_5-t)\,dt
+\int_{p_2-p_1}^{p_1}v(p_2-t,p_3-t,p_5-t,p_5-t)\,dt+\int_{p_1}^1\,dt	
\end{multline*}
which equals
\begin{multline}
\ \hspace{-1.7cm}
\frac{1}{120} \left(-5 p_1^5+25 p_1^4 p_5-50 p_1^3 p_5^2+50 p_1^2 \left(6 p_3^3-18 p_3^2 p_5+18 p_3 p_5^2-5 p_5^3\right)+25 p_1 \left(4 p_2^4-16 p_2^3 p_5\right.\right.\\ 
\ \hspace{-1.7cm}
\left.\left.+24 p_2^2 p_5^2-16 p_2 \left(3 p_3^3-9 p_3^2 p_5+9 p_3 p_5^2-2 p_5^3\right)+18 p_3^4-48 p_3^3 p_5+36 p_3^2 p_5^2-3 p_5^4\right)-2 \left(20 p_2^5 \right.\right.\\
\ \hspace{-1.7cm}
 \left.\left. -75 p_2^4 p_5+100 p_2^3 p_5^2-50 p_2^2 \left(3 p_3^3-9 p_3^2 p_5+9 p_3 p_5^2-2 p_5^3\right)+30 p_3^5-75 p_3^4 p_5+50 p_3^3 p_5^2-2 p_5^5-60\right)\right)\label{eq:eq:volsliceSJA55}
\end{multline}
By applying the transformation~\eqref{eq:transfpm}, plugging in the already computed values for $\mu_1=1$ and $\mu_2=2+\sqrt{2}$ and using the fact that $\mu_3$ is the root of Equation~\eqref{eq:mu3defpoly}, we get that equation $v(1-\mu_1k,2-\mu_2k,3-\mu_3k,5-\mu_5k,5-\mu_5k)=5k$ is equivalent to
\begin{multline}\label{eq:mu55defpoly}
\ \hspace{-1cm}
4 \mu_5^5-225 \mu_5^4+4350 \mu_5^3+800 \sqrt{2} \mu_5^2-34950 \mu_5^2+900 \mu_5 \mu_3^2-1800 \sqrt{2} \mu_5 \mu_3-3600 \mu_5 \mu_3-14600 \sqrt{2} \mu_5\\+121175 \mu_5+720 \sqrt{2} \mu_3^2-14220 \mu_3^2+22680 \sqrt{2} \mu_3+49680 \mu_3+41080 \sqrt{2}-161215=0
\end{multline}
If we pick the second largest  root $\mu_5\approx 18.0865$ of this equation we can see that indeed condition $p_3\leq p_5\leq p_4$ is respected (it is equivalent to $0.16422\leq k\leq 0.181994$ which holds since $k=\frac{1}{m+1}$ and $m= 5$), so we have computed that for $m=5$
$$
p_5^{(5)}\approx 5-\frac{18.0865}{6}=1.98558\quad\text{and}\quad \mu_5^{(5)}\approx 18.0865.
$$ 

\item \underline{$r=5$ and any $m\geq 6$:}
We can compute that for every $p_5$ such that $0\leq p_5-p_4\leq p_4-p_3$ it is
\begin{align*}
v(p_1,p_2,p_3,p_4,p_5) &=\int_0^{p_5-p_4}v(p_1,p_2,p_3,p_4)\,dt+\int_{p_5-p_4}^{p_4-p_3}v(p_1,p_2,p_3,p_5-t)\,dt\\ +&\int_{p_4-p_3}^{p_3-p_2}v(p_1,p_2,p_4-t,p_5-t)\,dt+\int_{p_3-p_2}^{p_2-p_1}v(p_1,p_3-t,p_4-t,p_5-t)\,dt
\\ +& \int_{p_2-p_1}^{p_1}v(p_2-t,p_3-t,p_4-t,p_5-t)\,dt+\int_{p_1}^1\,dt	
\end{align*}
which equals
\begin{multline}
\frac{1}{120} \left(-5 p_1^5+25 p_1^4 p_5+50 p_1^3 \left(4 p_4^2-8 p_4 p_5+3 p_5^2\right)+50 p_1^2 \left(6 p_3^3-18 p_3^2 p_5-18 p_3 \left(2 p_4^2\right.\right.\right. \\ \left. \left.\left.-4 p_4 p_5+p_5^2\right)+16 p_4^3-24 p_4^2 p_5+3 p_5^3\right)+25 p_1 \left(4 p_2^4-16 p_2^3 p_5-24 p_2^2 \left(3 p_4^2-6 p_4 p_5+2 p_5^2\right)\right.\right. \\ \left. \left.-16 p_2 \left(3 p_3^3-9 p_3^2 p_5-9 p_3 \left(2 p_4^2-4 p_4 p_5+p_5^2\right)+6 p_4^3-9 p_4^2 p_5+p_5^3\right)+18 p_3^4-48 p_3^3 p_5\right.\right. \\ \left. \left. -36 p_3^2 \left(2 p_4^2-4 p_4 p_5+p_5^2\right)+12 p_4^4-16 p_4^3 p_5+p_5^4\right)-40 p_2^5+150 p_2^4 p_5+200 p_2^3 \left(3 p_4^2 \right.\right. \\ \left. \left. -6 p_4 p_5+2 p_5^2\right)+100 p_2^2 \left(3 p_3^3-9 p_3^2 p_5-9 p_3 \left(2 p_4^2-4 p_4 p_5+p_5^2\right)+6 p_4^3-9 p_4^2 p_5+p_5^3\right) \right. \\ \left. -60 p_3^5+150 p_3^4 p_5+200 p_3^3 p_4^2-400 p_3^3 p_4 p_5+100 p_3^3 p_5^2-20 p_4^5+25 p_4^4 p_5-p_5^5+120\right)\label{eq:eq:volsliceSJA5}
\end{multline}
By applying the transformation~\eqref{eq:transfpm}, plugging in the already computed values for $\mu_1=1$ and $\mu_2=2+\sqrt{2}$ and using the fact that $\mu_3$ is the root of Equation~\eqref{eq:mu3defpoly} and $\mu_4$ is the root of~\eqref{eq:mu4defpoly}, we get that equation $v(1-\mu_1k,2-\mu_2k,3-\mu_3k,4-\mu_4k,5-\mu_5 k)=5k$ is equivalent to
\begin{multline}\label{eq:mu5defpoly}
\ \hspace{-1.6cm}
\mu_5^5-25 \mu_5^4+50 \mu_5^3-100 \mu_5^2 \mu_3^3+900 \mu_5^2 \mu_3^2-900 \mu_5^2 \mu_3-800 \sqrt{2} \mu_5^2-950 \mu_5^2-150 \mu_5 \mu_3^4+400 \mu_5 \mu_3^3 \mu_4\\
\ \hspace{-1.6cm}
+1200 \mu_5 \mu_3^3-3600 \mu_5 \mu_3^2 \mu_4-900 \mu_5 \mu_3^2+3600 \mu_5 \mu_3 \mu_4-25 \mu_5 \mu_4^4+400 \mu_5 \mu_4^3-600 \mu_5 \mu_4^2+2400 \sqrt{2} \mu_5 \mu_4\\
\ \hspace{-1.6cm}
+2800 \mu_5 \mu_4-1600 \sqrt{2} \mu_5-2225 \mu_5+60 \mu_3^5-150 \mu_3^4-200 \mu_3^3 \mu_4^2-400 \mu_3^3 \mu_4-1900 \mu_3^3+1800 \mu_3^2 \mu_4^2+3600 \mu_3^2 \mu_4\\
\ \hspace{-1.6cm}
-1800 \mu_3 \mu_4^2-3600 \mu_3 \mu_4+1800 \mu_3+20 \mu_4^5-275 \mu_4^4-1200 \sqrt{2} \mu_4^2-800 \mu_4^2-2400 \sqrt{2} \mu_4-2800 \mu_4+8960 \sqrt{2}\\
\ \hspace{-1.6cm}
+12185=0.
\end{multline}
If we pick the largest again (real) root $\mu_5\approx 18.0843$ of this equation we can see that indeed condition $0\leq p_5-p_4\leq p_4-p_3$ is respected (it is equivalent to $0\leq k\leq 0.16428$ which holds since $k\leq\frac{1}{m+1}$ and $m\geq 6$), so we have computed that for any $m\geq 6$
$$
p_5^{(m)}\approx 5-\frac{18.0843}{m+1}\quad\text{and}\quad \mu_5^{(m)}\approx 18.0843 .
$$ 

\item \underline{$r=6$ and $m=6$:}
If the new $p_6$ price is such that $p_4\leq p_6\leq p_5$, similar to the case of $r=m=5$,  we have that $v(p_1,p_2,p_3,p_4,p_5,p_6)=v(p_1,p_2,p_3,p_4,p_6,p_6)$ and so
\begin{multline*}
\ \hspace{-1cm}
v(p_1,p_2,p_3,p_4,p_5,p_6) =\int_0^{p_6-p_4}v(p_1,p_2,p_3,p_4,p_6-t)\,dt+\int_{p_6-p_4}^{p_4-p_3}v(p_1,p_2,p_3,p_5-t,p_5-t)\,dt\\+\int_{p_4-p_3}^{p_3-p_2}v(p_1,p_2,p_4-t,p_6-t,p_6-t)\,dt+\int_{p_3-p_2}^{p_2-p_1}v(p_1,p_3-t,p_4-t,p_6-t,p_6-t)\,dt\\
+\int_{p_2-p_1}^{p_1}v(p_2-t,p_3-t,p_4-t,p_6-t,p_6-t)\,dt+\int_{p_1}^1\,dt	
\end{multline*}
which equals
\begin{multline}
\ \hspace{-1.7cm}
\frac{1}{720} \left(6 p_1^6-36 p_1^5 p_6+90 p_1^4 p_6^2-120 p_1^3 \left(10 p_4^3-30 p_4^2 p_6+30 p_4 p_6^2-9 p_6^3\right)-90 p_1^2 \left(10 p_3^4-40 p_3^3 p_6\right.\right. \\ 
\ \hspace{-1.7cm}
\left. \left.+60 p_3^2 p_6^2-40 p_3 \left(3 p_4^3-9 p_4^2 p_6+9 p_4 p_6^2-2 p_6^3\right)+60 p_4^4-160 p_4^3 p_6+120 p_4^2 p_6^2-11 p_6^4\right)-36 p_1 \left(5 p_2^5\right.\right. \\ 
\ \hspace{-1.7cm}
\left. \left. -25 p_2^4 p_6+50 p_2^3 p_6^2-50 p_2^2 \left(6 p_4^3-18 p_4^2 p_6+18 p_4 p_6^2-5 p_6^3\right)-25 p_2 \left(4 p_3^4-16 p_3^3 p_6+24 p_3^2 p_6^2\right.\right.\right. \\ 
\ \hspace{-1.7cm}
\left. \left. \left.-16 p_3 \left(3 p_4^3-9 p_4^2 p_6+9 p_4 p_6^2-2 p_6^3\right)+18 p_4^4-48 p_4^3 p_6+36 p_4^2 p_6^2-3 p_6^4\right)+2 \left(20 p_3^5-75 p_3^4 p_6 \right.\right.\right. \\ 
\ \hspace{-1.7cm}
\left. \left. \left.+100 p_3^3 p_6^2-50 p_3^2 \left(3 p_4^3-9 p_4^2 p_6+9 p_4 p_6^2-2 p_6^3\right)+30 p_4^5-75 p_4^4 p_6+50 p_4^3 p_6^2-2 p_6^5\right)\right)+5 \left(15 p_2^6 \right.\right. \\ 
\ \hspace{-1.7cm}
\left. \left. -72 p_2^5 p_6+135 p_2^4 p_6^2-120 p_2^3 \left(6 p_4^3-18 p_4^2 p_6+18 p_4 p_6^2-5 p_6^3\right)-45 p_2^2 \left(4 p_3^4-16 p_3^3 p_6+24 p_3^2 p_6^2- \right.\right.\right. \\ 
\ \hspace{-1.7cm}
\left. \left.\left. 16 p_3 \left(3 p_4^3-9 p_4^2 p_6+9 p_4 p_6^2-2 p_6^3\right)+18 p_4^4-48 p_4^3 p_6+36 p_4^2 p_6^2-3 p_6^4\right)+40 p_3^6-144 p_3^5 p_6+180 p_3^4 p_6^2\right.\right. \\ 
\ \hspace{-1.7cm}
\left. \left.-80 p_3^3 \left(3 p_4^3-9 p_4^2 p_6+9 p_4 p_6^2-2 p_6^3\right)+30 p_4^6-72 p_4^5 p_6+45 p_4^4 p_6^2-p_6^6+144\right)\right)\label{eq:eq:volsliceSJA66}
\end{multline}

By applying the transformation~\eqref{eq:transfpm}, plugging in the already computed values for $\mu_1=1$ and $\mu_2=2+\sqrt{2}$ and using the fact that $\mu_3$ is the root of Equation~\eqref{eq:mu3defpoly} and $\mu_4$ is the root of~\eqref{eq:mu4defpoly} , we get that equation $v(1-\mu_1k,2-\mu_2k,3-\mu_3k,4-\mu_4 k,6-\mu_6k,6-\mu_6k)=6k$ is equivalent to
\begin{multline}\label{eq:mu66defpoly}
\ \hspace{-1.7cm}
\mu_6^6-36 \mu_6^5+270 \mu_6^4-160 \mu_6^3 \mu_3^3+1440 \mu_6^3 \mu_3^2-1440 \mu_6^3 \mu_3-1200 \sqrt{2} \mu_6^3-2160 \mu_6^3-180 \mu_6^2 \mu_3^4+720 \mu_6^2 \mu_3^3 \mu_4\\
\ \hspace{-1.7cm}
+2160 \mu_6^2 \mu_3^3-6480 \mu_6^2 \mu_3^2 \mu_4-7560 \mu_6^2 \mu_3^2+6480 \mu_6^2 \mu_3 \mu_4+6480 \mu_6^2 \mu_3-45 \mu_6^2 \mu_4^4+720 \mu_6^2 \mu_4^3-1080 \mu_6^2 \mu_4^2\\
\ \hspace{-1.7cm}
+4320 \sqrt{2} \mu_6^2 \mu_4+5040 \mu_6^2 \mu_4+4320 \sqrt{2} \mu_6^2+5760 \mu_6^2+144 \mu_6 \mu_3^5-720 \mu_6 \mu_3^3 \mu_4^2-2880 \mu_6 \mu_3^3 \mu_4-8640 \mu_6 \mu_3^3\\
\ \hspace{-1.7cm}
+6480 \mu_6 \mu_3^2 \mu_4^2+25920 \mu_6 \mu_3^2 \mu_4+12960 \mu_6 \mu_3^2-6480 \mu_6 \mu_3 \mu_4^2-25920 \mu_6 \mu_3 \mu_4-6480 \mu_6 \mu_3+72 \mu_6 \mu_4^5\\
\ \hspace{-1.7cm}
-900 \mu_6 \mu_4^4-1440 \mu_6 \mu_4^3-4320 \sqrt{2} \mu_6 \mu_4^2-720 \mu_6 \mu_4^2-17280 \sqrt{2} \mu_6 \mu_4-20160 \mu_6 \mu_4+8928 \sqrt{2} \mu_6+13104 \mu_6\\
\ \hspace{-1.7cm}
-40 \mu_3^6-144 \mu_3^5+1440 \mu_3^4+240 \mu_3^3 \mu_4^3+1440 \mu_3^3 \mu_4^2+2880 \mu_3^3 \mu_4+10560 \mu_3^3-2160 \mu_3^2 \mu_4^3-12960 \mu_3^2 \mu_4^2\\
\ \hspace{-1.7cm}
-25920 \mu_3^2 \mu_4-7560 \mu_3^2+2160 \mu_3 \mu_4^3+12960 \mu_3 \mu_4^2+25920 \mu_3 \mu_4-2160 \mu_3-30 \mu_4^6+288 \mu_4^5+1440 \mu_4^4\\
\ \hspace{-1.7cm}
+1440 \sqrt{2} \mu_4^3+1680 \mu_4^3+8640 \sqrt{2} \mu_4^2+5760 \mu_4^2+17280 \sqrt{2} \mu_4+20160 \mu_4-42048 \sqrt{2}-58344=0
\end{multline}
If we pick the second largest  root $\mu_6\approx 25.3585$ of this equation we can see that indeed condition $p_4\leq p_6\leq p_5$ is respected (it is equivalent to $0.137473\leq k\leq 0.149686$ which holds since $k=\frac{1}{m+1}$ and $m= 6$), so we have computed that for $m=6$
$$
p_6^{(6)}\approx 6-\frac{25.3585}{7}=2.37736\quad\text{and}\quad \mu_6^{(6)}\approx 25.3585.
$$ 
\end{itemize}

\end{document}